\documentclass[a4paper,UKenglish,cleveref, autoref, thm-restate]{lipics-v2021}
\nolinenumbers

\title{Meaningfulness and Genericity in a Subsuming Framework} 

\author{Delia Kesner}
    {Université Paris Cité - CNRS - IRIF, France}
    {}
    {0000-0003-4254-3129}
    {}

    \author{Victor Arrial}
    {Université Paris Cité - CNRS - IRIF, France}
    {}
    {0000-0002-0469-4279}
    {}

\author{Giulio Guerrieri}
    {University of Sussex, Department of Informatics, Brighton, United Kingdom}
    {g.guerrieri@sussex.ac.uk}
    {0000-0002-1607-7403}
    {}

\authorrunning{D. Kesner, V. Arrial, G. Guerrieri} 

\Copyright{Delia Kesner, Victor Arrial,  and Giulio Guerrieri} 



\category{} 

\relatedversion{} 

\EventEditors{John Q. Open and Joan R. Access}
\EventNoEds{2}
\EventLongTitle{42nd Conference on Very Important Topics (CVIT 2016)}
\EventShortTitle{CVIT 2016}
\EventAcronym{CVIT}
\EventYear{2016}
\EventDate{December 24--27, 2016}
\EventLocation{Little Whinging, United Kingdom}
\EventLogo{}
\SeriesVolume{42}
\ArticleNo{23}

\usepackage{S_Packages}
\usepackage{S_Comments}
\usepackage{S_Macro}
\usepackage{S_Macro_Bang}
\usepackage{S_Macro_CBN}
\usepackage{S_Macro_CBV}

\NewDocumentCommand{\CBNSymb}{}{{\tt CBN}\xspace}
\NewDocumentCommand{\CBVSymb}{}{{\tt CBV}\xspace}
\NewDocumentCommand{\CBPVSymb}{}{{\tt CBPV}\xspace}
\NewDocumentCommand{\BANGSymb}{}{{\tt BANG}\xspace}

\begin{document}

\maketitle

\begin{abstract}
    This paper studies the notion of meaningfulness for a unifying
framework called \bangCalculusSymb-calculus, which subsumes both
call-by-name (\cbnCalculusSymb) and call-by-value (\cbvCalculusSymb).
We first characterize meaningfulness in \bangCalculusSymb by means of
typability and inhabitation in an associated non-idempotent
intersection type system previously proposed in the literature. We
validate the proposed notion of meaningfulness by showing two
properties (1) consistency of the theory $\mathcal{H}$ equating
meaningless terms and (2) genericity, stating that meaningless
subterms have no bearing on the significance of meaningful terms. The
theory $\mathcal{H}$ is also shown to have a unique consistent and
maximal extension. Last but not least, we show that the notions of
meaningfulness and genericity in the literature for \cbnCalculusSymb
and \cbvCalculusSymb are subsumed by the respectively ones proposed
here for the \bangCalculusSymb-calculus.

\end{abstract}


\section{Introduction}
\label{sec:Introduction}

A common line of research in logic and theoretical computer science is
to find unifying frameworks that subsume different paradigms, systems
or calculi. Examples are call-by-push-value~\cite{Levy99,Levy04},
polarized system LU~\cite{Girard93}, linear
calculi~\cite{MaraistOderskyTurnerWadler95,MaraistOderskyTurnerWadler99,RonchiRoversi97},
bang-calculus~\cite{Ehrhard16,EhrhardGuerrieri16,BucciarelliKesnerRiosViso20,BucciarelliKesnerRiosViso23},
system L~\cite{Munch09,CurienMunch10}, ecumenical
systems~\cite{Prawitz15}, monadic calculus~\cite{Moggi89,Moggi91}, and
others~\cite{RoccaP04,EspiritoPintoUustalu19,BakelTyeWu23}.

The relevance of these unifying frameworks lies in the range of
properties and models they encompass. Finding \emph{unifying  and
simple primitives, tools and techniques} to reason about properties of
different systems is challenging, and provides
a deeper and more
abstract understanding of these properties. The advantages of this
kind of approach are numerous, for instance the \emph{several-for-one
deal}: study a property in a unifying framework gives
appropriate intuitions and hints  for free
for all the subsumed systems. The aim of this paper is to go beyond
the state of the art in a framework subsumig the \emph{call-by-name}
and \emph{call-by-value} evaluation mechanisms, by unifying the
notions of \emph{meaningful} (and \emph{meaningless}) programs. 

\smallskip
\noindent\textbf{Call-by-name and call-by-value.} Every programming
language implements a particular evaluation strategy, specifying when
and how parameters are evaluated during function calls. For example,
in call-by-value (\CBVSymb), the argument is evaluated before being
passed to the function, while in call-by-name (\CBNSymb) the argument
is passed immediately to the function body, so that it may never be
evaluated, or may be re-evaluated several times. These models of
computation serve as the basis for many theoretical and practical
studies in programming languages and proof assistants, such as OCaml,
Haskell, Coq, Isabelle, etc.

The \CBNSymb strategy has garnered significant attention in the
literature and is generally perceived as well-established. In
contrast, the  \CBVSymb strategy has received limited attention.
Indeed, despite their similarities, \CBNSymb and \CBVSymb strategies
have predominantly been studied independently, leading to a fragmented
research. This approach not only duplicates research efforts—once for
\CBNSymb and again for \CBVSymb— but also generally results in ad-hoc
methods for dealing with the \CBVSymb case that are naively adapted
from the  \CBNSymb one.

Understanding the (logical) duality between \CBNSymb and \CBVSymb~(\eg
\cite{CurienHerbelin00}) marked a significant step towards properly
unifying these models. It paved the way for the emergence of
Call-by-Push-Value (\CBPVSymb), a unifying framework introduced by
P.B. Levy~\cite{Levy99} which \emphit{subsumes}, among others,
\CBNSymb and \CBVSymb\ denotational and operational semantics thanks
to the distinction between \emph{computations} and \emph{values},
according to the slogan ``a value is, a computation does''. This
framework attracts growing attention: proving advanced properties of a
single \emph{unifying paradigm}, and subsequently instantiate them for
a \emph{wide range} of computational models.

\smallskip
\noindent\textbf{The distant \BANGSymb-calculus.}
Drawing inspiration from Girard's Linear Logic (LL)~\cite{Girard87}
and the interpretation~\cite{Ehrhard16} of \CBPVSymb into LL, Ehrhard
and Guerrieri \cite{EhrhardGuerrieri16} introduced an (untyped)
restriction of \CBPVSymb, named \BANGSymb-calculus, already capable of
subsuming both \CBNSymb and \CBVSymb. It is obtained by enriching the
$\lambda$-calculus with two modalities $\oc$ and its dual $\derSymb$.
The modality $\oc$ actually plays a twofold role: it freezes the
evaluation of subterms (called \emph{thunk} in \CBPVSymb), and it
marks what can be duplicated or erased during evaluation (\ie copied
an arbitrary number of times, including zero). The modality $\derSymb$
annihilates the effect of $\oc$, effectively restoring computation and
eliminating duplicability. Embedding \CBNSymb or \CBVSymb into the
\BANGSymb-calculus simply consists in decorating $\lambda$-terms with
$\oc$ and $\derSymb$, thereby forcing one model of computation or the
other one. Thanks to these elementary modalities and embeddings, the
Bang Calculus eases the identification of shared behaviors and
properties of \CBNSymb and \CBVSymb, encompassing both syntactic and
semantic aspects of them.

The original \BANGSymb-calculus~\cite{EhrhardGuerrieri16} uses
some permutation rules, similar
to the ones used in \cite{Regnier94,CarraroGuerrieri14}, that unveil
hidden redexes and unblock reductions that otherwise would be stuck.
These permutation rules make the calculus \emph{adequate},
preventing that some normal forms are observationally equivalent to
non-terminating terms. A major
drawback is that the resulting combined reduction is not
confluent (Page 6 in~\cite{EhrhardGuerrieri16}). As an alternative,
the \emph{distant Bang calculus} (\bangCalculusSymb)~\cite{BucciarelliKesnerRiosViso20,BucciarelliKesnerRiosViso23}
is both adequate and confluent.  This is achieved by
enriching the syntax with \emph{explicit substitutions}, in the vein of
Accattoli and Kesner's linear substitution
calculus~\cite{AccattoliKesner10,AccattoliKesner12bis,Accattoli12,AccattoliBonelliKesnerLombardi14}
(generalizing in turn Milner's calculus~\cite{Milner2006,KOCcorr}),
thanks to rewrite rules that act \emph{at a distance}, so that permutation
rules are no longer needed. 

In this paper, we focus on \bangCalculusSymb, and its relations with
\cbnCalculusSymb \cite{AccattoliKesner12bis,Accattoli12} and
\cbvCalculusSymb \cite{AccattoliPaolini12}, which are \emph{distant}
adequate variants of the \CBNSymb and \CBVSymb-calculi. This unifying
framework has proven fruitful, subsuming numerous \cbnCalculusSymb and
\cbvCalculusSymb properties through their associated embedding, as for
instance big step semantics: evaluating the result from the
\cbnCalculusSymb/\cbvCalculusSymb embedding of a given program $t$
with the \bangCalculusSymb model actually corresponds to the embedding
of the result of evaluating the original program $t$ with the
\cbnCalculusSymb/\cbvCalculusSymb model. In other words,
\bangCalculusSymb is a language that breaks down the \cbnCalculusSymb
and \cbvCalculusSymb paradigms into elementary primitives.

Let us now review the state of the art by discussing some advanced
properties of programming languages that have been studied in the
literature by using the unifying approach \bangCalculusSymb.
Some of these results, including this work, strongly rely on
semantical tools such as quantitative types. To ensure clarity
regarding the state of the art, let us briefly discuss in first place
the main ideas behind quantitative types.

\smallskip
\noindent\textbf{Quantitative Type Systems.} \emphit{Intersection type
systems}~\cite{CDC78,CoppoDezani:NDJoFL-80} increase the typability
power of simply typed $\lambda$-terms by introducing a new
\emphit{intersection} type constructor $\wedge$, which is, in
principle, associative, commutative and \emphit{idempotent} ({\ie}
$\sigma \wedge \sigma = \sigma$). Intersection types allow terms to
have different types simultaneously, \eg a term $t$ has type $\sigma
\wedge \tau$ whenever $t$ has both the type $\sigma$ and the type
$\tau$. They turns out to constitute a powerful tool to reason about
\emphit{qualitative} properties of programs. For example, different
notions of normalization can be characterized using intersection
types~\cite{Pottinger80,CoppoDezaniVenneri81}, in the sense that a
term $t$ is typable in a given system if and only if $t$ is
normalizing (as a consequence, typability in these systems is
undecidable). Removing idempotence \cite{Gardner94,deCarvalho07} gives
rise to \emphit{non-idempotent} type systems for the
$\lambda$-calculus where a term of type $\sigma \land \sigma \land
\tau$ can  be seen as a resource that is used exactly once as a data
of type $\tau$ and twice as a data of type $\sigma$.  Interestingly,
such type systems do not only provide qualitative characterizations of
different operational properties, but also \emphit{quantitative} ones:
\eg a term $t$ is still typable if and only if $t$ is normalizing, but
additionally, any type derivation of $t$ gives an \emphit{upper bound}
to the execution time for $t$ (the number of steps to reach a normal
form)~\cite{deCarvalho18}. These upper bounds can be further refined
into \emphit{exact measure} through the use of \emphit{tight
non-idempotent typing systems}, as pioneered by~\cite{AccattoliGK20}.

\smallskip
\noindent\textbf{State of the Art.}  
This paper contributes to a broader initiative aimed at consolidating
the theory of \cbnCalculusSymb and \cbvCalculusSymb, by unifying them
into \bangCalculusSymb. Several results have already been factorized
and generalized in this framework, we now revisit some of them.

In \cite{GuerrieriManzonetto18} it is shown that the interpretation of
a term $t$ in any denotational model of \CBNSymb/\CBVSymb obtained
from LL is included in  the interpretation of the \CBNSymb/\CBVSymb
translation of $t$ in any denotational model of \BANGSymb obtained
from LL. The reverse inclusion also holds for  \CBNSymb but not for
\CBVSymb. In particular, these results apply to \emph{typability} in a
non-idempotent intersection type system inspired by LL. Indeed, typing
is preserved by Girard's translations, meaning that if a term is
typable in the \CBNSymb/\CBVSymb type system, then its
\CBNSymb/\CBVSymb translation is typable in the type system
$\bangBKRVTypeSys$ for \BANGSymb, using the same types.  The converse also holds for \CBNSymb but not for \CBVSymb. In
\cite{BucciarelliKesnerRiosViso20,BucciarelliKesnerRiosViso23}, the
\CBVSymb typing system is modified so that the reverse implication also
holds. Moreover an extension of Girard's \CBNSymb translation to
\cbnCalculusSymb and a \emph{new} \CBVSymb translation to
\cbvCalculusSymb are proposed. Similar typing preservation results
have been obtained in~\cite{KesnerViso22} for the translations
in~\cite{BucciarelliKesnerRiosViso20,BucciarelliKesnerRiosViso23}, but
for the more precise notion of tight typing
introduced~in~\cite{AccattoliGK20}.

Retrieving \emph{dynamic} properties from \BANGSymb into \CBNSymb and
\CBVSymb turns out to be a more intricate task, especially in their
\emph{adequate} (distant)
variant~\cite{BucciarelliKesnerRiosViso20,FaggianGuerrieri21,BucciarelliKesnerRiosViso23}.

In \cite{GuerrieriManzonetto18} it is also shown that \CBNSymb and
\CBVSymb can be simulated by \emphit{reduction} in \BANGSymb through
Girard's original translations. But the \CBVSymb translation fails to
preserve \emphit{normal forms}, as some \CBVSymb normal forms
translate to reducible terms in \BANGSymb. This issue is solved in
\bangCalculusSymb~\cite{BucciarelliKesnerRiosViso20,BucciarelliKesnerRiosViso23},
thanks to the \emph{new} \CBVSymb translation to \cbvCalculusSymb
mentioned before. In the end, reductions \emph{and} normal forms are
preserved by both the \CBNSymb and the new \CBVSymb translations.

Even if \cbnCalculusSymb and \cbvCalculusSymb can be both simulated by
\emphit{reduction} in \bangCalculusSymb, the converse, known as
\emph{reverse simulation}, fails: a \bangCalculusSymb reduction
sequence from a term in the image of the \cbvCalculusSymb embedding
may not correspond to a valid reduction sequence in \cbvCalculusSymb.
Yet another new \cbvCalculusSymb translation is proposed in
\cite{ArrialGuerrieriKesner24bis} so that simulation and reverse
simulation are now verified.

Another major contribution concerns the \emphit{inhabitation} problem:
given an environment $\Gamma$ (a type assignment for variables) and a
type $\sigma$, decide whether there is a term $t$ that can be typed
with $\sigma$ under the environment $\Gamma$. While inhabitation was
shown~\cite{Urzyczyn99} to be \emphit{undecidable} in \CBNSymb for
idempotent intersection type systems, it turns out to be
\emphit{decidable}~\cite{bkdlr14,BucciarelliKR18} in the
non-idempotent setting. Decidability of the inhabitation problem leads
to the development of automatic tools for type-based \emphit{program
synthesis}~\cite{MannaWaldinger80,BCDDLR17}, whose goal is to
construct a program ---the term $t$--- that
satisfies~some~high-level~formal~specification ---expressed as a type
$\sigma$ with some assumptions described by the environment $\Gamma$.
It has been proved in~\cite{ArrialGuerrieriKesner23} that the
algorithms deciding the inhabitation problem for \cbnCalculusSymb and
\cbvCalculusSymb can be inferred from the corresponding one for
\bangCalculusSymb, thus providing a unified solution to this
relevant~problem.

\smallskip
\noindent\textbf{Meaningfulness and Genericity.} In this work, we aim
to unify the notions of  meaningfulness and genericity  in
\cbnCalculusSymb and \cbvCalculusSymb so as to derive them from the
respective ones~in~\bangCalculusSymb.

A naive approach to set a semantics for the pure untyped
$\lambda$-calculus is to define the meaning of a $\beta$-normalizing
$\lambda$-term as its normal form, and equating all $\lambda$-terms
that do not $\beta$-normalize. The underlying idea is that, as
$\beta$-reduction represents evaluation and a normal form stands for
its outcome, all non-$\beta$-normalizing $\lambda$-terms are
considered as meaningless. However, this simplistic approach is
flawed, as thoroughly discussed in \cite{barendregt84nh}. For example, 
any $\lambda$-theory equating all non-$\beta$-normalizing $\lambda$-terms
is inherently inconsistent ---it effectively equates all
$\lambda$-terms, not just the meaningless ones!

Alternatively, during the 70s, Wadsworth
\cite{Wadsworth71,Wadsworth76} and Barendregt
\cite{Barendregt71,Barendregt73,Barendregt75,barendregt84nh} showed
that the meaningful (\CBNSymb) $\lambda$-terms can be identified with
the \emph{solvable} ones. Solvability is defined in a rather technical
way: a $\lambda$-term $t$ is \emph{solvable} if there is a special
kind of context, called \emph{head} context ${\tt H}$, sending $t$ to
the identity function $\Id = \lambda z.z$, meaning that ${\tt
H}\cbnCtxtPlug{t}$ $\beta$-reduces to $\Id$. Roughly, a solvable
$\lambda$-term $t$ may be divergent, but its diverging subterms can be
eliminated by supplying the right  arguments to $t$ via an appropriate
interaction with a suitable head context ${\tt H}$. For instance, in
$\CBNSymb$, $x\Omega$ is divergent but solvable using the head context
${\tt H} = \app{(\abs{x}{\Hole})}{(\abs{y}{\Id})}$. It turns out that
\emph{unsolvable} $\lambda$-terms constitutes a strict subset of the
non-$\beta$-normalizing ones. 
Moreover, the smallest $\lambda$-theory that equates all unsolvable
$\lambda$-terms is \emph{consistent} (\ie it does not equate all
terms).  In Barendregt's book \cite{barendregt84nh}, theses results
relies on a keystone property known as \emphit{genericity}, which
states that meaningless subterms are computationally irrelevant —in
the sense that they do not play any role— in the evaluation of
$\beta$-normalizing terms. Formally, if $t$ is \emph{unsolvable} and
${\tt C}\cbnCtxtPlug{t}$ $\beta$-reduces to \emph{some}
$\beta$-normal term $u$ for some context ${\tt C}$, then ${\tt
C}\cbnCtxtPlug{s}$ $\beta$-reduces to $u$ for \emph{every}
$\lambda$-term~$s$. This property stands as a fool guard that the
choice of meaningfulness is adequate.

Meaningfulness was also studied for first order rewriting
systems~\cite{KennawayOV99} and other strategies of the
$\lambda$-calculus~\cite{RoccaP04}. Notably, finding the correct
notion of meaningfulness for \CBVSymb has been a
challenge~\cite{AccattoliGuerrieri22,ArrialGuerrieriKesner24}.
Similarly, a notable extension of the \cbnCalculusSymb was
studied~\cite{BucciarelliKesnerRonchi21} in the framework of a
$\lambda$-calculus equipped with pattern matching for pairs. The use
of different data structures in the language ---functions and pairs---
makes the meaningfulness problem more challenging. Indeed, it was
shown that meaningfulness cannot be characterized only by means of
typability alone, as in \CBNSymb and \CBVSymb, but also requires  some
additional conditions stated in terms of the inhabitation problem
mentioned before. This result for the $\lambda$-calculus with patterns
inspired the characterization of meaningfulness that we provide in
this paper. Genericity for \cbnCalculusSymb and the more subtle case
of \cbvCalculusSymb was recently proved
in~\cite{ArrialGuerrieriKesner24}.

\smallskip
\noindent\textbf{Our Contributions.} We first define meaningfulness
for \bangCalculusSymb, and provide a characterization by means of
typability \emph{and} inhabitation. As a second contribution, we
validate this notion of meaningfulness twofold: meaningless terms
enjoy genericity, and the theory $\mathcal{H}$ obtained by equating
all the meaningless terms is consistent. Moreover, we show that
$\mathcal{H}$ admits a unique maximal consistent extension. Last but
not least, as a third contribution, we show that the notions of
meaningfulness in the literature for \cbnCalculusSymb and
\cbvCalculusSymb are subsumed by the one proposed here for
\bangCalculusSymb. We also obtained genericity for \cbnCalculusSymb
and \cbvCalculusSymb as a consequence of the genericity property for
\bangCalculusSymb.


\smallskip
\noindent\textbf{Roadmap.} \Cref{sec:dbang} recalls \bangCalculusSymb
and its quantitative type system $\bangBKRVTypeSys$.
\Cref{sec:Meaningfulness} defines meaningfulness for
\bangCalculusSymb, and characterizes it in terms of typability and
inhabitation in the type system $\bangBKRVTypeSys$.
\Cref{sec:typed-genericity} addresses genericity, while
\Cref{sec:subsuming} establishes a precise relationship between
meaningless and genericity in \cbnCalculusSymb/\cbvCalculusSymb and
their corresponding notions in \bangCalculusSymb. \Cref{s:conclusion}
discusses future and related work and concludes.

\section{The \bangCalculusSymb-Calculus}
\label{sec:dbang}

\subsection{Syntax and Operational Semantics}
\label{subsec:lambda!}

We introduce here the term syntax of the \bangCalculusSymb-calculus
\cite{BucciarelliKesnerRiosViso20,BucciarelliKesnerRiosViso23}.
Given a countably infinite set $\bangSetVariables$ of variables $x, y,
z, \dots$, the set of terms $\bangSetTerms$ is given by the following
inductive definition:
\begin{equation*}
	\textbf{(Terms)} \;\;t, u, s \;\coloneqq\;
            x \in \bangSetVariables
        \vsep \app[\,]{t}{u}
	    \vsep \abs{x}{t}
        \vsep \oc t
        \vsep \der{t}
        \vsep t\esub{x}{u}
\end{equation*}

The set $\bangSetTerms$ includes $\lambda$-terms (\defn{variables}
$x$, \defn{abstractions} $\abs{x}{t}$ and \defn{applications} $tu$) as
well as three additional constructors: a \defn{closure}
$t\esub{x}{u}$ representing a pending  \defn{explicit substitution
(ES)} $\esub{x}{u}$ on a term $t$, a \defn{bang} $\oc t$ to freeze the
execution of $t$, and a \defn{dereliction} $\der{t}$ to fire again the
frozen term $t$. The \defn{argument} of an application
$\app[\,]{t}{u}$ (resp. a closure $t\esub{x}{u}$) is the subterm $u$.
From now on, we set $I \coloneqq \abs{z}{z}$, $\Delta \coloneqq
\abs{x}{\app{x}{\oc x}}$, and $\Omega \coloneqq
\app{\Delta}{\oc\Delta}$.

Abstractions $\abs{x}{t}$ and closures $t\esub{x}{u}$ bind the
variable $x$ in the term $t$. \defn{Free} and
\defn{bound} variables are defined as expected, in particular
$\freeVar{\abs{x}{t}} \coloneqq \freeVar{t} \setminus \{x\}$ and
$\freeVar{t\esub{x}{u}} \coloneqq \freeVar{u} \cup (\freeVar{t}
\setminus \{x\})$. The usual notion of $\alpha$-conversion
\cite{barendregt84nh} is extended to $\bangSetTerms$,
and terms are identified up to $\alpha$-conversion. We denote by
$t\isub{x}{u}$ the usual (capture avoiding) meta-level substitution of
the term $u$ for all free occurrences of the variable $x$ in the
term~$t$.

The set of \defn{\listTxt contexts} $(\bangLCtxt)$, \defn{\surfaceTxt
contexts} $(\bangStratCtxt)$ and \defn{\fullTxt contexts}
$(\bangFCtxt)$, can be seen as terms containing exactly one
\defn{hole} $\Hole$, they are inductively defined as follows:
\begin{equation*}
	\begin{array}{rcl}
    \bangLCtxt &\Coloneqq& \Hole \vsep
		\bangLCtxt\esub{x}{t}
\\
	\bangStratCtxt &\Coloneqq& \Hole
		\vsep \app[\,]{\bangStratCtxt}{t}
		\vsep \app[\,]{t}{\bangStratCtxt}
		\vsep \abs{x}{\bangStratCtxt}
		\vsep \der{\bangStratCtxt}
		\vsep \bangStratCtxt\esub{x}{t}
		\vsep t\esub{x}{\bangStratCtxt}
\\
    \bangFCtxt &\Coloneqq& \Hole
        \vsep \app[\,]{\bangFCtxt}{t}
        \vsep \app[\,]{t}{\bangFCtxt}
        \vsep \abs{x}{\bangFCtxt}
        \vsep \der{\bangFCtxt}
        \vsep \bangFCtxt\esub{x}{t}
        \vsep t\esub{x}{\bangFCtxt}
        \vsep \oc\bangFCtxt
	\end{array}
\end{equation*}
\listTxt^ contexts and \surfaceTxt contexts are special cases of
\fullTxt contexts. The hole can occur everywhere in \fullTxt contexts,
while it is forbidden in \surfaceTxt contexts under a $\oc$. For
example, $\app[\,]{y}{(\abs{x}{\Hole})}$ is a \surfaceTxt context
hence a full context, while $(\oc\Hole)\esub{x}{\Id}$ is a \fullTxt
context but not a \surfaceTxt one. We write $\bangFCtxt<t>$ for the
term obtained by replacing the hole in $\bangFCtxt$ by the term $t$.

The following \defn{rewrite rules} are the base components of the
reductions of \bangCalculusSymb. Any term having the shape of the
left-hand side of one of these three rules is called a \defn{redex}.
\begin{equation*}
	\app{\bangLCtxt<\abs{x}{t}>}{u}
		\mapstoR[\bangSymbBeta]
	\bangLCtxt<t\esub{x}{u}>
\qquad
	t\esub{x}{\bangLCtxt<\oc u>}
		\mapstoR[\bangSymbSubs]
	\bangLCtxt<t\isub{x}{u}>
\qquad
	\der{\bangLCtxt<\oc t>}
		\mapstoR[\bangSymbSubs]
	\bangLCtxt<t>
\end{equation*}

Rule $\bangSymbBeta$ (\resp $\bangSymbSubs$) is assumed to be capture
free: no free variable of $u$ (resp. $t$) is captured by the list
context $\bangLCtxt$. The rule $\bangSymbBeta$ fires a standard
$\beta$-redex and generates an ES. The rule $\bangSymbSubs$ operates a
substitution provided its argument is a bang. The rule $\bangSymbBang$
opens a bang. All these rewrite rules act \emphit{at a
distance}~\cite{AccattoliKesner10,AccattoliKesner12bis,AccattoliBonelliKesnerLombardi14}:
the main constructors involved in the rule can be separated by a
finite---possibly empty---list context $\bangLCtxt$ of ES. This
mechanism unblocks redexes that otherwise would be stuck, \eg
$\app{(\abs{x}{x})\esub{y}{w}}{\oc z} \!\mapstoR[\bangSymbBeta]\!
x\esub{x}{\oc z}\esub{y}{w}$ fires a $\beta$-redex where $\bangLCtxt =
\Hole\esub{y}{w}$ is the list context in between the function
$\abs{x}{x}$ and the argument $\oc z$.

The \emphasis{\surfaceTxt reduction} $\bangArr_S$ is the \surfaceTxt
closure of the three rewrite rules $\bangSymbBeta$, $\bangSymbSubs$
and $\bangSymbBang$, \ie\ $\bangArr_S$ only fires redexes in
\surfaceTxt contexts (not under bang). Similarly, the
\emphasis{\fullTxt reduction} $\bangArr_F$ is the \fullTxt closure of
the three rewrite rules, so that $\bangArr_F$ reduces under full
contexts and thus the bang looses its freezing behavior. For example,
\begin{equation*}
	\app{(\abs{x}{\oc\der{\oc x}})}{\oc y}
		\;\bangArr_S\;
	(\oc\der{\oc x})\esub{x}{\oc y}
		\;\bangArr_S\;
	\oc (\der{\oc y})
		\;\bangArr_F\;
	\oc y
\end{equation*}

The first two steps are also $\bangArr_F$-steps, the last one is not a
$\bangArr_S$-step. We denote by $\bangArr*_S$ the reflexive-transitive
closure of $\bangArr_S$, and similarly for $\bangArr_F$. 
	A reduction $\bangArr_{\R}$ is
  \emphasis{confluent} if for all $t, u_1, u_2$ such that $t\bangArr^*_{\R}u_1$ and
  $t\bangArr^*_{\R}u_2$, there is $s$ such that $u_1\bangArr^*_{\R} s$ and
  $u_2\bangArr^*_{\R} s$.

\begin{restatable}{theorem}{RecBangSurfaceFullConfluence}
    \LemmaToFromProof{Bang_Surface_Full_Confluence}%
    The reductions $\bangArr_S$ and $\bangArr_F$ are
    confluent.
\end{restatable}

A term $t$ is a \defn{\surfaceTxt} (\resp \defn{\fullTxt})
\defn{normal form} if there is no $u$ such that $t \bangArr_S u$
(\resp $t \bangArr_F u$). A term $t$ is \defn{\surfaceTxt} (\resp
\defn{\fullTxt}) \defn{normalizing} if $t \bangArr*_S u$ (\resp $t \bangArr*_F u$) for some
\surfaceTxt (\resp \fullTxt) normal form $u$. Since $\bangArr_S
\subsetneq \bangArr_F$, some terms may be \surfaceTxt-normalizing
but not \fullTxt-normalizing, \eg $\abs{x}{\oc
(\der{\oc\Omega})}$.
\par
As a matter of fact, some ill-formed terms are not redexes but neither
represent a desired computation result. They are called \defn{clashes}
and have one of the following forms:
\begin{equation*}
	\app{\bangLCtxt<\oc s>}{u}
\qquad
	s\esub{x}{\bangLCtxt<\abs{x}{u}>}
\qquad
	\der{\bangLCtxt<\abs{x}{u}>}
\qquad
	\app{t}{(\bangLCtxt<\abs{x}{u}>)}
	\text{ if } t \neq \bangLCtxt'<\abs{y}{s}>
\end{equation*}
This \emphit{static} notion of clash is lifted to a \emphit{dynamic}
level. A term $t$ is a \defn{\surfaceTxt} (\resp \defn{\fullTxt})
\defn{clash-free} if it does not \surfaceTxt (\resp \fullTxt) reduce
to a term with a clash in \surfaceTxt (\resp~\fullTxt) position, \ie
if there are no \surfaceTxt (\resp \fullTxt) context $\bangStratCtxt$
(\resp $\bangFCtxt$) and clash $c$ such that $t \bangArr*_S
\bangStratCtxt<c>$ (\resp $t \bangArr*_F \bangFCtxt<c>$). For example,
$\app{x}{\oc (\app{y}{(\abs{z}{z})})}$ is \surfaceTxt clash-free but
not \fullTxt clash-free as it has a clash $y(\abs{z}{z})$ under a
bang. Both notions \mbox{are stable under reduction}.

Finally, some terms contain neither redexes nor clashes. A
\defn{\surfaceTxt} (\resp \defn{\fullTxt}) \defn{clash-free normal
form} is a \surfaceTxt (\resp \fullTxt) normal form which is also
\surfaceTxt (\resp \fullTxt) clash-free, as \eg the term $\app{x}{x}$.
These are the desired results of the computation, and they can even be
\emph{syntactically} characterized by a tree grammar. 
\begin{align*}
		\bangNe_S &\coloneqq x \in \bangSetVariables
			\vsep \app{\bangNe_S}{\bangNa_S}
			\vsep \der{\bangNe_S}
			\vsep \bangNe_S\esub{x}{\bangNe_S}
	&
		\bangNa_S &\coloneqq \oc t
			\vsep \bangNe_S
			\vsep \bangNa_S\esub{x}{\bangNe_S}
	\\
		\bangNb_S &\coloneqq \bangNe_S
			\vsep \abs{x}{\bangNo_S}
			\vsep \bangNb_S\esub{x}{\bangNe_S}
	&
		\bangNo_S &\coloneqq \bangNa_S
			\vsep \bangNb_S
\end{align*}

\begin{lemma}[\cite{BucciarelliKesnerRiosViso20}]%
	Let $t \in \bangSetTerms$, then $t \in \bangNo_S$ iff $t$ is a
	\surfaceTxt clash-free normal form.
\end{lemma}

\subsection{Quantitative Typing System}
\label{subsec:Bang_Typing_System}%

We present the quantitative typing system
$\bangBKRVTypeSys$~\cite{BucciarelliKesnerRiosViso20}, based
on~\cite{Gardner94,deCarvalho07}. It contains functional and
intersection types. Here, intersections
are associative, commutative but \emphit{not idempotent}, thus an
intersection type is represented by a (possibly empty) \emphit{finite}
\emphit{multiset} $\typeMulti{\sigma_i}_{i\in I}$. Formally, given a
countably infinite set $\mathcal{TV}$ of type variables $\alpha,
\beta, \gamma, \dots$, we define by mutual induction:
\begin{equation*}
	\begin{array}{rrcl}
		(\emphasis{Types}) & \sigma, \tau, \rho &\coloneqq& \alpha \in \mathcal{TV}
				\vsep \M
				\vsep \M\typeArrow\sigma
	\\
		(\emphasis{Multitypes}) & \M, \N &\coloneqq& \typeMulti{\sigma_i}_{i\in I} \text{ where $I$ is a finite set}
	\end{array}
\end{equation*}

A (\emphasis{type}) \emphasis{environment}, noted $\Gamma$ or
$\Delta$, is a function from variables to multitypes, assigning the
\emphasis{empty multitype} $\emul$ to all variables except a finite
number (possibly zero). The \emphasis{empty environment}, noted
$\emptyset$, maps every variable to $\emul$. The \emphasis{domain} of
$\Gamma$ is $\typeCtxtDom{\Gamma} = \{x \in \bangSetVariables \mid
\Gamma(x) \neq \emptymset\}$, the \emphasis{image} of $\Gamma$ is
$\typeCtxtIm{\Gamma} = \{\Gamma(x) \mid x \in \typeCtxtDom{\Gamma}\}$.
Given the environments $\Gamma$ and $\Delta$, $\Gamma + \Delta$ is the
environment mapping $x$ to $\Gamma(x) \uplus \Delta(x)$, where
$\uplus$ denotes multiset union; and $+_{i\in I} \Delta_i$ (with $I$
finite) is its obvious extension to the non-binary case, in particular
$+_{i\in I} \Delta_i = \emptyset$ if $I = \emptyset$. An environment
$\Gamma$ is denoted by $x_1 \!:\! \M_1, \dots, x_n \!:\! \M_n$ when the $x_i$'s are pairwise
distinct variables and 
$\Gamma(x_i) = \M_i$ for all $1 \leq i \leq n$,
and \mbox{$\Gamma(y) = \emptymset$ for $y
\notin \{x_1, \dots, x_n\}$}.

A \emphasis{typing} is a pair $\typing{\Gamma}{\sigma}$, where
$\Gamma$ is an  environment and  $\sigma$ is a type. A
(\emphasis{typing}) \emphasis{judgment} is a tuple of the form $\Gamma
\vdash t : \sigma$, where $\typing{\Gamma}{\sigma}$ is a typing and
$t$ is a term (the \emphasis{subject} of the judgment).
The typing system $\bangBKRVTypeSys$ for \bangCalculusSymb is defined
by the rules in \Cref{fig:BangBKRV_Typing_System_Rules}. The axiom
rule \bangBKRVVarRuleName is relevant, \ie there is no weakening.
Rules \bangBKRVAbsRuleName, \bangBKRVAppRuleName and
\bangBKRVEsRuleName are standard. Rule \bangBKRVBgRuleName has as many
premises as elements in the finite (possibly empty) index set $I$,
the conclusion types $\oc u$ with a multitype
\emphit{gathering} all the (possibly different) types in the premises
typing $u$. In particular, when $I = \emptyset$, the rule has no
premises, and it types \emph{any} term $\oc u$ with
$\emptymset$, leaving the \emph{subterm} $u$ \emph{untyped}. Rule
\bangBKRVDerRuleName forces the argument of a dereliction to be typed
by a multitype of cardinality 1.

\begin{figure}[t]
        $\begin{array}{c}
            \begin{array}{c} \\[-0.1cm]
                \begin{prooftree}
                    \inferBangBKRVVar{x : \mset{\sigma} \vdash x : \sigma}
                \end{prooftree}
            \end{array}
        \qquad
            \begin{prooftree}
                \hypo{\Gamma \vdash t : \M \typeArrow \sigma}
                \hypo{\Delta \vdash u : \M}
                \inferBangBKRVApp{\Gamma + \Delta \vdash \app[\,]{t}{u} : \sigma}
            \end{prooftree}
        \qquad
            \begin{prooftree}
                \hypo{(\Gamma_i \vdash t : \sigma_i)_{i \in I}}
                \inferBangBKRVBg{+_{i \in I} \Gamma_i \vdash \oc t : \mset{\sigma_i}_{i \in I}}
            \end{prooftree}
    \\[16pt]
            \begin{prooftree}
                \hypo{\Gamma, x : \M \vdash t : \sigma}
                \inferBangBKRVAbs{\Gamma \vdash \abs{x}{t} : \M \typeArrow \sigma}
            \end{prooftree}
        \qquad
            \begin{prooftree}
                \hypo{\Gamma, x : \M \vdash t : \sigma}
                \hypo{\Delta \vdash u : \M}
                \inferBangBKRVEs{\Gamma + \Delta \vdash t\esub{x}{u} : \sigma}
            \end{prooftree}
        \qquad
            \begin{prooftree}
                \hypo{\Gamma \vdash t : \mset{\sigma}}
                \inferBangBKRVDer{\Gamma \vdash \der{t} : \sigma}
            \end{prooftree}
        \end{array}$
    \caption{Type System $\bangBKRVTypeSys$ for the \bangCalculusSymb-calculus.}
    \label{fig:BangBKRV_Typing_System_Rules}%
\end{figure}

A (\emphasis{type}) \emphasis{derivation} in system $\bangBKRVTypeSys$
is a tree obtained by applying the rules in
\Cref{fig:BangBKRV_Typing_System_Rules}. The judgment at the root of
the type derivation is the \emphasis{conclusion} of the derivation. A
term $t$ is $\bangBKRVTypeSys$\emphasis{-typable} if there is $\Pi
\bangBKRVTr \Gamma \vdash t : \sigma$ for some typing
$\typing{\Gamma}{\sigma}$, meaning that
$\Pi$ is a derivation in system $\bangBKRVTypeSys$ with conclusion $
\Gamma \vdash t : \sigma$. We write $\bangBKRVTr\; \Gamma \vdash t :
\sigma$ if there exists $\Pi \bangBKRVTr \Gamma \vdash t : \sigma$.

System $\bangBKRVTypeSys$ characterizes \surfaceTxt-normalizing
 clash-free terms.

\begin{restatable}[\cite{BucciarelliKesnerRiosViso20,ArrialGuerrieriKesner23}]{theorem}{RecSRCharact}
	\label{lem:Bang_BKRV}%
	Let $t, u \in \bangSetTerms$.
	\begin{enumerate}
	\item \label{lem:Bang_BKRV_Surface_Typing_SRE}%
		If $t \bangArr_F u$, then for any typing
		$\typing{\Gamma}{\sigma}$, one has $\bangBKRVTr\; \Gamma
		\vdash t : \sigma$ if and only if $\bangBKRVTr\; \Gamma \vdash
		u : \sigma$.
	
  	\item\label{lem:Bang_BKRV_Characterization}%
		$t$ is $\bangBKRVTypeSys$-typable if and only if $t$
		\surfaceTxt-reduces to a \surfaceTxt clash-free normal form.
	\end{enumerate} 
\end{restatable}

\section{Meaningfulness = Typability + Inhabitation}
\label{sec:Meaningfulness}%

In this section we introduce the notion of meaningfulness for
\bangCalculusSymb and we establish a logical characterization of
meaningful terms via system $\bangBKRVTypeSys$. Intuitively, a term
$t$ is meaningful if it can be supplied by some arguments (possibly
binding some free variables of $t$) so that it reduces to some
observable term. In \bangCalculusSymb, the observables are the bang
terms since they are the only terms enabling substitution to be fired.

\begin{definition}
    A term $t$ is \bangCalculusSymb-\defn{meaningful} if there are a testing context $\bangTCtxt$
    and $u \in \bangSetTerms$ such that $\bangTCtxt<t>
    \bangArr*_S \oc u$, where testing contexts are defined by the
    grammar $\bangTCtxt \;\coloneqq\; \Hole \vsep \app{\bangTCtxt}{s} \vsep
    \app{(\abs{x}{\bangTCtxt})}{s}$.
\end{definition}

For example, $\Id$ is \bangCalculusSymb-meaningful using the testing
context $\bangTCtxt = \app[\,]{\Hole}{\oc\oc u}$. Both $\Omega$ and
$\app{x}{\Omega}$ are \bangCalculusSymb-meaningless: every testing
context they are plugged in cannot erase~$\Omega$, which is not
normalizing. Note that the hole in a testing context is always in the
functional position of an application, in particular if the hole is in
the scope of some $\lambda$, then this $\lambda$ must~be~applied.

\begin{figure}
    \hspace{-0.5cm}
    \begin{minipage}[c]{0.49\textwidth}%
        \vspace{0.5cm}
        \begin{equation*}
            \begin{prooftree}[separation=1em]
                \inferBangBKRVVar{x : \mset{\M \typeArrow \sigma} \vdash x : \M \typeArrow \sigma}
                \inferBangBKRVVar{x : \mset{\M} \vdash x : \M}
                \inferBangBKRVApp{x : \mset{\M \typeArrow \sigma, \M} \vdash \app{x}{x} : \sigma}
            \end{prooftree}
        \end{equation*}
    \caption{A type derivation of $\app{x}{x}$ in~\mbox{system $\bangBKRVTypeSys$.}}
    \label{fig:Example_Typing_Derivation_Meaningless}%
    \end{minipage}
~\hspace{0.5cm}~
    \begin{minipage}[c]{0.49\textwidth}%
        \begin{equation*}
            \hspace{2cm}
            \begin{prooftree}
                \inferBangBKRVVar{x : \mset{\alpha} \vdash x : \alpha}
                \inferBangBKRVBg{x : \mset{\alpha} \vdash \oc x : \mset{\alpha}}
                \inferBangBKRVAbs{\emptyset \vdash \abs{x}{\oc x} : \mset{\alpha} \typeArrow \mset{\alpha}}
            \end{prooftree}
        \end{equation*}
    \caption{\mbox{Inhabitation of $\mset{\alpha} \typeArrow \mset{\alpha}$ in system $\bangBKRVTypeSys$.}}
    \label{fig:Example_Inhabitation}%
    \end{minipage}
\end{figure}

In an adequate calculus, meaningfulness is usually characterized both
operationally (normalizability) and logically (typability): a term is
meaningful iff it is normalizing for a suitable subreduction of the
calculus iff it is typable in a suitable type system.  Surprisingly,
these characterizations are subtler in \bangCalculusSymb, because the
language has two (incompatible) data structures: abstractions (playing
the role of functions) and bangs (playing as values).

We are not aware of any operational characterization of
\bangCalculusSymb-meaningfulness; a natural candidate would be
normalizability by surface reduction, but this fails, even if we force
the obtained surface normal form to be clash-free. For instance, the
term $t_0 \coloneqq \app{x}{x}$ is \bangCalculusSymb-meaningless
despite being a surface clash-free normal form. Indeed, for
$\app{x}{x}$ to be \bangCalculusSymb-meaningful, a testing context
$\bangTCtxt$ would need to provide a term $u$ to substitute the
variable $x$, so that $\bangTCtxt<\app{x}{x}>$ would eventually reduce
to a bang. However, achieving this requires the mentioned term $u$ to
reduce to both an abstraction and a bang, which is impossible. Hence,
\bangCalculusSymb-meaningfulness is not only the ability to produce a
surface clash-free normal form but also to transform this result into
an observable.

Concerning a logical characterization of
\bangCalculusSymb-meaningfulness, typability is not enough, at least
in system $\bangBKRVTypeSys$. For instance, the
\bangCalculusSymb-meaningless term $t_0 \coloneqq \app{x}{x}$ seen
above is typable in system $\bangBKRVTypeSys$. Every
type derivation of $t_0$ has the form in
\Cref{fig:Example_Typing_Derivation_Meaningless}, which reveals the
conflict when assigning to $x$ both an arrow type $\M \typeArrow
\sigma$ (the type of terms eventually reducing to
abstractions) and a multitype $\M$ (the type of terms
eventually reducing to bangs). The inhabitation problem can
be used to detect such conflicts, allowing for a handy
characterization of meaningfulness.  Indeed, the multitype $\mset{\M
\typeArrow \sigma, \M}$ assigned to the variable $x$ in
\Cref{fig:Example_Typing_Derivation_Meaningless} is not inhabited.

While it is complex to syntactically establish operational conditions
such as (not) reducing to abstractions or bangs, this is easily
achieved semantically. Indeed, we establish a logical characterization
of \bangCalculusSymb-meaningfulness based on \emph{typability} and
\emph{inhabitation} in system $\bangBKRVTypeSys$, similarly to what
happens in the $\lambda$-calculus with
pairs~\cite{AKV19,BucciarelliKR15,BucciarelliKesnerRonchi21}.
Intuitively, suppose that a term $t$ is \bangCalculusSymb-meaningful,
so there is a testing context $\bangTCtxt$ such that $\bangTCtxt<t>$
reduces to an observable, \ie~a bang, which can be typed with the
typing $\typing{\emptyset}{\emptymset}$ in system $\bangBKRVTypeSys$.
By~\Cref{lem:Bang_BKRV}.\ref{lem:Bang_BKRV_Surface_Typing_SRE},
$\bangTCtxt<t>$ must also be typable by the same typing
$(\emptyset; \emptymset)$, meaning that $t$ is \emph{typable} by an
environment $x_1 \!:\! \M_1, \dots, x_n
\!:\! \M_m$ and a type $\N_1 \typeArrow \ldots \typeArrow \N_n
\typeArrow \emptymset$, where all the $\M_i$'s and
$\N_i$'s~are~\emph{inhabited}.

A similar argument holds for other type systems and
calculi~\cite{BucciarelliKesnerRonchi21} with their own notions of
meaningfulness and observable. The point is to identify the set of
types $\typeobs{\genericsystem}$ associated with the observables.
Given any type system $\genericsystem$, whose types are those of
\Cref{subsec:Bang_Typing_System}, and a set of types
$\typeobs{\genericsystem}$ for observable terms, the set of
\defn{arguments} $\genericArgs{\genericsystem}{\sigma}$ of a type
$\sigma$ is the set of multitypes appearing to the left of arrows,
until reaching the type of an observable. Formally, if $\sigma \in
\typeobs{\genericsystem}$ then $\genericArgs{\genericsystem}{\sigma}
\coloneqq \emptyset$, \emph{otherwise}
$\genericArgs{\genericsystem}{\alpha}\coloneqq \emptyset$,
$\genericArgs{\genericsystem}{\M \typeArrow \sigma} \coloneqq \{\M\}
\cup \genericArgs{\genericsystem}{\sigma}$, and
$\genericArgs{\genericsystem}{\M} = \emptyset$. For example, in system
$\bangBKRVTypeSys$ where $\typeobs{\bangBKRVTypeSys}$ contains all
multitypes, $\genericArgs{\bangBKRVTypeSys}{[\tau] \typeArrow (\M
\typeArrow \mset{\alpha})} = \{\mset{\tau}, \M\}$. The
\cbnCalculusSymb and
\cbnCalculusSymb~cases~are~discussed~in~\Cref{sec:subsuming}.

\NewDocumentCommand{\genericInhPred}{ m }					{{\tt inh}_{\genericsystem}\!\left(#1\right)}				
\NewDocumentCommand{\genericTr}{ }
                   {\tr_{\genericsystem}}
\NewDocumentCommand{\genTypeArgs}				{ e{_} m }
	{{\tt args}_{\genericsystem}\!\IfValueT{#1}{_{#1}}\left(#2\right)}

\begin{definition} 
    Let $\genericsystem$ be a type system and $\genericInhPred{\cdot}$
    be a predicate on the types of $\genericsystem$. A set $S$ of
    types is \defn{inhabited}, noted $\genericInhPred{S}$, if
    $\genericInhPred{\sigma}$ for all $\sigma \in S$. We write
    $\genericInhPred{\Gamma}$ if
    $\genericInhPred{\typeCtxtIm{\Gamma}}$. A typing
    $\typing{\Gamma}{\sigma}$ (\resp~judgment $\Gamma \vdash t :
    \sigma$) is $\genericsystem$-\defn{\testableTxt} if
    $\genericInhPred{\Gamma}$ and
    $\genericInhPred{\genTypeArgs{\sigma}}$. A term $t$ is
    \defn{$\genericsystem$-\testableTxt}\ if $\genericTr\; \Gamma
    \vdash t : \sigma$ for some $\genericsystem$-\testableTxt typing
    $\typing{\Gamma}{\sigma}$.
\end{definition}

A type $\sigma$ is \defn{inhabited} in system $\bangBKRVTypeSys$,
noted $\bangBKRVInhPred{\sigma}$, if $\Pi \bangBKRVTr \emptyset \vdash
t : \sigma$ for some $\Pi$ and $t$. For instance, in system
$\bangBKRVTypeSys$, the type $\emul$ is inhabited using rule
$\bangBKRVBgRuleName$ with empty premises, and the
environment $\emptyset$ is trivially
inhabited. The type $\mset{\alpha} \typeArrow \mset{\alpha}$ is also
inhabited, see \Cref{fig:Example_Inhabitation}.

\begin{restatable}{lemma}{RecBangInhabitationFromTesting}
    \LemmaToFromProof{Bang_O_|-_T<t>_:_[]_==>_Gam_|-_t_:_sig_and_Gam_and_args(sig)_inh}%
    Let $t \in \bangSetTerms$ and $\bangTCtxt$ be a testing context.
    If $\bangBKRVTr\; \emptyset \vdash \bangTCtxt<t> : \emptymset$,
    then $\bangBKRVTr \Gamma \vdash t : \sigma$ with
    $\bangBKRVInhPred{\Gamma}$ and
    $\bangBKRVInhPred{\bangTypeArgs{\sigma}}$.
\end{restatable}

Inhabitation serves as a crucial tool to produce an observable from a
typable term. As discussed previously, any type of a variable $x$ from
the environment of a meaningful term $t$ should be inhabited. Hence,
we need the type environment $\Gamma$ to be inhabited. However,
relying solely on the inhabitation of $\Gamma$ is not sufficient, as
illustrated by the typable term $\bangBKRVTr\; \emptyset \vdash
\abs{x}{\app{x}{x}} : \mset{\mset{\M} \typeArrow \tau, \M} \typeArrow
\tau$, which, despite having a trivially inhabited environment, is
\bangCalculusSymb-meaningless. We thus also test the inhabitation of type arguments of
the given type $\sigma$. This therefore means that
$\bangBKRVTypeSys$-testability is sufficient to ensure \bangCalculusSymb-meaningfulness.
Surprisingly, this actually provides a characterization of
\bangCalculusSymb-meaningfulness.

\begin{restatable}[Logical Characterization]{theorem}{RecBangMeaningfulTypabilityInhabitation}
    \LemmaToFromProof{Bang_Meaningfulness_Typing_and_Inhabitation}%
    \label{t:logical-characterization}
    Let $t \in \bangSetTerms$. Then $t$ is
    \bangCalculusSymb-meaningful if and only if $\,\bangBKRVTr \Gamma \vdash t : \sigma$ for some
    $\bangBKRVTypeSys$-testable typing $\typing{\Gamma}{\sigma}$.
\end{restatable}

A \emph{\bangTheory-theory} is an equivalence $\equiv$ on
$\bangSetTerms$ containing $\bangArr_F$ and closed under full
contexts. Let $\bangTheoryH$ (also noted $\equiv_{\bangTheoryH}$) be
the smallest \bangTheory-theory equating all
\bangCalculusSymb-meaningless terms. \Cref{t:logical-characterization}
entails that $\bangTheoryH$ is \defn{consistent}, that is, it does not
equate~all~terms.

\begin{restatable}[Consistency of $\bangTheoryH$]{proposition}{RecConsistencyThH}
    \LemmaToFromProof{Bang_Theory_H_Consistent}%
    \label{prop:consistency-H}
    There exist $t, u \in \bangSetTerms$ such that $t
    \not\equiv_{\bangTheoryH} u$.
\end{restatable}
We also corroborate our definition of meaningfulness by proving that
it fulfills a genericity property and showing that $\mathcal{H}$
admits a unique maximal consistent extension
(\Cref{sec:typed-genericity}) and subsumes the well-established
notions of meaningfulness for \cbnCalculusSymb and \cbvCalculusSymb
(\Cref{sec:subsuming}).

\section{Typed and Surface Genericity in \bangCalculusSymb}
\label{sec:typed-genericity} 

In \Cref{sec:Meaningfulness}, we proved that
$\bangCalculusSymb$-meaningfulness is captured by typability in system
$\bangBKRVTypeSys$ with some $\bangBKRVTypeSys$-testable typing.
While this concise characterization as ``meaningfulness = typability +
inhabitation''~\cite{BucciarelliKesnerRonchi21} provides a high level
understanding, its practical manipulation might pose some challenges.
Suppose we study some properties of a $\bangCalculusSymb$-meaningful
term $t$ through the logical characterization
(\Cref{lem:Bang_Meaningfulness_Typing_and_Inhabitation}), thus having
a type derivation $\Pi \bangBKRVTr \Gamma \vdash t: \sigma$ with
$\typing{\Gamma}{\sigma}$ $\bangBKRVTypeSys$-testable. If we proceed
by induction on $\Pi$, then there is no guarantee that all the
judgments appearing in $\Pi$ have $\bangBKRVTypeSys$-testable typings
as well, which would make the reasoning awkward and the logical
characterization of
\Cref{lem:Bang_Meaningfulness_Typing_and_Inhabitation} difficult to
exploit. But this is not the case. Upcoming
\Cref{lem:Bang_BKRV_and_Inh_Equiv_AGKMean} states that
$\bangBKRVTypeSys$-testability propagates bottom-up: if the conclusion
of a derivation $\Pi$ has a $\bangBKRVTypeSys$-testable typing, then
so does every other~judgment~in~$\Pi$.

We write $\Pi \bangAGKMeanTr \Gamma \vdash t : \sigma$ if $\Pi
\bangBKRVTr \Gamma \vdash t : \sigma$ and each judgment in
$\Pi$ is
$\bangBKRVTypeSys$-testable, and $\Pi \bangAGKMeanTr t$ if
$\Pi \bangAGKMeanTr \Gamma \vdash t : \sigma$ for some typing
$\typing{\Gamma}{\sigma}$.

\begin{restatable}{lemma}{RecBangBKRVEquivAGKMean}
  \LemmaToFromProof{Bang_BKRV_and_Inh_Equiv_AGKMean}%
  Let $t \in \bangSetTerms$: $\Pi \bangBKRVTr \Gamma \vdash t :
  \sigma$ with
  $\typing{\Gamma}{\sigma}$ $\bangBKRVTypeSys$-testable iff $\Pi
  \bangAGKMeanTr \Gamma \vdash t : \sigma$.
\end{restatable}
\begin{proof}
    ($\Leftarrow$): Trivial.
    ($\Rightarrow$): By an induction on $\Pi$.
\end{proof}

We can therefore easily reason using the
logical characterization of meaningfulness, resulting in the following first genericity result for
\bangCalculusSymb.

\begin{restatable}[Typed Genericity]{theorem}{RecBangMeaningfulTypeMonotonicity}
  \LemmaToFromProof{Bang_AGKMean_Typed_Genericity}%
  \label{t:typed-genericity}
    Let $t \in \bangSetTerms$ be \bangCalculusSymb-meaningless and
      $\bangFCtxt$ be a full context. If $\bangAGKMeanTr\; \Gamma
      \vdash \bangFCtxt<t> : \sigma$, then $\bangAGKMeanTr\; \Gamma
      \vdash \bangFCtxt<u> : \sigma$ for all $u \in \bangSetTerms$.
    \end{restatable}
\begin{proof}
By induction on $\bangFCtxt$, using both
\Cref{lem:Bang_Meaningfulness_Typing_and_Inhabitation,lem:Bang_BKRV_and_Inh_Equiv_AGKMean}.
\end{proof}

This proof relies entirely on the fact that the meaningless subterm
$t$ cannot be explicitly typed in any of the judgments of $\Pi$, as
typing $t$ in $\bangAGKMeanTypeSys$ is equivalent to being
\bangCalculusSymb-meaningful (by
\Cref{t:logical-characterization,lem:Bang_BKRV_and_Inh_Equiv_AGKMean}).
Thus, the typed genericity theorem fails when weakening the hypothesis
from  $\bangAGKMeanTypeSys$-typability to
$\bangBKRVTypeSys$-typability. Consider for example the
\bangCalculusSymb-meaningless term $t = \app{x}{x}$ and the context
$\bangFCtxt = \app[\,]{y}{\Hole}$, then $\bangFCtxt<t>$ is
$\bangBKRVTypeSys$-typable as witnessed by $\bangBKRVTr\; y : \mset{\N
\typeArrow \alpha}, x : \mset{\M \typeArrow \N, \M} \vdash
\bangFCtxt<t> : \alpha$ --note that the type of $x$ is not inhabited--
while $\bangFCtxt<\Omega>$ is~not~$\bangBKRVTypeSys$-typable.

As a consequence, we can prove a  qualitative
(surface) genericity result. We call it \emph{surface},
despite it universally quantify over full contexts, as
$\bangCalculusSymb$-meaningful is defined in terms of
surface reduction. The corresponding results for \cbnCalculusSymb and
\cbvCalculusSymb are also called \emph{surface} in~\cite{ArrialGuerrieriKesner24} and \emph{light} in~\cite{AccattoliLancelot24}, later generalized to a \emph{stratified}
notion in~\cite{ArrialGuerrieriKesner24}.

\begin{restatable}[Qualitative Surface Genericity]{corollary}{RecBangQualitativeSurfaceGenericity}
    \label{lem:Bang_Qualitative_Surface_Genericity}%
    Let $\bangFCtxt$ be a full context. If $\bangFCtxt<t>$ is
    \bangCalculusSymb-meaningful for some \bangCalculusSymb-meaningless $t \in
    \bangSetTerms$, then $\bangFCtxt<u>$ is
    \bangCalculusSymb-meaningful \mbox{for all $u \in \bangSetTerms$}.
\end{restatable}
\begin{proof}
Let $u \!\in\! \bangSetTerms$. As $\bangFCtxt<t>$ is
\bangCalculusSymb-meaningful, then there is $\Pi \bangAGKMeanTr
\bangFCtxt<t>$ by~\Cref{t:logical-characterization}
and~\Cref{lem:Bang_BKRV_and_Inh_Equiv_AGKMean}. As $t$ is
\bangCalculusSymb-meaningless, then there is $\Pi' \bangAGKMeanTr
\bangFCtxt<u>$ by~\Cref{lem:Bang_AGKMean_Typed_Genericity}, and hence
$\bangFCtxt<u>$ is \bangCalculusSymb-meaningful
by~\Cref{t:logical-characterization}
and~\Cref{lem:Bang_BKRV_and_Inh_Equiv_AGKMean}.
\end{proof}

Genericity is a sanity check on meaningfulness: it holds only if all
\bangCalculusSymb-meaningless terms are \emph{truly} meaningless.
Still, some truly meaningless terms might be misinterpreted as
\bangCalculusSymb-meaningful. Indeed, when crafting a notion of
\bangCalculusSymb-meaningless that would satisfy genericity, one might
not take \emph{all} truly meaningless terms. 
The \bangTheory-theory $\bangTheoryH*$ is introduced to avoid that.
Let $\bangTheoryH*$, also noted $\equiv_{\bangTheoryH*}$, be the relation on $\bangSetTerms$ defined by:
\begin{equation*}
  \bangTheoryH* \coloneqq \{(t,u) \vsep \forall \; \bangFCtxt \text{ full context},\, \bangFCtxt<t> \text{\bangCalculusSymb-meaningful} \Leftrightarrow \bangFCtxt<u> \text{\bangCalculusSymb-meaningful}\}
\end{equation*}
We expect $\bangTheoryH*$ to be an extension of the theory
$\bangTheoryH$. Moreover, to check that all truly meaningless terms
are actually \bangCalculusSymb-meaningless, we also want this theory
to \defn{maximal}, meaning that no term can be additional equated
without compromising consistency.

\begin{restatable}{theorem}{RecMaximalityConsistency}
	\LemmaToFromProof{Bang_Theory_H*_Consistent}%
  	$\bangTheoryH*$ is the unique maximal consistent \bangTheory-theory
  	containing $\bangTheoryH$.
\end{restatable}


\section{Subsuming CBN and CBV Meaningfulness}
\label{sec:subsuming}

In this section we show that the notions of meaningfulness for
\cbnCalculusSymb and \cbvCalculusSymb in the
literature~\cite{ArrialGuerrieriKesner24} are subsumed by the one
proposed in~\Cref{sec:Meaningfulness} for \bangCalculusSymb. We also
deduce surface genericity for \cbnCalculusSymb and \cbvCalculusSymb as
a consequence of surface genericity for \bangCalculusSymb.

\subsection{\cbnCalculusSymb and \cbvCalculusSymb Calculi}

Both \cbnCalculusSymb~\cite{AccattoliKesner10,AK12,Accattoli12} and
\cbvCalculusSymb~\cite{AccattoliPaolini12} are specified  using ES and
action at a distance, as explained in~\Cref{subsec:lambda!} for
\bangCalculusSymb, and they share the same term syntax. The sets
$\cbvSetTerms$ of \defn{terms} and $\cbvSetValues$ of \defn{values}
are inductively defined below.
\begin{equation*}
    \textbf{(Terms)} \ \ 
        t, u \;\Coloneqq\; v
            \vsep \app[\,]{t}{u}
            \vsep t\esub{x}{u}
\hspace{1cm}
    \textbf{(Values)} \ \
        v \;\Coloneqq\; x 
            \vsep \abs{x}{t}
\end{equation*}
Note that the syntax  contains neither $\derSymb$~nor~$\oc$. The
distinction between terms and values is irrelevant in \cbnCalculusSymb
but crucial in \cbvCalculusSymb. The two calculi also share the same
\defn{\listTxt contexts} $\cbnLCtxt$, but use specialized
\defn{\surfaceTxt contexts} $\cbnStratCtxt$ and $\cbvStratCtxt$
respectively. Again, contexts can be seen as terms with exactly one
\defn{hole} $\Hole$ and are inductively defined below.
\begin{equation*}
    \begin{array}{r rcl}
        \defn{(List Contexts)}&
            \cbnLCtxt &\Coloneqq& \Hole
                \vsep \cbnLCtxt\esub{x}{t}
    \\[-2pt]
        \defn{(\cbnCalculusSymb Surface Contexts)}&
            \cbnStratCtxt &\Coloneqq&  \Hole
                \vsep \app[\,]{\cbnStratCtxt}{t}
                \vsep \abs{x}{\cbnStratCtxt}
                \vsep \cbnStratCtxt\esub{x}{t}
    \\[-2pt]
        \defn{(\cbvCalculusSymb Surface Contexts)}&
            \cbvStratCtxt &\Coloneqq&  \Hole
                \vsep \app[\,]{\cbvStratCtxt}{t}
                \vsep \app[\,]{t}{\cbvStratCtxt}
                \vsep \cbvStratCtxt\esub{x}{t}
                \vsep t\esub{x}{\cbvStratCtxt}
    \\[-2pt]
        \defn{(Full Contexts)}&
            \cbnFCtxt &\Coloneqq& \Hole
                \vsep \app[\,]{\cbnFCtxt}{t}
                \vsep \app[\,]{t}{\cbnFCtxt}
                \vsep \abs{x}{\cbnFCtxt}
                \vsep \cbnFCtxt\esub{x}{t}
                \vsep t\esub{x}{\cbnFCtxt}
    \end{array}
\end{equation*}
The differences between \cbnCalculusSymb and \cbvCalculusSymb are in
the previous notions of \emph{surface} contexts, as well as the
following \emph{rewrite rules}.
\begin{align*}
        \app{\cbnLCtxt<\abs{x}{t}>}{u}
            &\mapstoR[\cbnSymbBeta]
        \cbnLCtxt<t\esub{x}{u}>
    &
        t\esub{x}{u}
            &\mapstoR[\cbnSymbSubs]
        t\isub{x}{u}
    &
        t\esub{x}{\cbvLCtxt<v>}
            &\mapstoR[\cbvSymbSubs]
        \cbvLCtxt<t\isub{x}{v}>
\end{align*}

Rules \cbnSymbBeta and \cbvSymbSubs are
both capture-free: no free variable of $u$ (\resp $t$) is captured by
the list context $\cbnLCtxt$. 
The \defn{\cbnCalculusSymb \surfaceTxt reduction} $\cbnArr_S$ is
defined as the union of the \cbnCalculusSymb \surfaceTxt closure of
rewrite rules \cbnSymbBeta and \cbnSymbSubs, while the
\defn{\cbvCalculusSymb \surfaceTxt reduction} $\cbvArr_S$ is the union
of the \cbvCalculusSymb \surfaceTxt closure of the rewrite rules
\cbvSymbBeta and \cbvSymbSubs. 
Finally, we use $\cbnArr*_S$ (\resp
$\cbvArr*_S$) to denote the reflexive transitive closure of the
relation $\cbnArr_S$ (\resp $\cbvArr_S$).
\begin{example}
    \label{example:reduction_cbn_and_cbv}%
    For example, $t_0 \coloneqq
    \app{(\abs{x}{\app{\app{y}{x}}{x}})}{(\app{\Id}{\Id})} \cbnArr_S
    (\app{\app{y}{x}}{x})\esub{x}{\app{\Id}{\Id}} \cbnArr_S
    \app{\app{y}{(\app{\Id}{\Id})}}{(\app{\Id}{\Id})} \eqqcolon t_1$ and $t_0
    = \app{(\abs{x}{\app{\app{y}{x}}{x}})}{(\app{\Id}{\Id})} \cbvArr_S
    (\app{\app{y}{x}}{x})\esub{x}{\app{\Id}{\Id}} \cbvArr_S
    (\app{\app{y}{x}}{x})\esub{x}{z\esub{z}{\Id}} \cbvArr_S
    (\app{\app{y}{x}}{x})\esub{x}{\Id} \cbvArr_S
    \app{\app{y}{\Id}}{\Id} \eqqcolon t_2$.
\end{example}
Note that the {\cbnCalculusSymb surface reduction} is nothing but (a
non-deterministic  diamond variant of) the well-known \emph{head}
reduction~\cite{barendregt84nh}.

\smallskip
The quantitative type systems $\cbnBKRVTypeSys$ for \cbnCalculusSymb and $\cbvBKRVTypeSys$ for \cbvCalculusSymb\
are presented in \cref{fig:cbnBKRV_Typing_System_Rules,fig:cbvBKRV_Typing_System_Rules},  respectively, they were studied in
\cite{BucciarelliKesnerRiosViso20,BucciarelliKesnerRiosViso23}.
\emphasis{Types} and \emphasis{judgments} are the same as for system
$\bangBKRVTypeSys$.  A derivation $\Pi$ in system $\cbnBKRVTypeSys$ with
conclusion $\Gamma \vdash t : \sigma$ is noted $\Pi \cbnBKRVTr \Gamma
\vdash t : \sigma$; we write $\cbnBKRVTr\; \Gamma \vdash t : \sigma$ if there is a $\Pi \cbnBKRVTr \Gamma
\vdash t : \sigma$. We use similar notations
for system $\cbvBKRVTypeSys$.

\begin{figure}[t]
        \begin{equation*}
            \begin{array}{c}
                \begin{array}{c} \\[-0.1cm]
                    \begin{prooftree}
                        \inferBangBKRVVar{x : \mset{\sigma} \vdash x : \sigma}
                    \end{prooftree}
                \end{array}
            \qquad
                \begin{prooftree}
                    \hypo{\Gamma \vdash t : \mset{\tau_i}_{i \in I} \typeArrow \sigma}
                    \hypo{(\Delta_i \vdash u : \tau_i)_{i \in I}}
                    \inferBangBKRVApp{\Gamma +_{i \in I} \Delta_i \vdash \app[\,]{t}{u} : \sigma}
                \end{prooftree}
        \\[16pt]
                \begin{prooftree}
                    \hypo{\Gamma, x : \M \vdash t : \sigma}
                    \inferBangBKRVAbs{\Gamma \vdash \abs{x}{t} : \M \typeArrow \sigma}
                \end{prooftree}
            \qquad
                \begin{prooftree}
                    \hypo{\Gamma, x : \mset{\tau_i}_{i \in I} \vdash t : \sigma}
                    \hypo{(\Delta_i \vdash u : \tau_i)_{i \in I}}
                    \inferBangBKRVEs{\Gamma +_{i \in I} \Delta_i \vdash t\esub{x}{u} : \sigma}
                \end{prooftree}
            \end{array}
        \end{equation*}
    \caption{Type System $\cbnBKRVTypeSys$ for the \cbnCalculusSymb-calculus.}
    \label{fig:cbnBKRV_Typing_System_Rules}%
\end{figure}
\begin{figure}[t]
        \begin{equation*}
            \begin{array}{c}
                \begin{array}{c} \\[-0.1cm]
                    \begin{prooftree}
                        \inferBangBKRVVar{x : \M \vdash x : \M}
                    \end{prooftree}
                \end{array}
            \qquad
                \begin{prooftree}
                    \hypo{\Gamma \vdash t : \mset{\M \typeArrow \sigma}}
                    \hypo{\Delta \vdash u : \M}
                    \inferBangBKRVApp{\Gamma + \Delta \vdash \app[\,]{t}{u} : \sigma}
                \end{prooftree}
        \\[16pt]
                \begin{prooftree}
                    \hypo{(\Gamma_i, x : \M_i \vdash t : \sigma_i)_{i \in I}}
                    \inferBangBKRVAbs{+_{i \in I} \Gamma_i \vdash \abs{x}{t} : \mset{\M_i \typeArrow \sigma_i}_{i \in I}}
                \end{prooftree}
            \qquad
                \begin{prooftree}
                    \hypo{\Gamma, x : \M \vdash t : \sigma}
                    \hypo{\Delta \vdash u : \M}
                    \inferBangBKRVEs{\Gamma + \Delta \vdash t\esub{x}{u} : \sigma}
                \end{prooftree}
            \end{array}
        \end{equation*}
    \caption{Type System $\cbvBKRVTypeSys$ for the \cbvCalculusSymb-calculus.}
    \label{fig:cbvBKRV_Typing_System_Rules}%
\end{figure}

The salient property of type systems $\cbnBKRVTypeSys$ and
$\cbvBKRVTypeSys$ is characterizing normalization in \cbnCalculusSymb
and \cbvCalculusSymb, respectively.
\begin{lemma}[\cite{BucciarelliKesnerRiosViso20,BucciarelliKesnerRiosViso23}]
    Let $t \in \cbnSetTerms$, then:
    \begin{itemize}
    \item[\bltI] $t$ is $\cbnCalculusSymb$ \surfaceTxt normalizing iff
        it is $\cbnBKRVTypeSys$-typable.
    \item[\bltI] $t$ is $\cbvCalculusSymb$ \surfaceTxt normalizing iff
        it is $\cbvBKRVTypeSys$-typable.
    \end{itemize}
\end{lemma}

Both \cbnCalculusSymb and \cbvCalculusSymb can be embedded into
\bangCalculusSymb by decorating each term with the $\oc$ and
$\der{\cdot}$ modalities. The embedding for \cbnCalculusSymb is
straightforward, while various embeddings for \cbvCalculusSymb have
been proposed in the literature 
\cite{Girard87,MaraistOderskyTurnerWadler95,MaraistOderskyTurnerWadler99,GuerrieriManzonetto18,BucciarelliKesnerRiosViso20,BucciarelliKesnerRiosViso23,ArrialGuerrieriKesner24bis}, each with its own strengths and weaknesses. In this
work, we use the embeddings from
\cite{BucciarelliKesnerRiosViso20,BucciarelliKesnerRiosViso23} defined
below:
\begin{equation*}
    \begin{array}{rcl c rcl}
        \cbnToBangBKRV{x}
            &\coloneqq& x
    &~\hspace{1cm}~&
        \cbvToBangBKRV{x}
            &\coloneqq& \oc x
    \\
        \cbnToBangBKRV{(\abs{x}{t})}
            &\coloneqq& \abs{x}{\cbnToBangBKRV{t}}
    &~\hspace{1cm}~&
        \cbvToBangBKRV{(\abs{x}{t})}
            &\coloneqq& \oc\abs{x}{\cbvToBangBKRV{t}}
    \\
        \cbnToBangBKRV{(\app{t}{u})}
            &\coloneqq& \app[\,]{\cbnToBangBKRV{t}}{\oc\cbnToBangBKRV{u}}
    &~\hspace{1cm}~&
        \cbvToBangBKRV{(\app{t}{u})}
            &\coloneqq& \left\{\begin{array}{ll}
                \app{\bangLCtxt<s>}{\cbvToBangBKRV{u}}
                    &\text{if} \; \cbvToBangBKRV{t} = \bangLCtxt<\oc s>
            \\
                \app{\der{\cbvToBangBKRV{t}}}{\cbvToBangBKRV{u}}
                    &\text{otherwise}
            \end{array}\right.
    \\
        \cbnToBangBKRV{(t\esub{x}{u})}
            &\coloneqq& \cbnToBangBKRV{t}\esub{x}{\oc\cbnToBangBKRV{u}}
    &~\hspace{1cm}~&
        \cbvToBangBKRV{(t\esub{x}{u})}
            &\coloneqq& \cbvToBangBKRV{t}\esub{x}{\cbvToBangBKRV{u}}
    \end{array}
\end{equation*}

These translations are extended to contexts as expected by setting
$\cbnToBangBKRV{\Hole} \coloneqq \Hole$ and $\cbvToBangBKRV{\Hole}
\coloneqq \Hole$.
\begin{example}
    \label{example:embedding_cbn_and_cbv}%
    Recalling \Cref{example:reduction_cbn_and_cbv}, one has
    $\cbnToBangBKRV{t}_0 = \app[\,]{(\abs{x}{\app[\,]{\app[\,]{y}{\oc
    x}}{\oc x}})}{\oc (\app{\Id}{\oc\Id})}$, $\cbnToBangBKRV{t}_1 =
    \app[\,]{\app[\,]{y}{\oc(\app{\Id}{\oc\Id})}}{\oc(\app{\Id}{\oc\Id})}$,
    $\cbvToBangBKRV{t}_0 = \app{(\abs{x}{(\app{\der{\app[\,]{y}{\oc
    x}}}{\oc x})})}{(\app[\,]{(\abs{z}{\oc z})}{\oc(\abs{z}{\oc
    z})})}$ and $\cbvToBangBKRV{t}_2 =
    \app{\der{\app[\,]{y}{\oc(\abs{z}{\oc z})}}}{\oc(\abs{z}{\oc
    z})}$.
\end{example}
Let us give some intuition on these embeddings. In \cbnCalculusSymb,
any argument (right-hand side of application or substitution) can be
erased/duplicated, as bang terms in the \bangCalculusSymb-calculus, so
that arguments must be translated to bang terms. In \cbvCalculusSymb,
only values can be erased/duplicated so that values must be translated
to bang terms. However, this remark alone is not sufficient to achieve
a \cbvCalculusSymb embedding enjoying good properties, and in
particular to translate \cbvCalculusSymb-normal forms into
\bangCalculusSymb-normal forms. The translation of applications is
precisely designed in order to guarantee this property.

These embeddings preserve reductions, which will allows us to show
that meaningfulness if preserved through embedding.

\begin{lemma}[Simulation \cite{BucciarelliKesnerRiosViso20,BucciarelliKesnerRiosViso23}]
    \label{lem:Bang_BKRV_Reduction}%
    Let $t, u \in \cbnSetTerms$.
    \begin{enumerate}
    \item \label{lem:Bang_BKRV_Reduction_if_Cbn_BKRV_Reduction}%
        If $t \cbnArr*_S u$ then $\cbnToBangBKRV{t} \bangArr*_S
        \cbnToBangBKRV{u}$.

    \item \label{lem:Bang_BKRV_Reduction_if_Cbv_BKRV_Reduction}%
        If $t \cbvArr*_S u$ then $\cbvToBangBKRV{t} \bangArr*_S
        \cbvToBangBKRV{u}$.
    \end{enumerate}
\end{lemma}
\begin{example}
    In \Cref{example:reduction_cbn_and_cbv}, we showed that $t_0
    \cbnArr*_S t_1$ and  $t_0 \cbvArr*_S t_2$. Recalling
    \Cref{example:embedding_cbn_and_cbv}, one has $\cbnToBangBKRV{t}_0
    \bangArr_S (\app[\,]{\app[\,]{y}{\oc x}}{\oc
    x})\esub{x}{\oc(\app{\Id}{\oc\Id})} \bangArr_S
    \cbnToBangBKRV{t}_1$ and $\cbvToBangBKRV{t}_0 \bangArr_S
    (\app{\der{\app[\,]{y}{\oc x}}}{\oc
    x})\esub{x}{(\app[\,]{(\abs{z}{\oc z})}{\oc(\abs{z}{\oc z})})}
    \bangArr_S (\app{\der{\app[\,]{y}{\oc x}}}{\oc x})\esub{x}{(\oc
    z)\esub{z}{\oc(\abs{z}{\oc z})}}%
    \bangArr_S (\app{\der{\app[\,]{y}{\oc x}}}{\oc
    x})\esub{x}{\oc(\abs{z}{\oc z})} \bangArr_S \cbvToBangBKRV{t}_2$.
\end{example}

As the \cbvCalculusSymb-embedding uses $\derSymb$,
some $\bangSymbBang$-step might be needed in the simulation process.

\smallskip
These embeddings also preserve typing, which will make possible to
project \bangCalculusSymb meaningfulness and genericity onto
\cbnCalculusSymb and \cbvCalculusSymb. More precisely, the two
embeddings are proven to be sound and complete with respect to system
$\bangBKRVTypeSys$.

\begin{proposition}[\cite{BucciarelliKesnerRiosViso20,BucciarelliKesnerRiosViso23}]
    \label{lem:Bang_BKRV_Typable}%
    Let $t \in \cbnSetTerms$.
    \begin{enumerate}
    \item \label{lem:Bang_BKRV_Typable_iff_Cbn_BKRV_Typable}%
        One has $\cbnBKRVTr\; \Gamma \vdash t : \sigma$ if and only if $\bangBKRVTr\; \Gamma \vdash \cbnToBangBKRV{t} :
        \sigma$.

    \item \label{lem:Bang_BKRV_Typable_iff_Cbv_BKRV_Typable}%
        One has $\cbvBKRVTr\; \Gamma \vdash t : \sigma$ if and only if $\bangBKRVTr\; \Gamma \vdash \cbvToBangBKRV{t} :
        \sigma$.
    \end{enumerate}
\end{proposition}

A straightforward corollary is that \cbnCalculusSymb and
\cbvCalculusSymb inhabitation properties are well subsumed in
\bangCalculusSymb, as illustrated in~\cite{ArrialGuerrieriKesner23}.
In simpler words, any type inhabited in \cbnCalculusSymb (\resp
\cbvCalculusSymb) is also inhabited in \bangCalculusSymb. As expected,
the converse is false.

\smallskip
In \cbvCalculusSymb and \bangCalculusSymb, typing an arbitrary term
and typing an argument is similar,  as it
can be seen in the right premise $\Delta \vdash u : \M$ of the typing
rules $\bangBKRVAppRuleName$ and $\bangBKRVEsRuleName$ of systems
$\cbvBKRVTypeSys$ and $\bangBKRVTypeSys$. This is not the case in
\cbnCalculusSymb, as the right premise of the $\bangBKRVAppRuleName$
and $\bangBKRVEsRuleName$ rules of system $\cbnBKRVTypeSys$ requires,
not a \emph{single} derivation, but a
\emph{set} $(\Delta_i \vdash u : \tau_i)_{i \in I}$ of
typing derivation for the same term $u$. In the
logical
characterization (\Cref{t:logical-characterization}), we
check that arguments of a
  given type can be inhabited. We therefore need to
reflect the typability of arguments
  ---rather than typablity of arbitrary terms--- in the definition of \cbnCalculusSymb inhabitation.

\begin{definition}
    In system $\cbnBKRVTypeSys$, a non-multitype $\sigma$ is
    \defn{inhabited}, noted $\cbnBKRVInhPred{\sigma}$, if $\Pi
    \cbnBKRVTr \emptyset \vdash t : \sigma$ for some $\Pi$ and $t$. A
    multitype $\mset{\tau_i}_{i \in I}$ is \defn{inhabited} in system
    $\cbnBKRVTypeSys$, noted $\cbnBKRVInhPred{\mset{\tau_i}_{i \in
    I}}$ if there exists $u \in \cbnSetTerms$ such that for each $i
    \in I$, $\cbnBKRVTr\; \emptyset \vdash u : \tau_i$.

    In system $\cbvBKRVTypeSys$, a type $\sigma$ is \defn{inhabited},
    noted $\cbvBKRVInhPred{\sigma}$, if $\Pi \cbvBKRVTr \emptyset
    \vdash t : \sigma$ for some $\Pi$ and $t$.
\end{definition}
In particular, the type $\emptymset$ is inhabited
in both \cbnCalculusSymb and \cbvCalculusSymb
(\ie
$\cbnBKRVInhPred{\emptymset}$ and $\cbvBKRVInhPred{\emptymset}$).
Similarly, the environment $\emptyset$ is also trivially inhabited in both
(\ie $\cbnBKRVInhPred{\emptyset}$ and $\cbvBKRVInhPred{\emptyset}$).

\subsection{\cbnCalculusSymb Meaningfulness and Surface Genericity}
\label{subsec:Cbn_Meaningfulness}

In this subsection, our attention shifts towards the
\cbnCalculusSymb-calculus, where we show  that its notion of
meaningfulness is subsumed by that of \bangCalculusSymb.
This observation enables us to project the surface genericity theorem
accordingly.
We start by introducing \cbnCalculusSymb-meaningfulness.
\begin{definition}
    \label{def:cbn-meaningful}%
    A term $t\in \cbnSetTerms$ is \cbnCalculusSymb-\defn{meaningful}
    if there is a testing context $\cbnTCtxt$ such that $\cbnTCtxt<t>
    \cbnArr*_S \Id$, where testing contexts are defined by $\cbnTCtxt \Coloneqq \Hole \vsep
    \app[\,]{\cbnTCtxt}{u} \vsep
    \app[\,]{(\abs{x}{\cbnTCtxt})}{u}$.\footnotemark
    \footnotetext{Usually, \cbnCalculusSymb-meaningfulness (aka
    \emph{solvability}) is defined using contexts of the form
    $(\abs{x_1\dots x_m}\Hole)N_1\dots N_n$ ($m,n \geq 0$)
    \cite{Barendregt75,barendregt84nh,RoccaP04}, instead of testing
    contexts. It is easy to check that the two definitions are
    equivalent in \cbnCalculusSymb. The benefit of our definition is
    that the same testing contexts are also used to define
    \cbvCalculusSymb-meaningfulness~(\Cref{subsec:Cbv_Meaningfulness}).}
\end{definition}
For example $t = \app{x}{(\abs{y}{\Omega})}$ is
\cbnCalculusSymb-meaningful as $\cbvTCtxt<t> \cbnArr*_S \Id$ for
$\cbvTCtxt = \app{(\abs{x}{\Hole})}{(\abs{z}{\Id})}$, while $\Omega$
and $\abs{x}{\Omega}$ are \cbnCalculusSymb-meaningless as for whatever
testing context $\Omega$ and $\abs{x}{\Omega}$ are plugged into,
$\Omega$ will not be erased. Finally, based on
this definition of meaningfulness it is quite natural to
defined  the types of
observable terms in \cbnCalculusSymb as the identity types (\ie
$\typeobs{\cbnBKRVTypeSys} \coloneqq \{[\sigma] \typeArrow \sigma
\vsep \sigma \text{ type}\}$).

Unlike  \bangCalculusSymb,
\cbnCalculusSymb-meaningfulness can be characterized both
\emph{operationally}, through surface normalizability, and
\emph{logically}, through typability in system~$\cbnBKRVTypeSys$.
Moreover, this logical characterization
turns out to be equivalent to \emph{$\cbnBKRVTypeSys$-testability}, meaning
that \cbnCalculusSymb-meaningfulness can also be characterized via typability and inhabitation, as \mbox{already observed in~\cite{BucciarelliKR21}}.
 
\begin{restatable}[Characterizations of \cbnCalculusSymb-Meaningfulness~\cite{BucciarelliKR15,BucciarelliKR21,BucciarelliKesnerRiosViso20}]{lemma}{RecCbnMeaningfulnessCharacterizations}
    \LemmaToFromProof{cbnBKRV_characterizes_meaningfulness}%
    Let $t \in \cbvSetTerms$.
    \begin{enumerate}
    \item \textbf{(Operational)} %
        \label{lem:cbnBKRV_characterizes_meaningfulness_operational}%
        $t$ is \cbnCalculusSymb-meaningful
        iff $t$ is \cbnCalculusSymb surface-normalizing.

    \item \hspace{0.4cm}\textbf{(Logical)}\hspace{0.2cm} %
        \label{lem:cbnBKRV_characterizes_meaningfulness_logical}%
        (1) $t$ is \cbnCalculusSymb-meaningful iff
        (2) $t$ is
        $\cbnBKRVTypeSys$-typable iff (3) $t$ is
        $\cbnBKRVTypeSys$-testable.
  \end{enumerate}
\end{restatable}
\smallskip
Thanks to the specific shape of \cbnCalculusSymb-normal forms, we can
always type a \cbnCalculusSymb-meaningful term $t$ by a typing
$\typing{\Gamma}{\sigma}$ such that the \emph{non-empty} multitypes in
$\Gamma$ and $\cbnTypeArgs{\sigma}$ are of the form $\mset{\emptymset
\typeArrow \cdots \emptymset \typeArrow \mset{\alpha} \typeArrow
\alpha}$. These types are trivially inhabited by erasers of the form
$\abs{x_1}{\cdots\abs{x_n}{\Id}}$, used in the proof that
$\cbnBKRVTypeSys$-typability implies \cbnCalculusSymb-meaningfulness.

Having an operational characterization of meaningfulness seems to
point out that transforming a result into something observable is a
trivial operation in \cbnCalculusSymb. Indeed, using simulation
(\Cref{lem:Bang_BKRV_Reduction}.(\ref{lem:Bang_BKRV_Reduction_if_Cbv_BKRV_Reduction})),
we easily show that \cbnCalculusSymb-meaningful is preserved by the
\cbnCalculusSymb-embedding, thus confirming this intuition. Moreover,
and thanks to the logical characterization
(\Cref{lem:cbnBKRV_characterizes_meaningfulness}.(\ref{lem:cbnBKRV_characterizes_meaningfulness_logical})),
we show that the converse also holds, yielding the following~result.

\begin{restatable}{theorem}{RecCbnMeaningfulEquivBangMeaningful}
    \label{lem:Cbn_Meaningful_iff_Bang_Meaningful}
    \label{cor:CbnMeaningfulEquivBangMeaningful}
    Let $t \in \cbnSetTerms$, then $t$ is \cbnCalculusSymb-meaningful
    iff $\cbnToBangBKRV{t}$ is \bangCalculusSymb-meaningful.
\end{restatable}
\begin{proof} ~
\begin{description}
\item[$(\Rightarrow)$] We present here an operational proof.
    Let $t$ be \cbnCalculusSymb-meaningful, thus
    $\cbnTCtxt<t>
    \cbnArr*_S \Id$ for some testing context $\cbnTCtxt$. 
    By induction on $\cbnTCtxt$, one has that
    $\cbnToBangBKRV{(\cbnTCtxt<t>)} =
    \cbnToBangBKRV{\cbnTCtxt}\cbnCtxtPlug{\cbnToBangBKRV{t}}$. By
    simulation
    (\Cref{lem:Bang_BKRV_Reduction}.\ref{lem:Bang_BKRV_Reduction_if_Cbn_BKRV_Reduction}),
    one deduces that
    $\cbnToBangBKRV{\cbnTCtxt}\cbnCtxtPlug{\cbnToBangBKRV{t}}
    \bangArr*_S \abs{x}{x}$ thus
    $\app{\cbnToBangBKRV{\cbnTCtxt}\cbnCtxtPlug{\cbnToBangBKRV{t}}}{\oc\oc
    y} \bangArr*_S \app{(\abs{x}{x})}{\oc\oc y} \bangArr*_S \oc y$.
    Notice that $\app{\cbnToBangBKRV{\cbnTCtxt}}{\oc\oc y}$ is a
    \bangCalculusSymb-testing context. We thus conclude that
    $\cbnToBangBKRV{t}$ is \bangCalculusSymb-meaningful.

\item[$(\Leftarrow)$] Let $\cbnToBangBKRV{t}$ be
    \bangCalculusSymb-meaningful, then using
    \Cref{lem:Bang_Meaningfulness_Typing_and_Inhabitation}, it is
    $\bangBKRVTypeSys$-testable and thus $\bangBKRVTypeSys$-typable. 
    By
    \Cref{lem:Bang_BKRV_Typable}.\ref{lem:Bang_BKRV_Typable_iff_Cbn_BKRV_Typable},
    $t$ is $\cbnBKRVTypeSys$-typable and hence $t$ is
    \cbnCalculusSymb-meaningful by
    \Cref{lem:cbnBKRV_characterizes_meaningfulness}.
\qedhere
\end{description}

\end{proof}

\Cref{cor:CbnMeaningfulEquivBangMeaningful} states that
\cbnCalculusSymb-meaningfulness precisely aligns with
\bangCalculusSymb-meaningfulness on its image, strengthening the idea
that these two notions are adequately chosen.
Thanks to \Cref{cor:CbnMeaningfulEquivBangMeaningful}, we can now
project surface genericity from \bangCalculusSymb to \cbnCalculusSymb.

\begin{restatable}[\cbnCalculusSymb Qualitative Surface Genericity]{theorem}{RecCbnQualitativeSurfaceGenericity}
    \label{lem:Cbn_Qualitative_Surface_Genericity}%
    Let $\cbnFCtxt$ be a full context. If $\cbnFCtxt<t>$ is
    \cbnCalculusSymb-meaningful for some
    \cbnCalculusSymb-meaningless $t \in \cbnSetTerms$, then
    $\cbnFCtxt<u>$ is \cbnCalculusSymb-meaningful \mbox{for
    every $u \in \cbnSetTerms$}.
\end{restatable}
\begin{proof}
Let $t \in \cbnSetTerms$ be \cbnCalculusSymb-meaningless and
$\cbnFCtxt$ be a full context. Suppose that $\cbnFCtxt<t>$ is
\cbnCalculusSymb-meaningful, then using
\Cref{lem:Cbn_Meaningful_iff_Bang_Meaningful},
and the fact that
$\cbnToBangBKRV{(\cbnFCtxt<t>)} =
\cbnToBangBKRV{\cbnFCtxt}\cbnCtxtPlug{\cbnToBangBKRV{t}}$ (simple
  induction on $\cbnFCtxt$),  is
\bangCalculusSymb-meaningful, and $\cbnToBangBKRV{t}$ is
$\bangCalculusSymb$-meaningless. By
\Cref{lem:Bang_Qualitative_Surface_Genericity}, for any $u \in
\cbnSetTerms$,
$\cbnToBangBKRV{\cbnFCtxt}\cbnCtxtPlug{\cbnToBangBKRV{u}} =
\cbnToBangBKRV{(\cbnFCtxt<u>)}$ is \bangCalculusSymb-meaningful, and hence $\cbnFCtxt<u>$ is
\cbnCalculusSymb-meaningful using
\Cref{lem:Cbn_Meaningful_iff_Bang_Meaningful}.
\end{proof}

\subsection{\cbvCalculusSymb Meaningfulness and Surface Genericity}
\label{subsec:Cbv_Meaningfulness}%

We now move to the \cbvCalculusSymb-calculus, where we show that its
notion of meaningfulness is subsumed by that of the
\bangCalculusSymb-calculus, and then project surface genericity
theorem~accordingly.

\smallskip
Adapting meaningfulness from \cbnCalculusSymb to \cbvCalculusSymb by
replacing \cbnCalculusSymb-reduction with \cbvCalculusSymb-reduction
may seem initially promising. This notion, known as
\cbvCalculusSymb-solvability, has appealing
properties~\cite{paolini99tia,RoccaP04,AccattoliPaolini12,CarraroGuerrieri14,GuerrieriPaoliniRonchi17,AccattoliGuerrieri22bis}.
Unfortunately, Accattoli and Guerrieri showed that genericity fails in
such setting~\cite{AccattoliGuerrieri22bis}, and that equating
unsolvable terms yields an inconsistent
theory~\cite{AccattoliGuerrieri22bis}. Consequently,
\cbvCalculusSymb-meaningfulness cannot be identified with
\cbvCalculusSymb-solvability. Identifying appropriate notions to
capture \cbvCalculusSymb meaningful $\lambda$-terms and formally
validating these notions has been a longstanding and
challenging~open~question.

Paolini and Ronchi Della Rocca \cite{paolini99tia,RoccaP04} introduced
the notion of \emph{potentially valuability} for
\CBVSymb, also studied in
\cite{PaoliniPimentelRonchi06,AccattoliPaolini12,CarraroGuerrieri14,Garcia-PerezN16}
and renamed \emph{(\cbvCalculusSymb) scrutability} in
\cite{AccattoliGuerrieri22bis}. This notion, which we introduce below,
proves to be suitable \cbvCalculusSymb-meaningfulness. Notably, it
aligns seamlessly with \bangCalculusSymb-meaningfulness through the
\cbvCalculusSymb-embedding and thus enjoys a genericity theorem.

\begin{definition}
    \label{def:cbv-meaningful}%
    A term $t\in \cbvSetTerms$ is
    \cbvCalculusSymb-\defn{meaningful} if there exists a testing
    context $\cbvTCtxt$ and a value $v$ such that
    $\cbvTCtxt<t> \cbvArr*_S v$, where testing contexts are defined by
    $\cbvTCtxt\;\Coloneqq\; \Hole \vsep \app[\,]{\cbvTCtxt}{u} \vsep
    \app[\,]{(\abs{x}{\cbvTCtxt})}{u}$.
\end{definition}
For example $t = \app{x}{(\abs{y}{z})}$ is \cbvCalculusSymb-meaningful
as $\cbvTCtxt<t> \cbvArr*_S \abs{y}{z}$ for $\cbvTCtxt =
\app{(\abs{x}{\Hole})}{(\abs{z}{z})}$, while $\Omega$ and
$\app{x}{\Omega}$ are \cbvCalculusSymb-meaningless as for whatever
testing context $\Omega$ and $\app{x}{\Omega}$ are plugged into,
$\Omega$ will not be erased.

Notice that this definition closely mirrors that of
\bangCalculusSymb-meaningfulness, with  the primary difference being the
replacement of \bangCalculusSymb values
for  those of \cbvCalculusSymb. Since
values are typed with multitypes, it is natural to take them as types
of the observable terms in \cbvCalculusSymb (\ie $\typeobs{\cbvBKRVTypeSys}
\coloneqq \{\M \vsep \M \text{ multitype}\}$). Consequently, and thanks to the
preservation of typing
(\Cref{lem:Bang_BKRV_Typable}.\ref{lem:Bang_BKRV_Typable_iff_Cbv_BKRV_Typable}),
one easily shows that testability is preserved through the
\cbvCalculusSymb translation: if a term $t$ is
$\cbvBKRVTypeSys$-testable, then its image $\cbvToBangBKRV{t}$ is
$\bangBKRVTypeSys$-testable.

\smallskip
As in \cbnCalculusSymb and unlike 
\bangCalculusSymb, \cbvCalculusSymb-meaningfulness can
actually be characterized both \emph{operationally}, through surface
normalizability, and \emph{logically}, through typability in
system~$\cbvBKRVTypeSys$. Moreover, the
logical characterization turns out to be equivalent to
$\cbvBKRVTypeSys$-testability, meaning that
\cbvCalculusSymb-meaningfulness is also characterized by means of
typability and inhabitation.
\begin{restatable}[Characterizations of \cbvCalculusSymb-Meaningfulness~\cite{AccattoliPaolini12,AccattoliGuerrieri22bis,BucciarelliKesnerRiosViso20}]{lemma}{RecCbvMeaningfulnessCharacterizations}
    \LemmaToFromProof{cbvBKRV_characterizes_meaningfulness}%
    Let $t \in \cbvSetTerms$.
    \begin{enumerate}
    \item \textbf{(Operational)} %
        \label{lem:cbvBKRV_characterizes_meaningfulness_operational}%
        $t$ is \cbvCalculusSymb-meaningful iff $t$ is \cbvCalculusSymb
        surface-normalizing.

    \item \hspace{0.4cm}\textbf{(Logical)}\hspace{0.2cm} %
        \label{lem:cbvBKRV_characterizes_meaningfulness_logical}%
        (1) $t$ is \cbvCalculusSymb-meaningful iff (2) $t$ is
        $\cbvBKRVTypeSys$-typable iff (3) $t$ is
        $\cbvBKRVTypeSys$-testable.
  \end{enumerate}
\end{restatable}

The notion of observable aligns in \cbvCalculusSymb and
\bangCalculusSymb, at least from the type perspective. This yields a
simple fully semantical proof of the preservation of
\cbvCalculusSymb-meaningfulness.

\begin{theorem}
    \label{lem:Cbv_Meaningful_iff_Bang_Meaningful}%
    Let $t \in \cbvSetTerms$, then $t$ is \cbvCalculusSymb-meaningful
    iff $\cbvToBangBKRV{t}$ is \bangCalculusSymb-meaningful.
\end{theorem}
\begin{proof}~
\begin{description}
\item[$(\Rightarrow)$] We present here a semantical proof. Let $t$ be
    \cbvCalculusSymb-meaningful, then using
    \Cref{lem:cbvBKRV_characterizes_meaningfulness}, one has that $t$
    is $\cbvBKRVTypeSys$-testable thus, by preservation of
    testability, $\cbvToBangBKRV{t}$ is $\bangBKRVTypeSys$-testable
    and one concludes that $\cbvToBangBKRV{t}$ is
    \bangCalculusSymb-meaningful according to
    \Cref{t:logical-characterization}.

\item[$(\Leftarrow)$] Let $\cbvToBangBKRV{t}$ be
    \bangCalculusSymb-meaningful, then using
    \Cref{lem:Bang_Meaningfulness_Typing_and_Inhabitation}, it is
    $\bangBKRVTypeSys$-testable thus $\bangBKRVTypeSys$-typable. 
    By
    \Cref{lem:Bang_BKRV_Typable}.\ref{lem:Bang_BKRV_Typable_iff_Cbv_BKRV_Typable},
    $t$ is $\cbvBKRVTypeSys$-typable and thus $t$ is
    \cbvCalculusSymb-meaningful by
    \Cref{lem:cbvBKRV_characterizes_meaningfulness}.
\qedhere
\end{description}
\end{proof}


\Cref{lem:Cbv_Meaningful_iff_Bang_Meaningful} states that
\cbvCalculusSymb-meaningfulness precisely aligns with
\bangCalculusSymb-meaningfulness on its image, strengthening the idea
that these two notions are adequately chosen. Thanks to
\Cref{lem:Cbv_Meaningful_iff_Bang_Meaningful}, we can now project
surface genericity from \bangCalculusSymb to \cbvCalculusSymb.

\begin{restatable}[\cbvCalculusSymb Qualitative Surface Genericity]{theorem}{RecCbvQualitativeSurfaceGenericity}
    \label{lem:Cbv_Qualitative_Surface_Genericity}%
    Let $\cbvFCtxt$ be a full context. If $\cbvFCtxt<t>$ is
    \cbvCalculusSymb-meaningful for some
    \cbvCalculusSymb-meaningless $t \in \cbvSetTerms$, then
    $\cbvFCtxt<u>$ is \cbvCalculusSymb-meaningful \mbox{for
    every $u \in \cbvSetTerms$}.
\end{restatable}
\begin{proof}
Let $t \in \cbvSetTerms$ be \cbvCalculusSymb-meaningless and
$\cbvFCtxt$ be a full context. Suppose that $\cbvFCtxt<t>$ is
\cbvCalculusSymb-meaningful, then using
\Cref{lem:Cbv_Meaningful_iff_Bang_Meaningful},
$\cbvToBangBKRV{(\cbvFCtxt<t>)}$ is \bangCalculusSymb-meaningful, and
$\cbvToBangBKRV{t}$ is $\bangCalculusSymb$-meaningless. By induction
on $\cbnFCtxt$, $\cbvToBangBKRV{(\cbvFCtxt<t>)} =
\cbvToBangBKRV{\cbvFCtxt}\cbvCtxtPlug{\cbvToBangBKRV{t}}$ thus
$\cbvToBangBKRV{\cbvFCtxt}\cbvCtxtPlug{\cbvToBangBKRV{t}}$ is
\bangCalculusSymb-meaningful. By
\Cref{lem:Bang_Qualitative_Surface_Genericity}, for any $u \in
\cbvSetTerms$,
$\cbvToBangBKRV{\cbvFCtxt}\cbvCtxtPlug{\cbvToBangBKRV{u}}$ is
\bangCalculusSymb-meaningful. So, by typing preservation
(\Cref{lem:Bang_BKRV_Typable}.\ref{lem:Bang_BKRV_Typable_iff_Cbv_BKRV_Typable})
again, $\cbvToBangBKRV{(\cbvFCtxt<u>)}$ is
\bangCalculusSymb-meaningful, and hence $\cbvFCtxt<u>$ is \cbvCalculusSymb-meaningful using
\Cref{lem:Cbv_Meaningful_iff_Bang_Meaningful}.
\end{proof}

\section{Conclusion and Future Work}
\label{s:conclusion}

We define a notion of meaningful term, in a unifying well-established
framework \bangCalculusSymb that is able to capture both \CBNSymb and
\CBVSymb\ calculi. We validate this notion of meaningfulness by
providing a (high-level) characterization based on both typability and
inhabitation, and showing a (surface) genericity result. Furthermore,
both meaningfulness and genericity in \bangCalculusSymb are shown to
capture their respective notions in \cbnCalculusSymb and
\cbvCalculusSymb.

Several questions remain to be explored. First of all, a notion of
stratified reduction, a finer operational semantics generalizing
surface reduction to different levels, has been recently defined for
\cbnCalculusSymb and \cbvCalculusSymb~\cite{ArrialGuerrieriKesner24}.
Stratified reduction is a key tool to show a full genericity result
for both strategies. We plan to transfer these techniques to the more
general framework of \bangCalculusSymb, so that full genericity for
\cbnCalculusSymb and \cbvCalculusSymb can be simply obtained by
projecting the more general notion of full genericity for
\bangCalculusSymb\ via \CBNSymb/\CBVSymb translations.

We also plan to further study the properties of the theories
$\mathcal{H}$ (\resp $\mathcal{H}^*$) generated by equating all the
meaningful terms (resp. all observational equivalent terms) in
\bangCalculusSymb, which could be possibly related. We believe that
these theories can be related to the corresponding ones in
\cbnCalculusSymb and \cbvCalculusSymb.

We would like to extend our study to other natural objects in the
theory of programming, such as Böhm trees for \bangCalculusSymb\ and
their related theorems (\eg approximation and separability). Böhm
trees for \bangCalculusSymb are expected to encompass both
\cbnCalculusSymb \cite{barendregt84nh} and
\cbvCalculusSymb~\cite{KerinecManzonettoPagani20} ones. 

Unifying frameworks such as \bangCalculusSymb should also provide
other general results for \cbnCalculusSymb and \cbvCalculusSymb, such
as standardization, separability, etc. All this is left to future
work. Finally, a more ambitious goal would be to generalize these
results to models of computations with effects, such as global memory,
non-determinism, exceptions, etc. This would approach our study on
\bangCalculusSymb to a more general unifying framework such as
call-by-push-value~\cite{Levy99,Levy04}.

\nocite{AccattoliGK20} 
\bibliography{main}

\newpage 

\appendix
\setbool{inAppendix}{true}

{
\renewcommand{\bangSymbStratified}{\bangSymbFull}
\renewcommand{\indexOmega}{{ }}
\renewcommand{\bangStratCtxt}{\bangCtxtStyle{F}\bangCtxtPlugOption}

\section{Proofs of \Cref{sec:dbang}}
  
\subsection{Confluence of $\bangArr_{S_\indexOmega}<dB> \cup
\bangArr_{S_\indexOmega}<d!>$}

\begin{definition}
  Let us now extend the reduction relation to contexts. We define the
    context reduction relation $\bangArr_{S_\indexOmega}<dB>$ (\resp
    $\bangArr_{S_\indexOmega}<d!>$) to be the union of the
    $\bangSymbStratified_\indexOmega$-closures of the relations
    $\bangSymbBeta_1, \bangSymbBeta_2$ and $\bangSymbBeta_3$ (\resp
    $\bangSymbBang_1$ and $\bangSymbBang_2$) defined as follows:
    \begin{equation*}
        \begin{array}{rcl}
                \app{\bangLCtxt<\abs{x}{\bangStratCtxt_\indexOmega}>}{u}
                    &\mapstoR[\bangSymbBeta_1]&
                \bangLCtxt<\bangStratCtxt_\indexOmega\esub{x}{u}>
            \\
                \app{\bangLCtxt<\abs{x}{t}>}{\bangStratCtxt_\indexOmega}
                    &\mapstoR[\bangSymbBeta_2]&
                \bangLCtxt<t\esub{x}{\bangStratCtxt_\indexOmega}>
            \\
                \app{\bangLCtxt_1<\bangLCtxt_2<\abs{x}{t}>\esub{y}{\bangStratCtxt_\indexOmega}>}{u}
                    &\mapstoR[\bangSymbBeta_3]&
                \bangLCtxt_1<\bangLCtxt_2<t\esub{x}{u}>\esub{y}{\bangStratCtxt_\indexOmega}>
            \\[0.2cm]
                \der{\bangLCtxt<\oc \bangStratCtxt_\indexOmega>}
                    &\mapstoR[\bangSymbBang_1]&
                \bangLCtxt<\bangStratCtxt_\indexOmega>
            \\
                \der{\bangLCtxt_1<\bangLCtxt_2<\oc s>\esub{x}{\bangStratCtxt_\indexOmega}>}
                    &\mapstoR[\bangSymbBang_2]&
                \bangLCtxt_1<\bangLCtxt_2<s>\esub{x}{\bangStratCtxt_\indexOmega}>
        \end{array}
    \end{equation*}
\end{definition}

\begin{lemma}
    \label{lem:Bang_Full_dB_d!_Context_Reduction_Lifted_to_Terms}%
    Let $\bangStratCtxt_\indexOmega^1, \bangStratCtxt_\indexOmega^2$
    be two contexts such that $\bangStratCtxt_\indexOmega^1
    \bangArr_{S_\indexOmega}<R> \bangStratCtxt_\indexOmega^2$ for some
    $\rel \in \{\bangSymbBeta, \bangSymbBang\}$. Then, for any term $t
    \in \bangSetTerms$, one has that $\bangStratCtxt_\indexOmega^1<t>
    \bangArr_{S_\indexOmega}<R> \bangStratCtxt_\indexOmega^2<t>$.
\end{lemma}
\stableProof{%
    \begin{proof}
We distinguish two cases:
\begin{itemize}
\item[\bltI] $\rel = \bangSymbBeta$: By definition, there exist three
    contexts $\bangStratCtxt_\indexOmega,
    \bangStratCtxt_\indexOmega^{1'}$ and
    $\bangStratCtxt_\indexOmega^{2'}$ such that
    $\bangStratCtxt_\indexOmega^1 =
    \bangStratCtxt_\indexOmega<\bangStratCtxt_\indexOmega^{1'}>$,
    $\bangStratCtxt_\indexOmega^2 =
    \bangStratCtxt_\indexOmega<\bangStratCtxt_\indexOmega^{2'}>$ and
    $\bangStratCtxt_\indexOmega^{1'} \mapstoR[\rel']
    \bangStratCtxt_\indexOmega^{2'}$ for some $\rel' \in
    \{\bangSymbBeta_1, \bangSymbBeta_2, \bangSymbBeta_3\}$. We
    distinguish three cases:
    \begin{itemize}
    \item[\bltII] $\rel' = \bangSymbBeta_1$: Then
        $\bangStratCtxt_\indexOmega^{1'} =
        \app{\bangLCtxt<\abs{x}{\bangStratCtxt'_\indexOmega}>}{u}$ and
        $\bangStratCtxt_\indexOmega^{2'} =
        \bangLCtxt<\bangStratCtxt'_\indexOmega\esub{x}{u}>$. Thus
        $\bangStratCtxt_\indexOmega^{1'}<t> =
        \app{\bangLCtxt<\abs{x}{\bangStratCtxt'_\indexOmega<t>}>}{u}$
        and $\bangStratCtxt_\indexOmega^{2'}<t> =
        \bangLCtxt<\bangStratCtxt'_\indexOmega<t>\esub{x}{u}>$.
        Therefore $\bangStratCtxt_\indexOmega^{1'}<t>
        \mapstoR[\bangSymbBeta] \bangStratCtxt_\indexOmega^{2'}<t>$
        hence $\bangStratCtxt_\indexOmega^1<t>
        \bangArr_{S_\indexOmega}<dB> \bangStratCtxt_\indexOmega^2<t>$.

    \item[\bltII] $\rel' = \bangSymbBeta_2$: Then
        $\bangStratCtxt_\indexOmega^{1'} =
        \app{\bangLCtxt<\abs{x}{t}>}{\bangStratCtxt'_\indexOmega}$ and
        $\bangStratCtxt_\indexOmega^{2'} =
        \bangLCtxt<t\esub{x}{\bangStratCtxt'_\indexOmega}>$. Thus
        $\bangStratCtxt_\indexOmega^{1'}<t> =
        \app{\bangLCtxt<\abs{x}{t}>}{\bangStratCtxt'_\indexOmega<t>}$
        and $\bangStratCtxt_\indexOmega^{2'}<t> =
        \bangLCtxt<t\esub{x}{\bangStratCtxt'_\indexOmega<t>}>$.
        Therefore $\bangStratCtxt_\indexOmega^{1'}<t>
        \mapstoR[\bangSymbBeta] \bangStratCtxt_\indexOmega^{2'}<t>$
        hence $\bangStratCtxt_\indexOmega^1<t>
        \bangArr_{S_\indexOmega}<dB> \bangStratCtxt_\indexOmega^2<t>$.

    \item[\bltII] $\rel' = \bangSymbBeta_3$: Then
        $\bangStratCtxt_\indexOmega^{1'} =
        \app{\bangLCtxt_1<\bangLCtxt_2<\abs{x}{t}>\esub{y}{\bangStratCtxt'_\indexOmega}>}{u}$
        and $\bangStratCtxt_\indexOmega^{2'} =
        \bangLCtxt_1<\bangLCtxt_2<t\esub{x}{u}>\esub{y}{\bangStratCtxt'_\indexOmega}>$.
        Thus $\bangStratCtxt_\indexOmega^{1'}<t> =
        \app{\bangLCtxt_1<\bangLCtxt_2<\abs{x}{t}>\esub{y}{\bangStratCtxt'_\indexOmega<t>}>}{u}$
        and $\bangStratCtxt_\indexOmega^{2'}<t> =
        \bangLCtxt_1<\bangLCtxt_2<t\esub{x}{u}>\esub{y}{\bangStratCtxt'_\indexOmega<t>}>$.
        Therefore $\bangStratCtxt_\indexOmega^{1'}<t>
        \mapstoR[\bangSymbBeta] \bangStratCtxt_\indexOmega^{2'}<t>$
        hence $\bangStratCtxt_\indexOmega^1<t>
        \bangArr_{S_\indexOmega}<dB> \bangStratCtxt_\indexOmega^2<t>$.
    \end{itemize}

\item[\bltI] $\rel = \bangSymbBang$: By definition, there exist three
    contexts $\bangStratCtxt_\indexOmega,
    \bangStratCtxt_\indexOmega^{1'}$ and
    $\bangStratCtxt_\indexOmega^{2'}$ such that
    $\bangStratCtxt_\indexOmega^1 =
    \bangStratCtxt_\indexOmega<\bangStratCtxt_\indexOmega^{1'}>$,
    $\bangStratCtxt_\indexOmega^2 =
    \bangStratCtxt_\indexOmega<\bangStratCtxt_\indexOmega^{2'}>$ and
    $\bangStratCtxt_\indexOmega^{1'} \mapstoR[\rel']
    \bangStratCtxt_\indexOmega^{2'}$ for some $\rel' \in
    \{\bangSymbBang_1, \bangSymbBang_2\}$.
    \begin{itemize}
    \item[\bltII] $\rel' = \bangSymbBang_1$: Then
        $\bangStratCtxt_\indexOmega^{1'} = \der{\bangLCtxt<\oc
        \bangStratCtxt'_\indexOmega>}$ and
        $\bangStratCtxt_\indexOmega^{2'} =
        \bangLCtxt<\bangStratCtxt'_\indexOmega>$. Thus
        $\bangStratCtxt_\indexOmega^{1'}<t> = \der{\bangLCtxt<\oc
        \bangStratCtxt'_\indexOmega<t>>}$ and
        $\bangStratCtxt_\indexOmega^{2'}<t> =
        \bangLCtxt<\bangStratCtxt'_\indexOmega<t>>$. Therefore
        $\bangStratCtxt_\indexOmega^{1'}<t> \mapstoR[\bangSymbBang]
        \bangStratCtxt_\indexOmega^{2'}<t>$ hence
        $\bangStratCtxt_\indexOmega^1<t> \bangArr_{S_\indexOmega}<dB>
        \bangStratCtxt_\indexOmega^2<t>$.

    \item[\bltII] $\rel' = \bangSymbBang_2$: Then
        $\bangStratCtxt_\indexOmega^{1'} =
        \der{\bangLCtxt_1<\bangLCtxt_2<\oc
        s>\esub{x}{\bangStratCtxt'_\indexOmega}>}$ and
        $\bangStratCtxt_\indexOmega^{2'} =
        \bangLCtxt_1<\bangLCtxt_2<s>\esub{x}{\bangStratCtxt'_\indexOmega}>$.
        Thus $\bangStratCtxt_\indexOmega^{1'}<t> =
        \der{\bangLCtxt_1<\bangLCtxt_2<\oc
        s>\esub{x}{\bangStratCtxt'_\indexOmega<t>}>}$ and
        $\bangStratCtxt_\indexOmega^{2'}<t> =
        \bangLCtxt_1<\bangLCtxt_2<s>\esub{x}{\bangStratCtxt'_\indexOmega<t>}>$.
        Therefore $\bangStratCtxt_\indexOmega^{1'}<t>
        \mapstoR[\bangSymbBang] \bangStratCtxt_\indexOmega^{2'}<t>$
        hence $\bangStratCtxt_\indexOmega^1<t>
        \bangArr_{S_\indexOmega}<dB> \bangStratCtxt_\indexOmega^2<t>$.
        \qedhere
    \end{itemize}
\end{itemize}
\end{proof}%
} \deliaLu \giulioLu

\begin{lemma}
    \label{lem:Bang_Full_dB_d!_Confluent}%
    The reduction $\bangArr_{S_\indexOmega}<dB> \cup
    \bangArr_{S_\indexOmega}<d!>$ is diamond.
\end{lemma}
\stableProof{%
    \begin{proof}
Let $t, u_1, u_2 \in \bangSetTerms$ such that $u_1 \neq u_2$ with $t
\bangArr_{S_\indexOmega}<R_1> u_1$ and $t
\bangArr_{S_\indexOmega}<R_2> u_2$ for some $\rel_1, \rel_2 \in
\{\bangSymbBeta, \bangSymbBang\}$. Let us show that there exists $s
\in \bangSetTerms$ such that $u_1 \bangArr_{S_\indexOmega}<R_2> s$ and
$u_2 \bangArr_{S_\indexOmega}<R_1> s$. We distinguish two cases:
\begin{itemize}
\item[\bltI] $\rel_1 = \bangSymbBeta$: By definition, there exist
    $\bangStratCtxt_\indexOmega$, $\bangLCtxt$, $s_1, s_2 \in
    \bangSetTerms$ such that $t =
    \bangStratCtxt_\indexOmega<\app{\bangLCtxt<\abs{x}{s_1}>}{s_2}>$
    and $u_1 =
    \bangStratCtxt_\indexOmega<\bangLCtxt<s_1\esub{x}{s_2}>>$. Since
    $u_1 \neq u_2$, then the step $t \bangArr_{S_\indexOmega}<R_2>
    u_2$ either happens in $\bangStratCtxt_\indexOmega$, $\bangLCtxt$,
    $s_1$ or $s_2$.
    \begin{itemize}
    \item[\bltII] $\bangStratCtxt_\indexOmega
        \bangArr_{S_\indexOmega}<R_2> \bangStratCtxt'_\indexOmega$
        with $u_2 =
        \bangStratCtxt'_\indexOmega<\app{\bangLCtxt<\abs{x}{s_1}>}{s_2}>$:
        We set $s :=
        \bangStratCtxt'_\indexOmega<\bangLCtxt<s_1\esub{x}{s_2}>>$
        which concludes this case since using
        \Cref{lem:Bang_Full_dB_d!_Context_Reduction_Lifted_to_Terms}:
        \begin{equation*}
            \begin{array}{ccc}
                \bangStratCtxt_\indexOmega<\app{\bangLCtxt<\abs{x}{s_1}>}{s_2}>     &\bangArr_{S_\indexOmega}<R_1>  &\bangStratCtxt_\indexOmega<\bangLCtxt<s_1\esub{x}{s_2}>>
            \\[0.1cm]
                \bangDownArr_{S_\indexOmega}<R_2>                                   &                               &\bangDownArr_{S_\indexOmega}<R_2>
            \\[0.1cm]
                \bangStratCtxt'_\indexOmega<\app{\bangLCtxt<\abs{x}{s_1}>}{s_2}>    &\bangArr_{S_\indexOmega}<R_1>  &\bangStratCtxt'_\indexOmega<\bangLCtxt<s_1\esub{x}{s_2}>>
            \end{array}
        \end{equation*}

    \item[\bltII] $\bangLCtxt_\indexOmega
        \bangArr_{S_\indexOmega}<R_2> \bangLCtxt'_\indexOmega$ with
        $u_2 =
        \bangStratCtxt_\indexOmega<\app{\bangLCtxt'<\abs{x}{s_1}>}{s_2}>$:
        We set $s :=
        \bangStratCtxt_\indexOmega<\bangLCtxt'<s_1\esub{x}{s_2}>>$
        which concludes this case since using
        \Cref{lem:Bang_Full_dB_d!_Context_Reduction_Lifted_to_Terms}:
        \begin{equation*}
            \begin{array}{ccc}
                \bangStratCtxt_\indexOmega<\app{\bangLCtxt<\abs{x}{s_1}>}{s_2}>     &\bangArr_{S_\indexOmega}<R_1>  &\bangStratCtxt_\indexOmega<\bangLCtxt<s_1\esub{x}{s_2}>>
            \\[0.1cm]
                \bangDownArr_{S_\indexOmega}<R_2>                                   &                               &\bangDownArr_{S_\indexOmega}<R_2>
            \\[0.1cm]
                \bangStratCtxt_\indexOmega<\app{\bangLCtxt'<\abs{x}{s_1}>}{s_2}>    &\bangArr_{S_\indexOmega}<R_1>  &\bangStratCtxt_\indexOmega<\bangLCtxt'<s_1\esub{x}{s_2}>>
            \end{array}
        \end{equation*}

    \item[\bltII] $s_1 \bangArr_{S_\indexOmega}<R_2> s'_1$ with $u_2 =
        \bangStratCtxt_\indexOmega<\app{\bangLCtxt<\abs{x}{s'_1}>}{s_2}>$:
        We set $s :=
        \bangStratCtxt_\indexOmega<\bangLCtxt<s'_1\esub{x}{s_2}>>$
        which concludes this case since:
        \begin{equation*}
            \begin{array}{ccc}
                \bangStratCtxt_\indexOmega<\app{\bangLCtxt<\abs{x}{s_1}>}{s_2}>     &\bangArr_{S_\indexOmega}<R_1>  &\bangStratCtxt_\indexOmega<\bangLCtxt<s_1\esub{x}{s_2}>>
            \\[0.1cm]
                \bangDownArr_{S_\indexOmega}<R_2>                                   &                               &\bangDownArr_{S_\indexOmega}<R_2>
            \\[0.1cm]
                \bangStratCtxt_\indexOmega<\app{\bangLCtxt<\abs{x}{s'_1}>}{s_2}>    &\bangArr_{S_\indexOmega}<R_1>  &\bangStratCtxt_\indexOmega<\bangLCtxt<s'_1\esub{x}{s_2}>>
            \end{array}
        \end{equation*}

    \item[\bltII] $s_2 \bangArr_{S_\indexOmega}<R_2> s'_2$ with $u_2 =
        \bangStratCtxt_\indexOmega<\app{\bangLCtxt<\abs{x}{s_1}>}{s'_2}>$:
        We set $s :=
        \bangStratCtxt_\indexOmega<\bangLCtxt<s_1\esub{x}{s'_2}>>$
        which concludes this case since:
        \begin{equation*}
            \begin{array}{ccc}
                \bangStratCtxt_\indexOmega<\app{\bangLCtxt<\abs{x}{s_1}>}{s_2}>     &\bangArr_{S_\indexOmega}<R_1>  &\bangStratCtxt_\indexOmega<\bangLCtxt<s_1\esub{x}{s_2}>>
            \\[0.1cm]
                \bangDownArr_{S_\indexOmega}<R_2>                                   &                               &\bangDownArr_{S_\indexOmega}<R_2>
            \\[0.1cm]
                \bangStratCtxt_\indexOmega<\app{\bangLCtxt<\abs{x}{s_1}>}{s'_2}>    &\bangArr_{S_\indexOmega}<R_1>  &\bangStratCtxt_\indexOmega<\bangLCtxt<s_1\esub{x}{s'_2}>>
            \end{array}
        \end{equation*}
    \end{itemize}

\item[\bltI] $\rel_1 = \bangSymbBang$: By definition, there exist
    $\bangStratCtxt_\indexOmega$, $\bangLCtxt$, $s \in \bangSetTerms$
    such that $t = \bangStratCtxt_\indexOmega<\der{\bangLCtxt<\oc
    s>}>$ and $u_1 = \bangStratCtxt_\indexOmega<\bangLCtxt<s>>$. Since
    $u_1 \neq u_2$, then the step $t \bangArr_{S_\indexOmega}<R_2>
    u_2$ either happens in $\bangStratCtxt_\indexOmega$, $\bangLCtxt$
    or $s$.
    \begin{itemize}
    \item[\bltII] $\bangStratCtxt_\indexOmega
        \bangArr_{S_\indexOmega}<R_2> \bangStratCtxt'_\indexOmega$
        with $u_2 = \bangStratCtxt'_\indexOmega<\der{\bangLCtxt<\oc
        s>}>$: We set $s :=
        \bangStratCtxt'_\indexOmega<\bangLCtxt<s>>$ which concludes
        this case since using
        \Cref{lem:Bang_Full_dB_d!_Context_Reduction_Lifted_to_Terms}:
        \begin{equation*}
            \begin{array}{ccc}
                \bangStratCtxt_\indexOmega<\der{\bangLCtxt<\oc s>}>     &\bangArr_{S_\indexOmega}<R_1>  &\bangStratCtxt_\indexOmega<\bangLCtxt<s>>
            \\[0.1cm]
                \bangDownArr_{S_\indexOmega}<R_2>                       &                               &\bangDownArr_{S_\indexOmega}<R_2>
            \\[0.1cm]
                \bangStratCtxt'_\indexOmega<\der{\bangLCtxt<\oc s>}>    &\bangArr_{S_\indexOmega}<R_1>  &\bangStratCtxt'_\indexOmega<\bangLCtxt<s>>
            \end{array}
        \end{equation*}

    \item[\bltII] $\bangLCtxt_\indexOmega
        \bangArr_{S_\indexOmega}<R_2> \bangLCtxt'_\indexOmega$ with
        $u_2 = \bangStratCtxt_\indexOmega<\der{\bangLCtxt'<\oc s>}>$:
        We set $s := \bangStratCtxt_\indexOmega<\bangLCtxt'<s>>$ which
        concludes this case since using
        \Cref{lem:Bang_Full_dB_d!_Context_Reduction_Lifted_to_Terms}:
        \begin{equation*}
            \begin{array}{ccc}
                \bangStratCtxt_\indexOmega<\der{\bangLCtxt<\oc s>}>     &\bangArr_{S_\indexOmega}<R_1>  &\bangStratCtxt_\indexOmega<\bangLCtxt<s>>
            \\[0.1cm]
                \bangDownArr_{S_\indexOmega}<R_2>                       &                               &\bangDownArr_{S_\indexOmega}<R_2>
            \\[0.1cm]
                \bangStratCtxt_\indexOmega<\der{\bangLCtxt'<\oc s>}>    &\bangArr_{S_\indexOmega}<R_1>  &\bangStratCtxt_\indexOmega<\bangLCtxt'<s>>
            \end{array}
        \end{equation*}

    \item[\bltII] $s \bangArr_{S_\indexOmega}<R_2> s'$ with $u_2 =
        \bangStratCtxt_\indexOmega<\der{\bangLCtxt<\oc s'>}>$: We set
        $s := \bangStratCtxt_\indexOmega<\bangLCtxt<s'>>$ which
        concludes this case since:
        \begin{equation*}
            \begin{array}{ccc}
                \bangStratCtxt_\indexOmega<\der{\bangLCtxt<\oc s>}>     &\bangArr_{S_\indexOmega}<R_1>  &\bangStratCtxt_\indexOmega<\bangLCtxt<s>>
            \\[0.1cm]
                \bangDownArr_{S_\indexOmega}<R_2>                       &                               &\bangDownArr_{S_\indexOmega}<R_2>
            \\[0.1cm]
                \bangStratCtxt_\indexOmega<\der{\bangLCtxt<\oc s'>}>    &\bangArr_{S_\indexOmega}<R_1>  &\bangStratCtxt_\indexOmega<\bangLCtxt<s'>>
            \end{array}
        \end{equation*}
    \end{itemize}
\end{itemize}
\end{proof}
} \deliaLu \giulioLu

\begin{corollary}
    \label{lem:Bang_Full_dB_U_d!_Confluent}%
    The reduction $\bangArr_{S_\indexOmega}<dB> \cup
    \bangArr_{S_\indexOmega}<d!>$ is confluent.
\end{corollary}
\begin{proof}
    Immediate consequence of \Cref{lem:Bang_Full_dB_d!_Confluent}.
\end{proof}



\subsection{Confluence of $\bangArr_{S_\indexOmega}<s!>$}

\subsubsection{Local Confluence of $\bangArr_{S_\indexOmega}<s!>$}

\begin{definition}
    Let us now extend the $\bangSymbStratified_\indexOmega$-reduction
    relation to contexts. We define the context reduction relation
    $\bangArr_{S_\indexOmega}<s!>$ to be the union of the
    $\bangStratCtxt_\indexOmega$-closures of the following relations
    $\bangSymbSubs_1$ and $\bangSymbSubs_2$:
    \begin{equation*}
        \begin{array}{rcl}
                t\esub{x}{\bangLCtxt_1<\bangLCtxt_2<\oc u>\esub{y}{\bangStratCtxt_\indexOmega}>}
                    &\mapstoR[\bangSymbSubs_1]&
                \bangLCtxt_1<\bangLCtxt_2<t\isub{x}{u}>\esub{y}{\bangStratCtxt_\indexOmega}>
            \\
                t\esub{x}{\bangLCtxt<\oc\bangStratCtxt_\indexOmega>}
                    &\mapstoR[\bangSymbSubs_2]&
                \bangLCtxt<t\isub{x}{\bangStratCtxt_\indexOmega}>
        \end{array}
    \end{equation*}
\end{definition}

\begin{lemma}
    \label{lem:Bang_Full_s!_Context_Reduction_Lifted_to_Terms}%
    Let $\bangStratCtxt_\indexOmega^1, \bangStratCtxt_\indexOmega^2$
    be two contexts such that $\bangStratCtxt_\indexOmega^1
    \bangArr_{S_\indexOmega}<s!> \bangStratCtxt_\indexOmega^2$, then
    for any term $t \in \bangSetTerms$, one has that
    $\bangStratCtxt_\indexOmega^1<t> \bangArr_{S_\indexOmega}<s!>
    \bangStratCtxt_\indexOmega^2<t>$.
\end{lemma}
\stableProof{%
    \begin{proof}
By definition, there exist three contexts $\bangStratCtxt_\indexOmega,
\bangStratCtxt_\indexOmega^{1'}$ and $\bangStratCtxt_\indexOmega^{2'}$
such that $\bangStratCtxt_\indexOmega^1 =
\bangStratCtxt_\indexOmega<\bangStratCtxt_\indexOmega^{1'}>$,
$\bangStratCtxt_\indexOmega^2 =
\bangStratCtxt_\indexOmega<\bangStratCtxt_\indexOmega^{2'}>$ and
$\bangStratCtxt_\indexOmega^{1'} \mapstoR[\rel]
\bangStratCtxt_\indexOmega^{2'}$ for some $\rel \in \{\bangSymbSubs_1,
\bangSymbSubs_2\}$. We distinguish two cases:
\begin{itemize}
\item[\bltIII] $\rel = \bangSymbSubs_1$: Then
    $\bangStratCtxt_\indexOmega^{1'} =
    t\esub{x}{\bangLCtxt_1<\bangLCtxt_2<\oc
    u>\esub{y}{\bangStratCtxt_\indexOmega}>}$ and
    $\bangStratCtxt_\indexOmega^{2'} =
    \bangLCtxt_1<\bangLCtxt_2<t\isub{x}{u}>\esub{y}{\bangStratCtxt_\indexOmega}>$.
    Thus $\bangStratCtxt_\indexOmega^{1'}<t> =
    t\esub{x}{\bangLCtxt_1<\bangLCtxt_2<\oc
    u>\esub{y}{\bangStratCtxt_\indexOmega<t>}>}$ and
    $\bangStratCtxt_\indexOmega^{2'}<t> =
    \bangLCtxt_1<\bangLCtxt_2<t\isub{x}{u}>\esub{y}{\bangStratCtxt_\indexOmega<t>}>$.
    Therefore $\bangStratCtxt_\indexOmega^{1'}<t>
    \mapstoR[\bangSymbSubs] \bangStratCtxt_\indexOmega^{2'}<t>$ hence
    $\bangStratCtxt_\indexOmega^1<t> \bangArr_{S_\indexOmega}<s!>
    \bangStratCtxt_\indexOmega^2<t>$.

\item[\bltIII] $\rel = \bangSymbSubs_2$: Then
    $\bangStratCtxt_\indexOmega^{1'} =
    t\esub{x}{\bangLCtxt<\oc\bangStratCtxt_\indexOmega>}$ and
    $\bangStratCtxt_\indexOmega^{2'} =
    \bangLCtxt<t\isub{x}{\bangStratCtxt_\indexOmega}>$. Thus
    $\bangStratCtxt_\indexOmega^{1'}<t> =
    t\esub{x}{\bangLCtxt<\oc\bangStratCtxt_\indexOmega<t>>}$ and
    $\bangStratCtxt_\indexOmega^{2'}<t> =
    \bangLCtxt<t\isub{x}{\bangStratCtxt_\indexOmega<t>}>$. Therefore
    $\bangStratCtxt_\indexOmega^{1'}<t> \mapstoR[\bangSymbSubs]
    \bangStratCtxt_\indexOmega^{2'}<t>$ hence
    $\bangStratCtxt_\indexOmega^1<t> \bangArr_{S_\indexOmega}<s!>
    \bangStratCtxt_\indexOmega^2<t>$.
    \qedhere
\end{itemize}
\end{proof}%
} \giulioLu \deliaLu

\begin{remark}
    We also trivially extend the notion of substitution to
    $\bangStratCtxt_\indexOmega$-contexts by setting $\Hole\isub{x}{t}
    := \Hole$.
\end{remark}

\begin{lemma}
    \label{lem:Bang_Full_s!_Locally_Confluent}%
    The relation $\bangArr_{S_\indexOmega}<s!>$ is locally confluent.
\end{lemma}
\stableProof{%
    \begin{proof}
Let $t, u_1, u_2 \in \bangSetTerms$ such that $t
\bangArr_{S_\indexOmega}<s!> u_1$, $t \bangArr_{S_\indexOmega}<s!>
u_2$ and $u_1 \neq u_2$. We reason by induction on $t$:
\begin{itemize}
\item[\bltI] $t = x$: Impossible since it contradicts $t
    \bangArr_{S_\indexOmega}<s!> u_1$.

\item[\bltI] $t = \abs{x}{t'}$: Since $t$ cannot be a
    $\bangSymbSubs$-redex, then necessarily
    $u_1 = \abs{x}{t'_1}$ and $u_2 = \abs{x}{t'_2}$ with $t'
    \bangArr_{S_\indexOmega}<s!> t'_1$ and $t'
    \bangArr_{S_\indexOmega}<s!> t'_2$. By \ih on $t'$, there exists
    $s' \in \bangSetTerms$ such that $t_1
    \bangArr*_{S_\indexOmega}<s!> s'$ and $t_2
    \bangArr*_{S_\indexOmega}<s!> s'$. We set $s := \abs{x}{s'}$
    concluding this case since by contextual closure:
    \begin{equation*}
        \begin{array}{lll}
            \abs{x}{t'}                             &\bangArr_{S_\indexOmega}<s!>   &\abs{x}{t'_1}
    \\[0.2cm]
            \;\;\bangDownArr_{S_\indexOmega}<s!>    &                               &\;\;\bangDownArr*_{S_\indexOmega}<s!>
    \\[0.2cm]
            \abs{x}{t'_2}                           &\bangArr*_{S_\indexOmega}<s!>  &\abs{x}{s'}
        \end{array}
    \end{equation*}

\item[\bltI] $t = \app{t_1}{t_2}$: Since $t$ cannot be a
    $\bangSymbSubs$-redex, four different
    cases can be distinguished:
    \begin{itemize}
    \item[\bltII] $u_1 = \app{u'_1}{t_2}$, $u_2 = \app{u'_2}{t_2}$
        with $t_1 \bangArr_{S_\indexOmega}<s!> u'_1$ and $t_1
        \bangArr_{S_\indexOmega}<s!> u'_2$: By \ih on $t_1$, there
        exists $s_1 \in \bangSetTerms$ such that $u'_1
        \bangArr*_{S_\indexOmega}<s!> s_1$ and $u'_2
        \bangArr*_{S_\indexOmega}<s!> s_1$. We set $s :=
        \app{s_1}{t_2}$ concluding this case since by contextual
        closure:
        \begin{equation*}
            \begin{array}{lll}
                \app{t_1}{t_2}                          &\bangArr_{S_\indexOmega}<s!>   &\app{u'_1}{t_2}
        \\[0.2cm]
                \;\;\bangDownArr_{S_\indexOmega}<s!>    &                               &\;\;\bangDownArr*_{S_\indexOmega}<s!>
        \\[0.2cm]
                \app{u'_2}{t_2}                         &\bangArr*_{S_\indexOmega}<s!>  &\app{s_1}{t_2}
            \end{array}
        \end{equation*}

    \item[\bltII] $u_1 = \app{u'_1}{t_2}$, $u_2 = \app{t_1}{u'_2}$
        with $t_1 \bangArr_{S_\indexOmega}<s!> u'_1$ and $t_2
        \bangArr_{S_\indexOmega}<s!> u'_2$: We set $s :=
        \app{u'_1}{u'_2}$ concluding this case since by contextual
        closure:
        \begin{equation*}
            \begin{array}{lll}
                \app{t_1}{t_2}                          &\bangArr_{S_\indexOmega}<s!>   &\app{u'_1}{t_2}
        \\[0.2cm]
                \;\;\bangDownArr_{S_\indexOmega}<s!>    &                               &\;\;\bangDownArr_{S_\indexOmega}<s!>
        \\[0.2cm]
                \app{t_1}{u'_2}                         &\bangArr_{S_\indexOmega}<s!>   &\app{u'_1}{u'_2}
            \end{array}
        \end{equation*}

    \item[\bltII] $u_1 = \app{t_1}{u'_1}$, $u_2 = \app{u'_1}{t_2}$
        with $t_2 \bangArr_{S_\indexOmega}<s!> u'_1$ and $t_1
        \bangArr_{S_\indexOmega}<s!> u'_2$: We set $s :=
        \app{u'_1}{u'_2}$ concluding this case since by contextual
        closure:
        \begin{equation*}
            \begin{array}{lll}
                \app{t_1}{t_2}                          &\bangArr_{S_\indexOmega}<s!>   &\app{t_1}{u'_2}
        \\[0.2cm]
                \;\;\bangDownArr_{S_\indexOmega}<s!>    &                               &\;\;\bangDownArr_{S_\indexOmega}<s!>
        \\[0.2cm]
                \app{u'_1}{t_2}                         &\bangArr_{S_\indexOmega}<s!>   &\app{u'_1}{u'_2}
            \end{array}
        \end{equation*}

    \item[\bltII] $u_1 = \app{t_1}{u'_1}$, $u_2 = \app{t_1}{u'_2}$
        with $t_2 \bangArr_{S_\indexOmega}<s!> u'_1$ and $t_2
        \bangArr_{S_\indexOmega}<s!> u'_2$: By \ih on $t_2$, there
        exists $s_2 \in \bangSetTerms$ such that $t_1
        \bangArr*_{S_\indexOmega}<s!> s_2$ and $t_2
        \bangArr*_{S_\indexOmega}<s!> s_2$. We set $s :=
        \app{t_1}{s_2}$ concluding this case since by contextual
        closure:
        \begin{equation*}
            \begin{array}{lll}
                \app{t_1}{t_2}                          &\bangArr_{S_\indexOmega}<s!>   &\app{t_1}{u'_1}
        \\[0.2cm]
                \;\;\bangDownArr_{S_\indexOmega}<s!>    &                               &\;\;\bangDownArr*_{S_\indexOmega}<s!>
        \\[0.2cm]
                \app{t_1}{u'_2}                         &\bangArr*_{S_\indexOmega}<s!>  &\app{t_1}{s'}
            \end{array}
        \end{equation*}
    \end{itemize}

\item[\bltI] $t = t_1\esub{x}{t_2}$: Notice that $t$ can be a
    $\bangSymbSubs$-redex but since $u_1
    \neq u_2$, three cases have to be distinguished:
    \begin{itemize}
    \item[\bltII] $t$ is the redex reduced in $t
        \bangArr_{S_\indexOmega}<s!> u_1$: Then $t_2 = \bangLCtxt<\oc
        t'_2>$ and $u_1 = \bangLCtxt<t_1\isub{x}{t'_2}>$. We
        distinguish two cases:
        \begin{itemize}
        \item[\bltIII] $u_2 = u'_1\esub{x}{t_2}$ and $t_1
            \bangArr_{S_\indexOmega}<s!> u'_1$: We set $s =
            \bangLCtxt<u'_1\isub{x}{t'_2}>$. By induction on $t_1$,
            one has that $t_1\isub{x}{t'_2}
            \bangArr_{S_\indexOmega}<s!> u'_1\isub{x}{t'_2}$ thus by
            contextual closure, one has that:
            \begin{equation*}
                \begin{array}{lll}
                    t_1\esub{x}{\bangLCtxt<\oc t'_2>}               &\bangArr_{S_\indexOmega}<s!>   &u'_1\esub{x}{\bangLCtxt<\oc t'_2>}
            \\[0.2cm]
                    \hspace{0.75cm}\bangDownArr_{S_\indexOmega}<s!> &                               &\hspace{0.75cm}\bangDownArr_{S_\indexOmega}<s!>
            \\[0.2cm]
                    \bangLCtxt<t_1\isub{x}{t'_2}>                   &\bangArr_{S_\indexOmega}<s!>   &\bangLCtxt<u'_1\isub{x}{t'_2}>
                \end{array}
            \end{equation*}
        \end{itemize}

        \item[\bltIII] $u_2 = t_1\esub{x}{u'_2}$ and $t_2
            \bangArr_{S_\indexOmega}<s!> u'_2$: Since $t_2 =
            \bangLCtxt<\oc t'_2>$, three cases can be distinguished:
            \begin{itemize}
            \item[\bltIV] $\bangLCtxt =
                \bangLCtxt_1<\bangLCtxt_3\esub{y}{\bangLCtxt_2<\oc
                t''_2>}>$ and $u'_2 =
                \bangLCtxt_1<\bangLCtxt_2<\bangLCtxt_3<\oc
                t'_2>\isub{y}{t''_2}>>$: We set $s' :=
                \bangLCtxt_1<\bangLCtxt_2<\bangLCtxt_3\isub{x}{t'_2}\bangCtxtPlug{t_1\isub{x}{t'_2\isub{y}{t''_2}}}>>$.
                By induction on $\bangLCtxt_3$, one has that
                $\bangLCtxt_3<t'_2>\isub{x}{t'_2} =
                \bangLCtxt_3\isub{x}{t'_2}\bangCtxtPlug{t_1\isub{x}{t'_2}}$.
                By \mbox{$\alpha$-conversion}, $y \notin
                \freeVar{t_1}$ and therefore, by induction on $t_1$,
                $t_1\isub{x}{t'_2}\isub{y}{t''_2} =
                t_1\isub{y}{t''_2}\isub{x}{t'_2\isub{y}{t''_2}}$ thus:
                \begin{equation*}
                    \begin{array}{lll}
                        t_1\esub{x}{\bangLCtxt_1<\bangLCtxt_3<\oc t'_2>\esub{y}{\bangLCtxt_2<\oc t''_2>}>}  &\bangArr_{S_\indexOmega}<s!>   &t_1\esub{x}{\bangLCtxt_1<\bangLCtxt_2<\bangLCtxt_3\isub{y}{t''_2}\bangCtxtPlug{\oc t'_2\isub{y}{t''_2}}>>}
                \\[0.2cm]
                        \hspace{2cm}\bangDownArr_{S_\indexOmega}<s!>                                        &                               &\hspace{2cm}\bangDownArr_{S_\indexOmega}<s!>
                \\[0.2cm]
                        \bangLCtxt_1<\bangLCtxt_3<t_1\isub{x}{t'_2}>\esub{y}{\bangLCtxt_2<\oc t''_2>}>      &\bangArr_{S_\indexOmega}<s!>   &\bangLCtxt_1<\bangLCtxt_2<\bangLCtxt_3\isub{x}{t'_2}\bangCtxtPlug{t_1\isub{x}{t'_2\isub{y}{t''_2}}}>>
                \\
                                                                                                            &                               &=\bangLCtxt_1<\bangLCtxt_2<\bangLCtxt_3\isub{x}{t'_2}\bangCtxtPlug{t_1\isub{y}{t''_2}\isub{x}{t'_2\isub{y}{t''_2}}}>>
                    \end{array}
                \end{equation*}

            \item[\bltIV] $u'_2 = \bangLCtxt'<\oc t'_2>$ and
                $\bangLCtxt \bangArr_{S_\indexOmega}<s!> \bangLCtxt'$:
                We set $s' := \bangLCtxt'<t_1\isub{x}{t'_2}>$ which
                concludes this case since using
                \Cref{lem:Bang_Full_s!_Context_Reduction_Lifted_to_Terms},
                one has:
                \begin{equation*}
                    \begin{array}{lll}
                        t_1\esub{x}{\bangLCtxt<\oc t'_2>}               &\bangArr_{S_\indexOmega}<s!>   &t_1\esub{x}{\bangLCtxt'<\oc t'_2>}
                \\[0.2cm]
                        \hspace{0.75cm}\bangDownArr_{S_\indexOmega}<s!> &                               &\hspace{0.75cm}\bangDownArr_{S_\indexOmega}<s!>
                \\[0.2cm]
                        \bangLCtxt<t_1\isub{x}{t'_2}>                   &\bangArr_{S_\indexOmega}<s!>   &\bangLCtxt'<t_1\isub{x}{t'_2}>
                    \end{array}
                \end{equation*}

            \item[\bltIV] $u'_2 = \bangLCtxt<\oc t''_2>$ and $t'_2
                \bangArr_{S_\indexOmega}<s!> t''_2$: We set $s' :=
                \bangLCtxt<t_1\isub{x}{t''_2}>$. By induction on
                $t_1$, one has that $t_1\isub{x}{t'_2}
                \bangArr*_{S_\indexOmega}<s!> t_1\isub{x}{t''_2}$
                thus:
                \begin{equation*}
                    \begin{array}{lll}
                        t_1\esub{x}{\bangLCtxt<\oc t'_2>}               &\bangArr_{S_\indexOmega}<s!>   &t_1\esub{x}{\bangLCtxt<\oc t''_2>}
                \\[0.2cm]
                        \hspace{0.75cm}\bangDownArr_{S_\indexOmega}<s!> &                               &\hspace{0.75cm}\bangDownArr_{S_\indexOmega}<s!>
                \\[0.2cm]
                        \bangLCtxt<t_1\isub{x}{t'_2}>                   &\bangArr*_{S_\indexOmega}<s!>  &\bangLCtxt<t_1\isub{x}{t''_2}>
                    \end{array}
                \end{equation*}
            \end{itemize}

    \item[\bltII] $t$ is the redex reduced in $t
        \bangArr_{S_\indexOmega}<s!> u_2$: Same as the previous case.

    \item[\bltII] $t$ is not the redex reduced in $t
        \bangArr_{S_\indexOmega}<s!> u_1$ nor $t
        \bangArr_{S_\indexOmega}<s!> u_2$: We distinguish four cases:
        \begin{itemize}
        \item[\bltIII] $u_1 = u'_1\esub{x}{t_2}$, $u_2 =
            u'_2\esub{x}{t_2}$ with $t_1 \bangArr_{S_\indexOmega}<s!>
            u'_1$ and $t_1 \bangArr_{S_\indexOmega}<s!> u'_2$: By \ih
            on $t_1$, there exists $s_1 \in \bangSetTerms$ such that
            $u'_1 \bangArr*_{S_\indexOmega}<s!> s_1$ and $u'_2
            \bangArr*_{S_\indexOmega}<s!> s_1$. We set $s :=
            s_1\esub{x}{t_2}$ concluding this case since by contextual
            closure:
            \begin{equation*}
                \begin{array}{lll}
                    t_1\esub{x}{t_2}                            &\bangArr_{S_\indexOmega}<s!>   &u'_1\esub{x}{t_2}
            \\[0.2cm]
                    \;\;\;\;\bangDownArr_{S_\indexOmega}<s!>    &                               &\;\;\;\;\bangDownArr*_{S_\indexOmega}<s!>
            \\[0.2cm]
                    u'_2\esub{x}{t_2}                           &\bangArr*_{S_\indexOmega}<s!>  &s_1\esub{x}{t_2}
                \end{array}
            \end{equation*}

        \item[\bltIII] $u_1 = u'_1\esub{x}{t_2}$, $u_2 =
            t_1\esub{x}{u'_2}$ with $t_1 \bangArr_{S_\indexOmega}<s!>
            u'_1$ and $t_2 \bangArr_{S_\indexOmega}<s!> u'_2$: We set
            $s := u'_1\esub{x}{u'_2}$ concluding this case since by
            contextual closure:
            \begin{equation*}
                \begin{array}{lll}
                    t_1\esub{x}{t_2}                            &\bangArr_{S_\indexOmega}<s!>   &u'_1\esub{x}{t_2}
            \\[0.2cm]
                    \;\;\;\;\bangDownArr_{S_\indexOmega}<s!>    &                               &\;\;\;\;\bangDownArr_{S_\indexOmega}<s!>
            \\[0.2cm]
                    t_1\esub{x}{u'_2}                           &\bangArr_{S_\indexOmega}<s!>   &u'_1\esub{x}{u'_2}
                \end{array}
            \end{equation*}

        \item[\bltIII] $u_1 = t_1\esub{x}{u'_1}$, $u_2 =
            u'_1\app{x}{t_2}$ with $t_2 \bangArr_{S_\indexOmega}<s!>
            u'_1$ and $t_1 \bangArr_{S_\indexOmega}<s!> u'_2$: We set
            $s := u'_1\esub{x}{u'_2}$ concluding this case since by
            contextual closure:
            \begin{equation*}
                \begin{array}{lll}
                    t_1\esub{x}{t_2}                            &\bangArr_{S_\indexOmega}<s!>   &t_1\esub{x}{u'_2}
            \\[0.2cm]
                    \;\;\;\;\bangDownArr_{S_\indexOmega}<s!>    &                               &\;\;\;\;\bangDownArr_{S_\indexOmega}<s!>
            \\[0.2cm]
                    u'_1\esub{x}{t_2}                           &\bangArr_{S_\indexOmega}<s!>   &u'_1\esub{x}{u'_2}
                \end{array}
            \end{equation*}

        \item[\bltIII] $u_1 = t_1\esub{x}{u'_1}$, $u_2 =
            t_1\esub{x}{u'_2}$ with $t_2 \bangArr_{S_\indexOmega}<s!>
            u'_1$ and $t_2 \bangArr_{S_\indexOmega}<s!> u'_2$: By \ih
            on $t_2$, there exists $s_2 \in \bangSetTerms$ such that
            $t_1 \bangArr*_{S_\indexOmega}<s!> s_2$ and $t_2
            \bangArr*_{S_\indexOmega}<s!> s_2$. We set $s :=
            t_1\esub{x}{s_2}$ concluding this case since by contextual
            closure:
            \begin{equation*}
                \begin{array}{lll}
                    t_1\esub{x}{t_2}                            &\bangArr_{S_\indexOmega}<s!>   &t_1\esub{x}{u'_1}
            \\[0.2cm]
                    \;\;\;\;\bangDownArr_{S_\indexOmega}<s!>    &                               &\;\;\;\;\bangDownArr*_{S_\indexOmega}<s!>
            \\[0.2cm]
                    t_1\esub{x}{u'_2}                           &\bangArr*_{S_\indexOmega}<s!>  &t_1\esub{x}{s'}
                \end{array}
            \end{equation*}
        \end{itemize}
    \end{itemize}

\item[\bltI] $t = \der{t'}$: Since $t$ cannot be a
    $\bangSymbSubs$-redex, then necessarily
    $u_1 = \der{t'_1}$ and $u_2 = \der{t'_2}$ with $t'
    \bangArr_{S_\indexOmega}<s!> t'_1$ and $t'
    \bangArr_{S_\indexOmega}<s!> t'_2$. By \ih on $t'$, there exists
    $s' \in \bangSetTerms$ such that $t_1
    \bangArr*_{S_\indexOmega}<s!> s'$ and $t_2
    \bangArr*_{S_\indexOmega}<s!> s'$. We set $s := \der{s'}$
    concluding this case since by contextual closure:
    \begin{equation*}
        \begin{array}{lll}
            \der{t'}                                &\bangArr_{S_\indexOmega}<s!>   &\der{t'_1}
    \\[0.2cm]
            \;\;\;\bangDownArr_{S_\indexOmega}<s!>  &                               &\;\;\;\bangDownArr*_{S_\indexOmega}<s!>
    \\[0.2cm]
            \der{t'_2}                              &\bangArr*_{S_\indexOmega}<s!>  &\der{s'}
        \end{array}
    \end{equation*}

\item[\bltI] $t = \oc t'$: Since $t$ cannot be a
    $\bangSymbSubs$-redex, then necessarily
    $u_1 = \oc t'_1$ and $u_2 = \oc t'_2$ with $t'
    \bangArr_{S_\indexOmega}<s!> t'_1$ and $t'
    \bangArr_{S_\indexOmega}<s!> t'_2$. By \ih on $t'$, there exists
    $s' \in \bangSetTerms$ such that $t_1
    \bangArr*_{S_\indexOmega}<s!> s'$ and $t_2
    \bangArr*_{S_\indexOmega}<s!> s'$. We set $s := \oc s'$ concluding
    this case since by contextual closure:
    \begin{equation*}
        \begin{array}{lll}
            \oc t'                              &\bangArr_{S_\indexOmega}<s!>   &\oc t'_1
    \\[0.2cm]
            \bangDownArr_{S_\indexOmega}<s!>    &                               &\bangDownArr*_{S_\indexOmega}<s!>
    \\[0.2cm]
            \oc t'_2                            &\bangArr*_{S_\indexOmega}<s!>  &\oc s'
        \end{array}
    \end{equation*}
\end{itemize}
\end{proof}
} \deliaLu \giulioLu%

\subsubsection{Termination of $\bangArr_{S_\indexOmega}<s!>$}

We use the measure defined in~\cite{AccattoliKesner12bis} for a simpler framework. 

\begin{definition}
    The \emphasis{potential multiplicity} $\bangPotMult_x{t}$ of the
    variable $x$ in the term $t$ is a natural number defined as
    follows: if $x \notin \freeVar{t}$, then $\bangPotMult_x{t} = 0$
    and otherwise (where $y \neq x$ by $\alpha$-conversion):
    \begin{equation*}
        \begin{array}{rcl}
            \bangPotMult_x{x}
                &\coloneqq& 1
        \\
            \bangPotMult_x{\abs{y}{t}}
                &\coloneqq& \bangPotMult_x{t}
        \\
            \bangPotMult_x{\app{t_1}{t_2}}
                &\coloneqq& \bangPotMult_x{t_1} + \bangPotMult_x{t_2}
        \\
            \bangPotMult_x{t_1\esub{y}{t_2}}
                &\coloneqq& \bangPotMult_x{t_1} + \max(1, \bangPotMult_y{t_1}) \cdot \bangPotMult_x{t_2}
        \\
            \bangPotMult_x{\der{t}}
                &\coloneqq& \bangPotMult_x{t}
        \\
            \bangPotMult_x{\oc t}
                &\coloneqq& \bangPotMult_x{t}
        \end{array}
    \end{equation*}
\end{definition}

\NewDocumentCommand{\multisetScProd}{ }{\cdot}

The potential multiplicity of $x$ in $t$ is the number of free
occurrences of $x$ in the \emph{unfolding} of $t$, except that for
subterms of the form $s\esub{y}{u}$ with $y \notin \freeVar{s}$ the
potential multiplicity of $x$ in $u$ is counted as if there were $1$
free occurrence of $y$ in $s$.

By exploiting potential multiplicities we can define a measure of the
global degree of sharing of a given term, and use this measure to
prove that the $\bangArr_{S_\indexOmega}<s!>$ reduction terminates. We
consider multisets of integers. We use $\emptymset$ to denote the
empty multiset, $\uplus$ to denote multiset union, and $\preceq$
(resp. $\prec$) for the standard order (resp. strict order) on
multisets. Given $n, k \in \Nat$ and a finite multiset $M = \mset{i_1,
\dots , i_k} $ over natural numbers, $n \multisetScProd M$ denotes
 the multiset $\mset{n
\multisetScProd i_1, \dots, n \multisetScProd i_k}$ (in particular, $n
\multisetScProd M = \emptymset$ if $M = \emptymset$).

\begin{definition}
    The \emphasis{multiset measure} $\multiSize{t}$ of $t \in \bangSetTerms$ is a finite multiset over $\Nat$ defined by:
    \begin{align*}
            \multiSize{y}
                &\coloneqq \emptymset
        \\
            \multiSize{\abs{x}{t}}
                &\coloneqq \multiSize{t}
        \\
            \multiSize{\app{t_1}{t_2}}
                &\coloneqq \multiSize{t_1} \uplus \multiSize{t_2}
        \\
            \multiSize{t_1\esub{x}{t_2}}
                &\coloneqq \mset{\bangPotMult_x{t_1}} \uplus \multiSize{t_1} \uplus \max(1, \bangPotMult_x{t_1}) \cdot \multiSize{t_2}
        \\
            \multiSize{\der{t}}
                &\coloneqq \multiSize{t}
        \\
            \multiSize{\oc t}
                &\coloneqq \multiSize{t}
    \end{align*}
\end{definition}

\begin{lemma}
    \label{lem:Bang_Substitution_compat_PotMult}%
    Let $t, u \in \bangSetTerms$. For any variables $x, y$ such that $x \neq y$ and $y \notin
    \freeVar{u}$, one has that $\bangPotMult_y{t} =
    \bangPotMult_y{t\isub{x}{u}}$.
\end{lemma}
\stableProof{%
    \begin{proof}
Since $y \notin \freeVar{u}$, one deduces that $\bangPotMult_y{u} =
0$. We proceed by induction on $t$:
\begin{itemize}
\item[\bltI] $t = z$: We distinguish two cases:
    \begin{itemize}
    \item[\bltII] $z=x$: Thus $t\isub{x}{u} = u$ and since $x \neq
        y$ then $\bangPotMult_y{t} = 0$ hence
        $\bangPotMult_y{t\isub{x}{u}} = \bangPotMult_y{t}$.

    \item[\bltII] $z \neq x$: Thus $t\isub{x}{u} = z = t$ hence
        $\bangPotMult_y{t\isub{x}{u}} = \bangPotMult_y{t}$.
    \end{itemize}

\item[\bltI] $t = \abs{z}{t'}$: We can suppose without loss of
	generality that $z \notin \freeVar{u} \cup \{x\}$, thus
	$t\isub{x}{u} = \abs{z}{(t'\isub{x}{u})}$. By \ih on $t'$, one has
	that $\bangPotMult_x{t'} = \bangPotMult_x{t'\isub{x}{u}}$ thus:
    \begin{equation*}
        \begin{array}{rclcl}
            \bangPotMult_y{t}
            &=& \bangPotMult_y{\abs{z}{t'}}
        \\
            &=& \bangPotMult_y{t'}
        \\
            &\eqih& \bangPotMult_y{t'\isub{x}{u}}
        \\
            &=& \bangPotMult_y{\abs{z}{(t'\isub{x}{u})}}
            &=& \bangPotMult_y{t\isub{x}{u}}
        \end{array}
    \end{equation*}

\item[\bltI] $t = \app{t_1}{t_2}$: By \ih $t_1$ and $t_2$, one has
    that $\bangPotMult_y{t_1} = \bangPotMult_y{t_1\isub{x}{u}}$ and
    $\bangPotMult_y{t_2} = \bangPotMult_y{t_2\isub{x}{u}}$ thus
    \begin{equation*}
        \begin{array}{rclcl}
            \bangPotMult_y{t}
            &=& \bangPotMult_y{\app{t_1}{t_2}}
        \\
            &=& \bangPotMult_y{t_1} + \bangPotMult_y{t_2}
        \\
            &\eqih& \bangPotMult_y{t_1\isub{x}{u}} + \bangPotMult_y{t_2\isub{x}{u}}
        \\
            &=& \bangPotMult_y{\app{(t_1\isub{x}{u})}{(t_2\isub{x}{u})}}
            &=& \bangPotMult_y{t\isub{x}{u}}
        \end{array}
    \end{equation*}

\item[\bltI] $t = t_1\esub{z}{t_2}$: We can suppose without loss of
    generality that $z \notin \freeVar{u} \cup \{x\}$, thus
    $t\isub{x}{u} = (t_1\isub{x}{u})\esub{z}{t_2 \isub{x}{u}}$. By \ih
    $t_1$ and $t_2$, one has that $\bangPotMult_y{t_1} =
    \bangPotMult_y{t_1\isub{x}{u}}$, $\bangPotMult_y{t_2} =
    \bangPotMult_y{t_2\isub{x}{u}}$ and $\bangPotMult_z{t_1} =
    \bangPotMult_z{t_1\isub{x}{u}}$. Thus
    \begin{equation*}
        \begin{array}{rclcl}
            \bangPotMult_y{t}
            &=& \bangPotMult_y{t_1\esub{z}{t_2}}
        \\
            &=& \bangPotMult_y{t_1} + \max(1, \bangPotMult_z{t_1}) \times \bangPotMult_y{t_2}
        \\
            &\eqih& \bangPotMult_y{t_1\isub{x}{u}} + \max(1, \bangPotMult_z{t_1\isub{x}{u}}) \times \bangPotMult_y{t_2\isub{x}{u}}
        \\
            &=& \bangPotMult_y{t_1\isub{x}{u}\esub{z}{t_2\isub{x}{u}}}
            &=& \bangPotMult_y{t\isub{x}{u}}
        \end{array}
    \end{equation*}

\item[\bltI] $t = \der{t'}$: By \ih on $t'$, one has that
    $\bangPotMult_x{t'} = \bangPotMult_x{t'\isub{x}{u}}$ thus:
    \begin{equation*}
        \begin{array}{rclcl}
            \bangPotMult_y{t}
            &=& \bangPotMult_y{\der{t'}}
        \\
            &=& \bangPotMult_y{t'}
        \\
            &\eqih& \bangPotMult_y{t'\isub{x}{u}}
        \\
            &=& \bangPotMult_y{\der{(t'\isub{x}{u})}}
            &=& \bangPotMult_y{t\isub{x}{u}}
        \end{array}
    \end{equation*}

\item[\bltI] $t = \oc t'$: By \ih on $t'$, one has that
    $\bangPotMult_x{t'} = \bangPotMult_x{t'\isub{x}{u}}$ thus:
    \begin{equation*}
        \begin{array}{rclcl}
            \bangPotMult_y{t}
            &=& \bangPotMult_y{\oc t'}
        \\
            &=& \bangPotMult_y{t'}
        \\
            &\eqih& \bangPotMult_y{t'\isub{x}{u}}
        \\
            &=& \bangPotMult_y{\oc (t'\isub{x}{u})}
            &=& \bangPotMult_y{t\isub{x}{u}}
        \end{array}
    \end{equation*}
\qedhere
\end{itemize}
\end{proof}
} \deliaLu \giulioLu%

\begin{lemma}
    \label{lem:Bang_Substitution_decreases_PotMult}%
    Let $t, u \in \bangSetTerms$ be terms, then
    $\bangPotMult_y{\bangLCtxt<t\isub{x}{u}>} \leq
    \bangPotMult_y{t\esub{x}{\bangLCtxt<\oc u>}}$.
\end{lemma}
\stableProof{%
    \begin{proof} 
By induction on $\bangLCtxt$:
\begin{itemize}
\item[\bltI] $\bangLCtxt = \Hole$: We proceed by induction on $t$. If $y \notin \freeVar{t\isub{x}{u}}$, then $\bangPotMult_y{t\isub{x}{u}} = 0 \leq \bangPotMult_y{t\esub{x}{\oc u}}$. 
	Let us suppose $y \in \freeVar{t\isub{x}{u}}$. Cases: 
    \begin{itemize}
    \item[\bltII] $t = z$: We distinguish two subcases:
        \begin{itemize}
        \item[\bltIII] $z = x$: Then $t\isub{x}{u} = u$ thus
            $\bangPotMult_y{t\isub{x}{u}} = \bangPotMult_y{u} \leq
            \bangPotMult_y{t} + \max(1, \bangPotMult_x{t}) \cdot
            \bangPotMult_y{u} = \bangPotMult_y{t} + \max(1,
            \bangPotMult_x{t}) \cdot \bangPotMult_y{\oc u} =
            \bangPotMult_y{t\esub{x}{\oc u}}$.

        \item[\bltIII] $z \neq x$: Then $t\isub{x}{u} = t$ thus
            $\bangPotMult_y{t\isub{x}{u}} = \bangPotMult_y{t} \leq
            \bangPotMult_y{t} + \max(1, \bangPotMult_x{t}) \cdot
            \bangPotMult_y{\oc u} = \bangPotMult_y{t\esub{x}{\oc u}}$.
        \end{itemize}

    \item[\bltII] $t = \abs{z}{t'}$: We can suppose without loss of generality that $z \notin \freeVar{u} \cup \{x\}$, so $t\isub{x}{u} =
        \abs{z}{(t'\isub{x}{u})}$. By \ih on $t'$,
        $\bangPotMult_y{t'\isub{x}{u}} \leq
        \bangPotMult_y{t'\esub{x}{\oc u}}$. Thus:
        \begin{equation*}
            \begin{array}{rcl}
                \bangPotMult_y{t\isub{x}{u}}
                    &=& \bangPotMult_y{\abs{z}{(t'\isub{x}{u})}}
                \\
                    &=& \bangPotMult_y{t'\isub{x}{u}}
                \\
                    &\leqih& \bangPotMult_y{t'\esub{x}{\oc u}}
                \\
                    &=& \bangPotMult_y{t'} + \max(1, \bangPotMult_x{t'}) \cdot \bangPotMult_y{\oc u}
                \\
                    &=& \bangPotMult_y{\abs{z}{t'}} + \max(1, \bangPotMult_x{\abs{z}{t'}}) \cdot \bangPotMult_y{\oc u}
                \\
                    &=& \bangPotMult_y{t\esub{x}{\oc u}}
            \end{array}
        \end{equation*}

    \item[\bltII] $t = \app{t_1}{t_2}$: Then $t\isub{x}{u} =
        \app{(t_1\isub{x}{u})}{(t_2\isub{x}{u})}$ and by \ih on $t_1$
        and $t_2$, one has that $\bangPotMult_y{t_1\isub{x}{u}} \leq
        \bangPotMult_y{t_1\esub{x}{\oc u}}$ and
        $\bangPotMult_y{t_2\isub{x}{u}} \leq
        \bangPotMult_y{t_2\esub{x}{\oc u}}$. We distinguish four
        cases:
        \begin{itemize}
        \item[\bltIII] $x \notin \freeVar{t_1}$ and $x \notin
            \freeVar{t_2}$: Then $t\isub{x}{u} = t$ and
            $\bangPotMult_x{t}=0$, thus:
            \begin{equation*}
                \begin{array}{rcl}
                    \bangPotMult_y{t\isub{x}{u}}
                        &=& \bangPotMult_y{t}
                    \\
                        &\leq& \bangPotMult_y{t}
                        + \bangPotMult_y{\oc u} \\
                        &=& \bangPotMult_y{t}
                        + \max(1, \bangPotMult_x{t}) \cdot \bangPotMult_y{\oc u}
                    \\
                        &=& \bangPotMult_y{t\esub{x}{\oc u}}
                \end{array}
            \end{equation*}

        \item[\bltIII] $x \in \freeVar{t_1}$ and $x \notin
            \freeVar{t_2}$: Then $\bangPotMult_x{t_1} \geq 1$ and
            $\bangPotMult_x{t_2} = 0$. By \ih on $t_1$, one has that
            $\bangPotMult_y{t_1\isub{x}{u}} \leq
            \bangPotMult_y{t_1\esub{x}{\oc u}}$ and
            $\bangPotMult_z{t_1\isub{x}{u}} \leq
            \bangPotMult_z{t_1\esub{x}{\oc u}}$.  Therefore:
            \begin{equation*}
                \begin{array}{rcl}
                    \bangPotMult_y{t\isub{x}{u}}
                        &=& \bangPotMult_y{\app{(t_1\isub{x}{u})}{(t_2\isub{x}{u})}}
                    \\
                        &=& \bangPotMult_y{t_1\isub{x}{u}} + \bangPotMult_y{t_2}
                    \\
                        &\leqih& \bangPotMult_y{t_1\esub{x}{\oc u}} + \bangPotMult_y{t_2}
                    \\
                        &=& \bangPotMult_y{t_1} + \max(1, \bangPotMult_x{t_1}) \cdot \bangPotMult_y{\oc u}
                        + \bangPotMult_y{t_2} 
                    \\
                        &=& \left(\bangPotMult_y{t_1} + \bangPotMult_y{t_2}\right)
                        + \max(1, \bangPotMult_x{t_1})  \cdot \bangPotMult_y{\oc u}
                    \\
                        &=& \bangPotMult_y{\app{t_1}{t_2}}
                        + \max(1, \bangPotMult_x{\app{t_1}{t_2}}) \cdot \bangPotMult_y{\oc u}
                    \\
                        &=& \bangPotMult_y{t\esub{x}{\oc u}}
                \end{array}
            \end{equation*}

        \item[\bltIII] $x \notin \freeVar{t_1}$ and $x \in
            \freeVar{t_2}$: Then $\bangPotMult_x{t_1} = 0$ and
            $\bangPotMult_x{t_2} \geq 1$. By \ih on $t_2$, one has
            that $\bangPotMult_y{t_2\isub{x}{u}} \leq
            \bangPotMult_y{t_2\esub{x}{\oc u}}$. Therefore:
            \begin{equation*}
                \begin{array}{rcl}
                    \bangPotMult_y{t\isub{x}{u}}
                        &=& \bangPotMult_y{\app{(t_1\isub{x}{u})}{(t_2\isub{x}{u})}}
                    \\
                        &=& \bangPotMult_y{t_1} + \bangPotMult_y{t_2\isub{x}{u}}
                    \\
                        &\leqih& \bangPotMult_y{t_1} + \bangPotMult_y{t_2\esub{x}{\oc u}}
                    \\
                        &=& \bangPotMult_y{t_1} + \bangPotMult_y{t_2} + \max(1, \bangPotMult_x{t_2}) \cdot \bangPotMult_y{\oc u}
                    \\
                        &=& \left(\bangPotMult_y{t_1} + \bangPotMult_y{t_2}\right)
                        + \max(1, \bangPotMult_x{t_2}) \cdot \bangPotMult_y{\oc u}
                    \\
                        &=& \bangPotMult_y{\app{t_1}{t_2}}
                        + \max(1, \bangPotMult_x{\app{t_1}{t_2}}) \cdot \bangPotMult_y{\oc u}
                    \\
                        &=& \bangPotMult_y{t\esub{x}{\oc u}}
                \end{array}
            \end{equation*}

        \item[\bltIII] $x \in \freeVar{t_1}$ and $x \in
            \freeVar{t_2}$: Then $\bangPotMult_x{t_1} \geq 1$ and
            $\bangPotMult_x{t_2} \geq 1$. By \ih on $t_1$, one has
            that $\bangPotMult_y{t_1\isub{x}{u}} \leq
            \bangPotMult_y{t_1\esub{x}{\oc u}}$ and
            $\bangPotMult_y{t_2\isub{x}{u}} \leq
            \bangPotMult_y{t_2\esub{x}{\oc u}}$.   Therefore:
            \begin{equation*}
                \begin{array}{rcl}
                    \bangPotMult_y{t\isub{x}{u}}
                        &=& \bangPotMult_y{\app{(t_1\isub{x}{u})}{(t_2\isub{x}{u})}}
                    \\
                        &=& \bangPotMult_y{t_1\isub{x}{u}} + \bangPotMult_y{t_2\isub{x}{u}}
                    \\
                        &\leqih& \bangPotMult_y{t_1\esub{x}{\oc u}} + \bangPotMult_y{t_2\esub{x}{\oc u}}
                    \\
                        &=& \left(\bangPotMult_y{t_1} + \max(1, \bangPotMult_x{t_1}) \cdot \bangPotMult_y{\oc u}\right)
                        + \left(\bangPotMult_y{t_2} + \max(1, \bangPotMult_x{t_2}) \cdot \bangPotMult_y{\oc u}\right)
                    \\
                        &=& \left(\bangPotMult_y{t_1} + \bangPotMult_y{t_2}\right)
                        + \left(\max(1, \bangPotMult_x{t_1}) + \max(1, \bangPotMult_x{t_2})\right) \cdot \bangPotMult_y{\oc u}
                    \\
                        &=& \bangPotMult_y{\app{t_1}{t_2}}
                        + \max(1, \bangPotMult_x{\app{t_1}{t_2}}) \cdot \bangPotMult_y{\oc u}
                    \\
                        &=& \bangPotMult_y{t\esub{x}{\oc u}}
                \end{array}
            \end{equation*}
        \end{itemize}

    \item[\bltII] $t = t_1\esub{z}{t_2}$: Then $t\isub{x}{u} =
        t_1\isub{x}{u}\esub{z}{t_2\isub{x}{u}}$ and by \ih on $t_1$
        and $t_2$, one has that $\bangPotMult_y{t_1\isub{x}{u}} \leq
        \bangPotMult_y{t_1\esub{x}{\oc u}}$ and
        $\bangPotMult_y{t_2\isub{x}{u}} \leq
        \bangPotMult_y{t_2\esub{x}{\oc u}}$. We distinguish four
        cases:
        \begin{itemize}
        \item[\bltIII] $x \notin \freeVar{t_1}$ and $x \notin
            \freeVar{t_2}$: Then $t\isub{x}{u} = t$, thus:
            \begin{equation*}
                \bangPotMult_y{t\isub{x}{u}}
                = \bangPotMult_y{t}
                \leq \bangPotMult_y{t} + \max(1, \bangPotMult_x{t}) \cdot \bangPotMult_y{\oc u}
                = \bangPotMult_y{t\esub{x}{\oc u}}
            \end{equation*}

        \item[\bltIII] $x \in \freeVar{t_1}$ and $x \notin
            \freeVar{t_2}$: Then $\bangPotMult_x{t_1} \geq 1$ and
            $\bangPotMult_x{t_2} = 0$. By \ih on $t_1$, one has that
            $\bangPotMult_y{t_1\isub{x}{u}} \leq
            \bangPotMult_y{t_1\esub{x}{\oc u}}$ and
            $\bangPotMult_z{t_1\isub{x}{u}} \leq
            \bangPotMult_z{t_1\esub{x}{\oc u}}$.  Moreover, by
            $\alpha$-conversion $z \notin \freeVar{t_2}$ and $z \notin
            \freeVar{u}$ thus $\bangPotMult_z{t_2} = 0$ and
            $\bangPotMult_z{\oc u} = 0$. Therefore:
            \begin{equation*}
                \begin{array}{rcl}
                    \bangPotMult_y{t\isub{x}{u}}
                    &=& \bangPotMult_y{t_1\isub{x}{u}\esub{z}{t_2\isub{x}{u}}}
                \\
                    &=& \bangPotMult_y{t_1\isub{x}{u}\esub{z}{t_2}}
                \\
                    &=& \bangPotMult_y{t_1\isub{x}{u}} + \max(1, \bangPotMult_z{t_1\isub{x}{u}}) \cdot \bangPotMult_y{t_2}
                \\
                    &\overset{\ih}{\leq}& \bangPotMult_y{t_1\esub{x}{\oc u}} + \max(1, \bangPotMult_z{t_1\esub{x}{\oc u}}) \cdot \bangPotMult_y{t_2}
                \\
                    &=& \bangPotMult_y{t_1} + \max(1, \bangPotMult_x{t_1}) \cdot \bangPotMult_y{\oc u}
                        + \max(1, \bangPotMult_z{t_1} + \max(1, \bangPotMult_x{t_1}) \cdot \bangPotMult_z{\oc u}) \cdot \bangPotMult_y{t_2}
                \\
                    &=& \bangPotMult_y{t_1} + \bangPotMult_x{t_1} \cdot \bangPotMult_y{\oc u}
                        + \max(1, \bangPotMult_z{t_1}) \cdot \bangPotMult_y{t_2}
                \\
                    &=& \bangPotMult_y{t_1} + \max(1, \bangPotMult_z{t_1}) \cdot \bangPotMult_y{t_2}
                    + \bangPotMult_x{t_1} \cdot \bangPotMult_y{\oc u}
                \\
                    &=& \bangPotMult_y{t_1} + \max(1, \bangPotMult_z{t_1}) \cdot \bangPotMult_y{t_2}
                    + \max(1, \bangPotMult_x{t_1} + \max(1, \bangPotMult_z{t_1}) \cdot \bangPotMult_x{t_2}) \cdot \bangPotMult_y{\oc u}
                \\
                    &=& \bangPotMult_y{t} + \max(1, \bangPotMult_x{t}) \cdot \bangPotMult_y{\oc u}
                \\
                    &=& \bangPotMult_y{t\esub{x}{\oc u}}
                \end{array}
            \end{equation*}

        \item[\bltIII] $x \notin \freeVar{t_1}$ and $x \in
            \freeVar{t_2}$: Then $\bangPotMult_x{t_1} = 0$ and
            $\bangPotMult_x{t_2} \geq 1$. By \ih on $t_2$, one has
            that $\bangPotMult_y{t_2\isub{x}{u}} \leq
            \bangPotMult_y{t_2\esub{x}{\oc u}}$. Moreover, by
            $\alpha$-conversion $z \notin \freeVar{t_2}$ thus
            $\bangPotMult_z{t_2} = 0$. Therefore:
            \begin{equation*}
                \begin{array}{rcl}
                    \bangPotMult_y{t\isub{x}{u}}
                    &=& \bangPotMult_y{t_1\isub{x}{u}\esub{z}{t_2\isub{x}{u}}}
                \\
                    &=& \bangPotMult_y{t_1\esub{z}{t_2\isub{x}{u}}}
                \\
                    &=& \bangPotMult_y{t_1} + \max(1, \bangPotMult_z{t_1}) \cdot \bangPotMult_y{t_2\isub{x}{u}}
                \\
                    &\leq& \bangPotMult_y{t_1} + \max(1, \bangPotMult_z{t_1}) \cdot \bangPotMult_y{t_2\esub{x}{\oc u}}
                \\
                    &=& \bangPotMult_y{t_1} + \max(1, \bangPotMult_z{t_1}) \cdot 
                    (\bangPotMult_y{t_2} + \max(1, \bangPotMult_x{t_2}) \cdot \bangPotMult_y{\oc u})
                \\
                    &=& \bangPotMult_y{t_1} + \max(1, \bangPotMult_z{t_1}) \cdot 
                    (\bangPotMult_y{t_2} + \bangPotMult_x{t_2} \cdot \bangPotMult_y{\oc u})
                \\
                    &=& \bangPotMult_y{t_1} + \max(1, \bangPotMult_z{t_1}) \cdot \bangPotMult_y{t_2}
                    + \max(1, \bangPotMult_z{t_1}) \cdot \bangPotMult_x{t_2} \cdot \bangPotMult_y{\oc u}
                \\
                    &=& \bangPotMult_y{t_1} + \max(1, \bangPotMult_z{t_1}) \cdot \bangPotMult_y{t_2}
                    + \max(1, \max(1, \bangPotMult_z{t_1}) \cdot \bangPotMult_x{t_2}) \cdot \bangPotMult_y{\oc u}
                \\
                    &=& \bangPotMult_y{t_1} + \max(1, \bangPotMult_z{t_1}) \cdot \bangPotMult_y{t_2}
                    + \max(1, \bangPotMult_x{t_1} + \max(1, \bangPotMult_z{t_1}) \cdot \bangPotMult_x{t_2}) \cdot \bangPotMult_y{\oc u}
                \\
                    &=& \bangPotMult_y{t} + \max(1, \bangPotMult_x{t}) \cdot \bangPotMult_y{\oc u}
                \\
                    &=& \bangPotMult_y{t\esub{x}{\oc u}}
                \end{array}
            \end{equation*}

        \item[\bltIII] $x \in \freeVar{t_1}$ and $x \in
            \freeVar{t_2}$: Then $\bangPotMult_x{t_1} \geq 1$ and
            $\bangPotMult_x{t_2} \geq 1$. By \ih on $t_1$, one has
            that $\bangPotMult_y{t_1\isub{x}{u}} \leq
            \bangPotMult_y{t_1\esub{x}{\oc u}}$ and
            $\bangPotMult_y{t_2\isub{x}{u}} \leq
            \bangPotMult_y{t_2\esub{x}{\oc u}}$. Moreover, by
            $\alpha$-conversion $z \notin \freeVar{t_2}$ and $z \notin
            \freeVar{\oc u}$ thus $\bangPotMult_z{t_2} = 0$ and
            $\bangPotMult_z{\oc u} = 0$. Therefore:
            \begin{equation*}
                \begin{array}{l}
                    \bangPotMult_y{t\isub{x}{u}}
                \\
                    \quad = \; \bangPotMult_y{t_1\isub{x}{u}\esub{z}{t_2\isub{x}{u}}}
                \\
                    \quad = \; \bangPotMult_y{t_1\isub{x}{u}} + \max(1, \bangPotMult_z{t_1\isub{x}{u}}) \cdot \bangPotMult_y{t_2\isub{x}{u}}
                \\
                    \quad \leq \; \bangPotMult_y{t_1\esub{x}{\oc u}} + \max(1, \bangPotMult_z{t_1\esub{x}{\oc u}}) \cdot \bangPotMult_y{t_2\esub{x}{\oc u}}
                \\
                    \quad = \; \bangPotMult_y{t_1} + \max(1, \bangPotMult_x{t_1}) \cdot \bangPotMult_y{\oc u}
                    + \max(1, \bangPotMult_z{t_1}
                    \\ \qquad \quad + \max(1, \bangPotMult_x{t_1}) \cdot \bangPotMult_z{\oc u})
                    \cdot (\bangPotMult_y{t_2} + \max(1, \bangPotMult_x{t_2}) \cdot \bangPotMult_y{\oc u})
                \\
                    \quad = \; \bangPotMult_y{t_1} + \bangPotMult_x{t_1} \cdot \bangPotMult_y{\oc u}
                    + \max(1, \bangPotMult_z{t_1}) \cdot (\bangPotMult_y{t_2} + \bangPotMult_x{t_2} \cdot \bangPotMult_y{\oc u})
                \\
                    \quad \leq \; \bangPotMult_y{t_1} + \bangPotMult_x{t_1} \cdot \bangPotMult_y{\oc u} \cdot \bangPotMult_x{t_2}
                    + \max(1, \bangPotMult_z{t_1}) \cdot (\bangPotMult_y{t_2} + \bangPotMult_x{t_2} \cdot \bangPotMult_y{\oc u})
                \\
                    \quad = \; \bangPotMult_y{t_1} + \max(1, \bangPotMult_z{t_1}) \cdot \bangPotMult_y{t_2}
                    + (\bangPotMult_x{t_1} + \max(1, \bangPotMult_z{t_1})) \cdot \bangPotMult_x{t_2} \cdot \bangPotMult_y{\oc u}
                \\
                    \quad = \; \bangPotMult_y{t_1} + \max(1, \bangPotMult_z{t_1}) \cdot \bangPotMult_y{t_2}
                    + \max(1, \bangPotMult_x{t_1} + \max(1, \bangPotMult_z{t_1}) \cdot \bangPotMult_x{t_2}) \cdot \bangPotMult_y{\oc u}
                \\
                    \quad = \; \bangPotMult_y{t} + \max(1, \bangPotMult_x{t}) \cdot \bangPotMult_y{\oc u}
                \\
                    \quad = \; \bangPotMult_y{t\esub{x}{\oc u}}
                \end{array}
            \end{equation*}
        \end{itemize}

    \item[\bltII] $t = \der{t'}$: Then $t\isub{x}{u} =
        \der{t'\isub{x}{\oc u}}$ and by \ih on $t'$, one has that
        $\bangPotMult_y{t'\isub{x}{u}} \leq
        \bangPotMult_y{t'\esub{x}{\oc u}}$. Thus:
        \begin{equation*}
            \begin{array}{rcl}
                \bangPotMult_y{t\isub{x}{u}}
                    &=& \bangPotMult_y{\der{t'\isub{x}{u}}}
                \\
                    &=& \bangPotMult_y{t'\isub{x}{u}}
                \\
                    &\leqih& \bangPotMult_y{t'\esub{x}{\oc u}}
                \\
                    &=& \bangPotMult_y{t'} + \max(1, \bangPotMult_x{t'}) \cdot \bangPotMult_y{\oc u}
                \\
                    &=& \bangPotMult_y{\der{t'}} + \max(1, \bangPotMult_x{\der{t'}}) \cdot \bangPotMult_y{\oc u}
                \\
                    &=& \bangPotMult_y{t\esub{x}{\oc u}}
            \end{array}
        \end{equation*}

    \item[\bltII] $t = \oc t'$: Then $t\isub{x}{u} = \oc
        (t'\isub{x}{u})$ and by \ih on $t'$, one has that
        $\bangPotMult_y{t'\isub{x}{u}} \leq
        \bangPotMult_y{t'\esub{x}{\oc u}}$. Thus:
        \begin{equation*}
            \begin{array}{rcl}
                \bangPotMult_y{t\isub{x}{u}}
                    &=& \bangPotMult_y{\oc (t'\isub{x}{u})}
                \\
                    &=& \bangPotMult_y{t'\isub{x}{u}}
                \\
                    &\leqih& \bangPotMult_y{t'\esub{x}{\oc u}}
                \\
                    &=& \bangPotMult_y{t'} + \max(1, \bangPotMult_x{t'}) \cdot \bangPotMult_y{\oc u}
                \\
                    &=& \bangPotMult_y{\oc t'} + \max(1, \bangPotMult_x{\oc t'}) \cdot \bangPotMult_y{\oc u}
                \\
                    &=& \bangPotMult_y{t\esub{x}{\oc u}}
            \end{array}
        \end{equation*}
    \end{itemize}

\item[\bltI] $\bangLCtxt = \bangLCtxt'\esub{z}{s}$: By \ih on
    $\bangLCtxt'$, one has that
    $\bangPotMult_y{\bangLCtxt'<t\isub{x}{u}>} \leq
    \bangPotMult_y{t\esub{x}{\bangLCtxt'<\oc u>}}$. By
    $\alpha$-conversion, $z \notin \freeVar{t}$ thus
    $\bangPotMult_z{t} = 0$ and therefore:

    \begin{equation*}
        \begin{array}{l}
            \bangPotMult_y{\bangLCtxt<t\isub{x}{u}>}
        \\
            \quad = \; \bangPotMult_y{\bangLCtxt'<t\isub{x}{u}>\esub{z}{s}}
        \\
            \quad = \; \bangPotMult_y{\bangLCtxt'<t\isub{x}{u}>} + \max(1, \bangPotMult_z{\bangLCtxt'<t\isub{x}{u}>}) \cdot \bangPotMult_y{s}
        \\
            \quad \overset{\ih}{\leq} \; \bangPotMult_y{t\esub{x}{\bangLCtxt'<\oc u>}} + \max(1, \bangPotMult_z{t\esub{x}{\bangLCtxt'<\oc u>}}) \cdot \bangPotMult_y{s}
        \\
            \quad = \; \bangPotMult_y{t} + \max(1, \bangPotMult_x{t}) \cdot \bangPotMult_y{\bangLCtxt'<\oc u>}
            + \max(1, \bangPotMult_z{t} + \max(1, \bangPotMult_x{t}) \cdot \bangPotMult_z{\bangLCtxt'<\oc u>}) \cdot \bangPotMult_y{s}
        \\
            \quad = \; \bangPotMult_y{t} + \max(1, \bangPotMult_x{t}) \cdot \bangPotMult_y{\bangLCtxt'<\oc u>}
            + \max(1, \max(1, \bangPotMult_x{t}) \cdot \bangPotMult_z{\bangLCtxt'<\oc u>}) \cdot \bangPotMult_y{s}
        \\
            \quad = \; \bangPotMult_y{t} + \max(1, \bangPotMult_x{t}) \cdot \bangPotMult_y{\bangLCtxt'<\oc u>}
            + \max(1, \bangPotMult_x{t}) \cdot \max(1, \bangPotMult_z{\bangLCtxt'<\oc u>}) \cdot \bangPotMult_y{s}
        \\
            \quad = \; \bangPotMult_y{t} + \max(1, \bangPotMult_x{t}) \cdot
            (\bangPotMult_y{\bangLCtxt'<\oc u>} + \max(1, \bangPotMult_z{\bangLCtxt'<\oc u>}) \cdot \bangPotMult_y{s})
        \\
            \quad = \; \bangPotMult_y{t} + \max(1, \bangPotMult_x{t}) \cdot \bangPotMult_y{\bangLCtxt'<\oc u>\esub{z}{s}}
        \\
            \quad = \; \bangPotMult_y{t} + \max(1, \bangPotMult_x{t}) \cdot \bangPotMult_y{\bangLCtxt<\oc u>}
        \\
            \quad = \; \bangPotMult_y{t\esub{x}{\bangLCtxt<\oc u>}}
        \end{array}
    \end{equation*}
\end{itemize}
\end{proof}
} \deliaLu \giulioLu%

\begin{lemma}
    \label{lem:Bang_MultiSize_Decreases_Substitution}%
    Let $s_1, s_2 \in \bangSetTerms$, then
    $\multiSize{s_1\esub{x}{\bangLCtxt<\oc s_2>}} \succ
    \multiSize{\bangLCtxt<s_1\isub{x}{s_2}>}$.
\end{lemma}
\stableProof{%
    \begin{proof}
By induction on $\bangLCtxt$. Cases:
\begin{itemize}
\item[\bltI] $\bangLCtxt = \Hole$: Then $t = s_1\esub{x}{\oc s_2}$ and
    $u = s_1\isub{x}{s_2}$. By induction on $s_1$:
    \begin{itemize}
    \item[\bltII] $s_1 = y$: Then:
        \begin{equation*}
            \begin{array}{rcl}
                \multiSize{s_1\esub{x}{\oc s_2}}
                    &=& \mset{\bangPotMult_x{s_1}} \uplus \multiSize{s_1} \uplus \max(1, \bangPotMult_x{s_1}) \multisetScProd \multiSize{\oc s_2}
                \\
                    &=& \mset{\bangPotMult_x{y}} \uplus \multiSize{y} \uplus \max(1, \bangPotMult_x{y}) \multisetScProd \multiSize{\oc s_2}
                \\
                &=& \mset{\bangPotMult_x{y}} \uplus \multiSize{\oc s_2}\\
                &=& \mset{\bangPotMult_x{y}} \uplus \multiSize{ s_2}
            \end{array}
        \end{equation*}
        We distinguish two cases:
        \begin{itemize}
        \item[\bltIII] $x = y$: Then $\multiSize{t} =
            \mset{\bangPotMult_x{y}} \uplus \multiSize{s_2} = \mset{1}
            \uplus \multiSize{ s_2}$ and $s_1\isub{x}{s_2} = s_2$ thus
            $\multiSize{u} = \multiSize{s_1\isub{x}{s_2}} =
            \multiSize{ s_2} $ and therefore $\multiSize{t} \succ
            \multiSize{u}$.

        \item[\bltIII] $x \neq y$: Then $\multiSize{t} =
            \mset{\bangPotMult_x{y}} \uplus \multiSize{s_2} = \mset{0}
            \uplus \multiSize{s_2}$ and $s_1\isub{x}{s_2} = y$ thus
            $\multiSize{u} = \multiSize{s_1\isub{x}{s_2}} = \emptymset$
            and therefore $\multiSize{t} \succ \multiSize{u}$.
        \end{itemize}

    \item[\bltII] $s_1 = \abs{y}{s'_1}$: 
    We can suppose without loss of generality that $y \notin \freeVar{s_2} \cup \{x\}$,
    thus $s_1\isub{x}{s_2} = \abs{y}{(s'_1\isub{x}{s_2})}$.
        By \ih on $s'_1$, $\multiSize{s'_1\esub{x}{\oc s_2}} \succ
        \multiSize{s'_1\isub{x}{s_2}}$. 
        Therefore:
        \begin{equation*}
            \begin{array}{rcl}
                \multiSize{s_1\esub{x}{\oc s_2}}
                    &=& \multiSize{(\abs{y}{s'_1})\esub{x}{ \oc s_2}}
                \\
                    &=& \mset{\bangPotMult_x{\abs{y}{s'_1}}} \uplus \multiSize{\abs{y}{s'_1}} \uplus \max(1, \bangPotMult_x{\abs{y}{s'_1}}) \multisetScProd \multiSize{\oc  s_2}
                \\
                    &=& \mset{\bangPotMult_x{s'_1}} \uplus \multiSize{s'_1} \uplus \max(1, \bangPotMult_x{s'_1}) \multisetScProd \multiSize{ \oc s_2}
                \\
                    &=& \multiSize{s'_1\esub{x}{\oc s_2}}
                \\
                    &\overset{\ih}{\succ}& \multiSize{s'_1\isub{x}{s_2}}
                \\
                    &=& \multiSize{\abs{y}{(s'_1\isub{x}{s_2})}}
                    \quad=\quad \multiSize{s_1\isub{x}{s_2}}
            \end{array}
        \end{equation*}

    \item[\bltII] $s_1 = \app{s_1^1}{s_1^2}$: Then $s_1\isub{x}{s_2} =
        \app{(s_1^1\isub{x}{s_2})}{(s_1^2\isub{x}{s_2})}$.
        By \ih
        on $s_1^1$ and $s_1^2$, 
        $\multiSize{s_1^1\esub{x}{\oc s_2}} \succ
        \multiSize{s_1^1\isub{x}{s_2}}$ and
        $\multiSize{s_1^2\esub{x}{\oc s_2}} \succ
        \multiSize{s_1^2\isub{x}{s_2}}$. We distinguish four cases:
        \begin{itemize}
        \item[\bltIII] $x \notin \freeVar{s_1^1}$ and $x \notin
            \freeVar{s_1^2}$: Then $\bangPotMult_x{s_1^1} = 0$,
            $\bangPotMult_x{s_1^2} = 0$, $s_1^1\isub{x}{s_2} = s_1^1$
            and $s_1^1\isub{x}{s_2} = s_1^1$. Thus:
            \begin{equation*}
                \begin{array}{rcl}
                    \multiSize{s_1\esub{x}{\oc  s_2}}
                        &=& \multiSize{(\app{s_1^1}{s_1^2})\esub{x}{\oc  s_2}}
                    \\
                        &=& \mset{\bangPotMult_x{\app{s_1^1}{s_1^2}}} \uplus \multiSize{\app{s_1^1}{s_1^2}} \uplus \max(1, \bangPotMult_x{\app{s_1^1}{s_1^2}}) \multisetScProd \multiSize{\oc  s_2}
                    \\
                        &=& \mset{0} \uplus \multiSize{s_1^1} \uplus \multiSize{s_1^2} \uplus \multiSize{\oc s_2}
                    \\
                        &\succ& \multiSize{s_1^1} \uplus \multiSize{s_1^2}
                    \\
                        &=& \multiSize{s_1^1\isub{x}{s_1}} \uplus \multiSize{s_1^2\isub{x}{s_1}}
                    \\
                        &=& \multiSize{\app{(s_1^1\isub{x}{s_2})}{(s_1^2\isub{x}{s_2})}}
                        \quad=\quad \multiSize{s_1\isub{x}{s_2}}
                \end{array}
            \end{equation*}

        \item[\bltIII] $x \in \freeVar{s_1^1}$ and $x \notin
            \freeVar{s_1^2}$: Then $\bangPotMult_x{s_1^1} \geq 1$,
            $\bangPotMult_x{s_1^2} = 0$ and $s_1^2\isub{x}{s_2} =
            s_1^2$. Thus:
            \begin{equation*}
                \begin{array}{rcl}
                    \multiSize{s_1\esub{x}{\oc s_2}}
                        &=& \multiSize{(\app{s_1^1}{s_1^2})\esub{x}{\oc  s_2}}
                    \\
                        &=& \mset{\bangPotMult_x{\app{s_1^1}{s_1^2}}} \uplus \multiSize{\app{s_1^1}{s_1^2}} \uplus \max(1, \bangPotMult_x{\app{s_1^1}{s_1^2}}) \multisetScProd \multiSize{\oc  s_2}
                    \\
                        &=& \mset{\bangPotMult_x{s_1^1}} \uplus \multiSize{s_1^1} \uplus \multiSize{s_1^2} \uplus \max(1, \bangPotMult_x{s_1^1}) \multisetScProd \multiSize{\oc s_2}
                    \\
                        &=& \multiSize{s_1^1\esub{x}{\oc s_2}} \uplus \multiSize{s_1^2}
                    \\
                        &\overset{\ih}{\succ}& \multiSize{s_1^1\isub{x}{s_2}} \uplus \multiSize{s_1^2}
                    \\
                        &=& \multiSize{s_1^1\isub{x}{s_1}} \uplus \multiSize{s_1^2\isub{x}{s_1}}
                    \\
                        &=& \multiSize{\app{(s_1^1\isub{x}{s_2})}{(s_1^2\isub{x}{s_2})}}
                        \quad=\quad \multiSize{s_1\isub{x}{s_2}}
                \end{array}
            \end{equation*}

        \item[\bltIII] $x \notin \freeVar{s_1^1}$ and $x \in
            \freeVar{s_1^2}$: Then $\bangPotMult_x{s_1^1} = 0$,
            $\bangPotMult_x{s_1^2} \geq 1$ and $s_1^1\isub{x}{s_2} =
            s_1^2$. Thus:
            \begin{equation*}
                \begin{array}{rcl}
                    \multiSize{s_1\esub{x}{\oc s_2}}
                        &=& \multiSize{(\app{s_1^1}{s_1^2})\esub{x}{\oc s_2}}
                    \\
                        &=& \mset{\bangPotMult_x{\app{s_1^1}{s_1^2}}} \uplus \multiSize{\app{s_1^1}{s_1^2}} \uplus \max(1, \bangPotMult_x{\app{s_1^1}{s_1^2}}) \multisetScProd \multiSize{\oc s_2}
                    \\
                        &=& \mset{\bangPotMult_x{s_1^2}} \uplus \multiSize{s_1^1} \uplus \multiSize{s_1^2} \uplus \max(1, \bangPotMult_x{s_1^2}) \multisetScProd \multiSize{\oc s_2}
                    \\
                        &=& \multiSize{s_1^1} \uplus \multiSize{s_1^2\esub{x}{\oc s_2}}
                    \\
                        &\overset{\ih}{\succ}& \multiSize{s_1^1} \uplus \multiSize{s_1^2\isub{x}{s_2}}
                    \\
                        &=& \multiSize{s_1^1\isub{x}{s_1}} \uplus \multiSize{s_1^2\isub{x}{s_1}}
                    \\
                        &=& \multiSize{\app{(s_1^1\isub{x}{s_2})}{(s_1^2\isub{x}{s_2})}}
                        \quad=\quad \multiSize{s_1\isub{x}{s_2}}
                \end{array}
            \end{equation*}

        \item[\bltIII] $x \in \freeVar{s_1^1}$ and $x \in
            \freeVar{s_1^2}$: Then $\bangPotMult_x{s_1^1} \geq 1$ and
            $\bangPotMult_x{s_1^2} \geq 1$. Thus:
            \begin{equation*}
                \begin{array}{rcl}
                    \multiSize{s_1\esub{x}{\oc s_2}}
                        &=& \multiSize{(\app{s_1^1}{s_1^2})\esub{x}{\oc s_2}}
                    \\
                        &=& \mset{\bangPotMult_x{\app{s_1^1}{s_1^2}}} \uplus \multiSize{\app{s_1^1}{s_1^2}} \uplus \max(1, \bangPotMult_x{\app{s_1^1}{s_1^2}}) \multisetScProd \multiSize{\oc s_2}
                    \\
                        &=& \mset{\bangPotMult_x{s_1^1} + \bangPotMult_x{s_1^2}} \uplus \multiSize{s_1^1} \uplus \multiSize{s_1^2} \uplus (\bangPotMult_x{s_1^1} + \bangPotMult_x{s_1^2}) \multisetScProd \multiSize{\oc s_2}
                    \\
                        &\succ& \mset{\bangPotMult_x{s_1^1}} \uplus \mset{\bangPotMult_x{s_1^2}} \uplus \multiSize{s_1^1} \uplus \multiSize{s_1^2} \uplus (\bangPotMult_x{s_1^1} \multisetScProd \multiSize{\oc s_2}) \uplus (\bangPotMult_x{s_1^2} \multisetScProd \multiSize{\oc s_2})
                    \\
                        &=& \multiSize{s_1^1\esub{x}{\oc s_2}} \uplus \multiSize{s_1^2\esub{x}{\oc s_2}}
                    \\
                        &\overset{\ih}{\succ}& \multiSize{s_1^1\isub{x}{s_2}} \uplus \multiSize{s_1^2\isub{x}{s_2}}
                    \\
                        &=& \multiSize{\app{(s_1^1\isub{x}{s_2})}{(s_1^2\isub{x}{s_2})}}
                        \quad=\quad \multiSize{s_1\isub{x}{s_2}}
                \end{array}
            \end{equation*}
        \end{itemize}

    \item[\bltII] $s_1 = s_1^1\esub{y}{s_1^2}$: 
    We can suppose without loss of generality that  $y \notin \freeVar{s_2} \cup\{x\}$, then
        $s_1\isub{x}{s_2} =
        (s_1^1\isub{x}{s_2})\esub{y}{(s_1^2\isub{x}{s_2})}$ and by \ih
        on $s_1^1$ and $s_1^2$, one has that
        $\multiSize{s_1^1\esub{x}{\oc s_2}} \succ
        \multiSize{s_1^1\isub{x}{s_2}}$ and
        $\multiSize{s_1^2\esub{x}{\oc s_2}} \succ
        \multiSize{s_1^2\isub{x}{s_2}}$. We distinguish four cases:
        \begin{itemize}
        \item[\bltIII] $x \notin \freeVar{s_1^1}$ and $x \notin
            \freeVar{s_1^2}$: Then $\bangPotMult_x{s_1^1} = 0$,
            $\bangPotMult_x{s_1^2} = 0$, $s_1^1\isub{x}{s_2} = s_1^1$
            and $s_1^2\isub{x}{s_2} = s_1^2$. Thus:
            \begin{equation*}
                \begin{array}{rcl}
                    \multiSize{s_1\esub{x}{\oc s_2}}
                        &=& \multiSize{(s_1^1\esub{y}{s_1^2})\esub{x}{\oc s_2}}
                    \\
                        &=& \mset{\bangPotMult_x{s_1^1\esub{y}{s_1^2}}}
                            \uplus \multiSize{s_1^1\esub{y}{s_1^2}}
                            \uplus \max(1, \bangPotMult_x{s_1^1\esub{y}{s_1^2}}) \multisetScProd \multiSize{\oc s_2}
                    \\
                        &=& \mset{0}
                            \uplus \mset{\bangPotMult_y{s_1^1}}
                            \uplus \multiSize{s_1^1}
                            \uplus \max(1, \bangPotMult_y{s_1^1}) \multisetScProd \multiSize{s_1^2}\uplus \multiSize{\oc s_2}
                    \\
                        &=& \mset{0}
                            \uplus \mset{\bangPotMult_y{s_1^1\isub{x}{s_2}}}
                            \uplus \multiSize{s_1^1\isub{x}{s_2}}
                            \uplus \max(1, \bangPotMult_y{s_1^1\isub{x}{s_2}}) \multisetScProd \multiSize{s_1^2\isub{x}{s_2}}
                            \uplus \multiSize{\oc s_2}
                    \\
                        &\succ& \mset{\bangPotMult_y{s_1^1\isub{x}{s_2}}}
                            \uplus \multiSize{s_1^1\isub{x}{s_2}}
                            \uplus \max(1, \bangPotMult_y{s_1^1\isub{x}{s_2}}) \multisetScProd \multiSize{s_1^2\isub{x}{s_2}}
                    \\
                        &=& \multiSize{s_1^1\isub{x}{s_2}\esub{y}{s_1^2\isub{x}{s_2}}}
                        \quad=\quad \multiSize{s_1\isub{x}{s_2}}
                \end{array}
            \end{equation*}

        \item[\bltIII] $x \in \freeVar{s_1^1}$ and $x \notin
            \freeVar{s_1^2}$: Then $\bangPotMult_x{s_1^1} \geq 1$,
            $\bangPotMult_x{s_1^2} = 0$ and $s_1^2\isub{x}{s_2} =
            s_1^2$. Then using
            \Cref{lem:Bang_Substitution_compat_PotMult}, one has that:
            \begin{equation*}
                \begin{array}{rcl}
                    \multiSize{s_1\esub{x}{\oc s_2}}
                        &=& \multiSize{(s_1^1\esub{y}{s_1^2})\esub{x}{\oc s_2}}
                    \\
                        &=& \mset{\bangPotMult_x{s_1^1\esub{y}{s_1^2}}} \uplus \multiSize{s_1^1\esub{y}{s_1^2}} \uplus \max(1, \bangPotMult_x{s_1^1\esub{y}{s_1^2}}) \multisetScProd \multiSize{\oc s_2}
                    \\
                        &=& \mset{\bangPotMult_x{s_1^1} + \max(1, \bangPotMult_y{s_1^1}) \cdot \bangPotMult_x{s_1^2}}
                            \uplus \mset{\bangPotMult_y{s_1^1}}
                            \uplus \multiSize{s_1^1}
                            \uplus \max(1, \bangPotMult_y{s_1^1}) \multisetScProd \multiSize{s_1^2}
                        \\ && \hspace{0.5cm}
                            \uplus \max(1, \bangPotMult_x{s_1^1} + \max(1, \bangPotMult_y{s_1^1}) \cdot \bangPotMult_x{s_1^2}) \multisetScProd \multiSize{\oc s_2}
                    \\
                        &=& \mset{\bangPotMult_x{s_1^1}}
                            \uplus \mset{\bangPotMult_y{s_1^1}}
                            \uplus \multiSize{s_1^1}
                            \uplus \max(1, \bangPotMult_y{s_1^1}) \multisetScProd \multiSize{s_1^2}
                            \uplus \max(1, \bangPotMult_x{s_1^1}) \multisetScProd \multiSize{\oc s_2}
                    \\
                        &=& \mset{\bangPotMult_y{s_1^1}}
                            \uplus \left(\mset{\bangPotMult_x{s_1^1}}
                            \uplus \multiSize{s_1^1}
                            \uplus \max(1, \bangPotMult_x{s_1^1}) \multisetScProd \multiSize{\oc s_2} \right)
                            \uplus \max(1, \bangPotMult_y{s_1^1}) \multisetScProd \multiSize{s_1^2}
                    \\
                        &=& \mset{\bangPotMult_y{s_1^1}}
                            \uplus \multiSize{s_1^1\esub{x}{\oc s_2}}
                            \uplus \max(1, \bangPotMult_y{s_1^1}) \multisetScProd \multiSize{s_1^2}
                    \\
                        &\overset{\ih}{\succ}& \mset{\bangPotMult_y{s_1^1}}
                            \uplus \multiSize{s_1^1\isub{x}{s_2}}
                            \uplus \max(1, \bangPotMult_y{s_1^1}) \multisetScProd \multiSize{s_1^2}
                    \\
                        &=& \mset{\bangPotMult_y{s_1^1\isub{x}{s_2}}}
                            \uplus \multiSize{s_1^1\isub{x}{s_2}}
                            \uplus \max(1, \bangPotMult_y{s_1^1\isub{x}{s_2}}) \multisetScProd \multiSize{s_1^2\isub{x}{s_2}}
                    \\
                        &=& \multiSize{s_1^1\isub{x}{s_2}\esub{y}{s_1^2\isub{x}{s_2}}}
                        \quad=\quad \multiSize{s_1\isub{x}{s_2}}
                \end{array}
            \end{equation*}

        \item[\bltIII] $x \notin \freeVar{s_1^1}$ and $x \in
            \freeVar{s_1^2}$: Then $\bangPotMult_x{s_1^1} = 0$,
            $\bangPotMult_x{s_1^2} \geq 1$ and $s_1^1\isub{x}{s_2} =
            s_1^2$. We distinguish two cases:
            \begin{itemize}
            \item[\bltIV] $y \notin \freeVar{s_1^1}$: Then
                $\bangPotMult_y{s_1^1} = 0$. By
                \Cref{lem:Bang_Substitution_compat_PotMult},
                $\bangPotMult_y{s_1^1\isub{x}{s_2}} = 0$, thus:
                \begin{equation*}
                    \begin{array}{rcl}
                        \multiSize{s_1\esub{x}{\oc s_2}}
                            &=& \multiSize{(s_1^1\esub{y}{s_1^2})\esub{x}{\oc s_2}}
                        \\
                            &=& \mset{\bangPotMult_x{s_1^1\esub{y}{s_1^2}}} \uplus \multiSize{s_1^1\esub{y}{s_1^2}} \uplus \max(1, \bangPotMult_x{s_1^1\esub{y}{s_1^2}}) \multisetScProd \multiSize{\oc s_2}
                        \\
                            &=& \mset{\bangPotMult_x{s_1^1} + \max(1, \bangPotMult_y{s_1^1}) \cdot \bangPotMult_x{s_1^2}}
                                \uplus \mset{\bangPotMult_y{s_1^1}}
                                \uplus \multiSize{s_1^1}
                                \uplus \max(1, \bangPotMult_y{s_1^1}) \multisetScProd \multiSize{s_1^2}
                            \\ && \hspace{0.5cm}
                                \uplus \max(1, \bangPotMult_x{s_1^1} + \max(1, \bangPotMult_y{s_1^1}) \cdot \bangPotMult_x{s_1^2}) \multisetScProd \multiSize{\oc s_2}
                        \\
                            &=& \mset{\bangPotMult_x{s_1^2}}
                                \uplus \mset{\bangPotMult_y{s_1^1}}
                                \uplus \multiSize{s_1^1}
                                \uplus \max(1, \bangPotMult_y{s_1^1}) \multisetScProd \multiSize{s_1^2}
                                \uplus \max(1, \bangPotMult_x{s_1^2}) \multisetScProd \multiSize{\oc s_2}
                        \\
                            &=& \mset{\bangPotMult_x{s_1^2}}
                                \uplus \mset{\bangPotMult_y{s_1^1\isub{x}{s_2}}}
                                \uplus \multiSize{s_1^1\isub{x}{s_2}}
                                \uplus \max(1, \bangPotMult_y{s_1^1\isub{x}{s_2}}) \multisetScProd \multiSize{s_1^2}
                            \\ && \hspace{0.5cm}
                                \uplus \max(1, \bangPotMult_x{s_1^2}) \multisetScProd \multiSize{\oc s_2}
                        \\
                            &\succ& \mset{\bangPotMult_y{s_1^1\isub{x}{s_2}}}
                                \uplus \multiSize{s_1^1\isub{x}{s_2}}
                                \uplus \left( \mset{\bangPotMult_x{s_1^2}}
                                \uplus \multiSize{s_1^2}
                                \uplus \max(1, \bangPotMult_x{s_1^2}) \multisetScProd\multiSize{\oc s_2} \right)
                        \\
                            &=& \mset{\bangPotMult_y{s_1^1\isub{x}{s_2}}}
                                \uplus \multiSize{s_1^1\isub{x}{s_2}}
                                \uplus \multiSize{s_1^2\esub{x}{\oc s_2}}
                        \\
                            &=& \mset{\bangPotMult_y{s_1^1\isub{x}{s_2}}}
                                \uplus \multiSize{s_1^1\isub{x}{s_2}}
                                \uplus \multiSize{s_1^2\esub{x}{\oc s_2}}
                        \\
                            &\overset{\ih}{\succ}& \mset{\bangPotMult_y{s_1^1\isub{x}{s_2}}}
                                \uplus \multiSize{s_1^1\isub{x}{s_2}}
                                \uplus \multiSize{s_1^2\isub{x}{s_2}}
                        \\
                            &=& \mset{\bangPotMult_y{s_1^1\isub{x}{s_2}}}
                                \uplus \multiSize{s_1^1\isub{x}{s_2}}
                                \uplus \max(1, \bangPotMult_y{s_1^1\isub{x}{s_2}}) \multisetScProd \multiSize{s_1^2\isub{x}{s_2}}
                        \\
                            &=& \multiSize{s_1^1\isub{x}{s_2}\esub{y}{s_1^2\isub{x}{s_2}}}
                            \quad=\quad \multiSize{s_1\isub{x}{s_2}}
                    \end{array}
                \end{equation*}

            \item[\bltIV] $y \in \freeVar{s_1^1}$: Then
                $\bangPotMult_y{s_1^1} \geq 1$. By
                \Cref{lem:Bang_Substitution_compat_PotMult},
                $\bangPotMult_y{s_1^1\isub{x}{s_2}} =
                \bangPotMult_y{s_1^1}$, thus:
                \begin{equation*}
                    \begin{array}{rcl}
                        \multiSize{s_1\esub{x}{\oc s_2}}
                            &=& \multiSize{(s_1^1\esub{y}{s_1^2})\esub{x}{\oc s_2}}
                        \\
                            &=& \mset{\bangPotMult_x{s_1^1\esub{y}{s_1^2}}} \uplus \multiSize{s_1^1\esub{y}{s_1^2}} \uplus \max(1, \bangPotMult_x{s_1^1\esub{y}{s_1^2}}) \multisetScProd \multiSize{\oc s_2}
                        \\
                            &=& \mset{\bangPotMult_x{s_1^1} + \max(1, \bangPotMult_y{s_1^1}) \cdot \bangPotMult_x{s_1^2}}
                                \uplus \mset{\bangPotMult_y{s_1^1}}
                                \uplus \multiSize{s_1^1}
                                \uplus \max(1, \bangPotMult_y{s_1^1}) \multisetScProd \multiSize{s_1^2}
                            \\ && \hspace{0.5cm}
                                \uplus \max(1, \bangPotMult_x{s_1^1} + \max(1, \bangPotMult_y{s_1^1}) \cdot \bangPotMult_x{s_1^2}) \multisetScProd \multiSize{\oc s_2}
                        \\
                            &=& \mset{\bangPotMult_y{s_1^1} \cdot \bangPotMult_x{s_1^2}}
                                \uplus \mset{\bangPotMult_y{s_1^1}}
                                \uplus \multiSize{s_1^1}
                                \uplus \max(1, \bangPotMult_y{s_1^1}) \multisetScProd \multiSize{s_1^2}
                            \\ && \hspace{0.5cm}
                                \uplus \max(1, \bangPotMult_y{s_1^1} \cdot \bangPotMult_x{s_1^2}) \multisetScProd \multiSize{\oc s_2}
                        \\
                            &=& \mset{\bangPotMult_y{s_1^1} \cdot \bangPotMult_x{s_1^2}}
                                \uplus \mset{\bangPotMult_y{s_1^1\isub{x}{s_2}}}
                                \uplus \multiSize{s_1^1\isub{x}{s_2}}
                                \uplus \max(1, \bangPotMult_y{s_1^1\isub{x}{s_2}}) \multisetScProd \multiSize{s_1^2}
                            \\ && \hspace{0.5cm}
                                \uplus \max(1, \bangPotMult_y{s_1^1} \cdot \bangPotMult_x{s_1^2}) \multisetScProd \multiSize{\oc s_2}
                        \\
                            &=& \bangPotMult_y{s_1^1} \multisetScProd \mset{\bangPotMult_x{s_1^2}
                                \uplus \mset{\bangPotMult_y{s_1^1\isub{x}{s_2}}}
                                \uplus \multiSize{s_1^1\isub{x}{s_2}}
                                \uplus \bangPotMult_y{s_1^1} \multisetScProd \multiSize{s_1^2}}
                            \\ && \hspace{0.5cm}
                                \uplus\; \bangPotMult_y{s_1^1} \multisetScProd (\max(1, \bangPotMult_x{s_1^2}) \multisetScProd \multiSize{\oc s_2})
                        \\
                            &=& \mset{\bangPotMult_y{s_1^1\isub{x}{s_2}}}
                                \uplus \multiSize{s_1^1\isub{x}{s_2}}
                                \uplus \bangPotMult_y{s_1^1} \multisetScProd \left( \mset{\bangPotMult_x{s_1^2}}
                                \uplus \multiSize{s_1^2}
                                \uplus \max(1, \bangPotMult_x{s_1^2}) \multisetScProd\multiSize{\oc s_2} \right)
                        \\
                            &=& \mset{\bangPotMult_y{s_1^1\isub{x}{s_2}}}
                                \uplus \multiSize{s_1^1\isub{x}{s_2}}
                                \uplus \bangPotMult_y{s_1^1} \multisetScProd \multiSize{s_1^2\esub{x}{\oc s_2}}
                        \\
                            &=& \mset{\bangPotMult_y{s_1^1\isub{x}{s_2}}}
                                \uplus \multiSize{s_1^1\isub{x}{s_2}}
                                \uplus \bangPotMult_y{s_1^1} \multisetScProd \multiSize{s_1^2\esub{x}{\oc s_2}}
                        \\
                            &\overset{\ih}{\succ}& \mset{\bangPotMult_y{s_1^1\isub{x}{s_2}}}
                                \uplus \multiSize{s_1^1\isub{x}{s_2}}
                                \uplus \bangPotMult_y{s_1^1} \multisetScProd \multiSize{s_1^2\isub{x}{s_2}}
                        \\
                            &=& \mset{\bangPotMult_y{s_1^1\isub{x}{s_2}}}
                                \uplus \multiSize{s_1^1\isub{x}{s_2}}
                                \uplus \max(1, \bangPotMult_y{s_1^1\isub{x}{s_2}}) \multisetScProd \multiSize{s_1^2\isub{x}{s_2}}
                        \\
                            &=& \multiSize{s_1^1\isub{x}{s_2}\esub{y}{s_1^2\isub{x}{s_2}}}
                            \quad=\quad \multiSize{s_1\isub{x}{s_2}}
                    \end{array}
                \end{equation*}
            \end{itemize}

        \item[\bltIII] $x \in \freeVar{s_1^1}$ and $x \in
            \freeVar{s_1^2}$: Then $\bangPotMult_x{s_1^1} \geq 1$ and
            $\bangPotMult_x{s_1^2} \geq 1$. We distinguish two
            cases:
            \begin{itemize}
            \item[\bltIV] $y \notin \freeVar{s_1^1}$: Then
                $\bangPotMult_y{s_1^1} = 0$. By
                \Cref{lem:Bang_Substitution_compat_PotMult},
                $\bangPotMult_y{s_1^1\isub{x}{s_2}} =
                \bangPotMult_y{s_1^1}$ thus one has that:
                \begin{equation*}
                    \begin{array}{l}
                        \multiSize{s_1\esub{x}{\oc s_2}}
                        \\
                            \quad= \multiSize{(s_1^1\esub{y}{s_1^2})\esub{x}{\oc s_2}}
                        \\
                            \quad= \mset{\bangPotMult_x{s_1^1\esub{y}{s_1^2}}}
                                \uplus \multiSize{s_1^1\esub{y}{s_1^2}}
                                \uplus \max(1, \bangPotMult_x{s_1^1\esub{y}{s_1^2}}) \multisetScProd \multiSize{\oc s_2}
                        \\
                            \quad= \mset{\bangPotMult_x{s_1^1} + \max(1, \bangPotMult_y{s_1^1})} \cdot \bangPotMult_x{s_1^2}
                                \uplus \mset{\bangPotMult_y{s_1^1}}
                                \uplus \multiSize{s_1^1} \uplus \max(1, \bangPotMult_y{s_1^1}) \multisetScProd \multiSize{s_1^2}
                            \\
                            \quad\;\;\quad
                                \uplus \max(1, \bangPotMult_x{s_1^1} + \max(1, \bangPotMult_y{s_1^1}) \cdot \bangPotMult_x{s_1^2}) \multisetScProd \multiSize{\oc s_2}
                        \\
                            \quad= \mset{\bangPotMult_x{s_1^1}}
                                \uplus \mset{\bangPotMult_x{s_1^2}}
                                \uplus \mset{\bangPotMult_y{s_1^1}}
                                \uplus \multiSize{s_1^1}
                                \uplus \multiSize{s_1^2}
                                \uplus \max(1, \bangPotMult_x{s_1^1}) \multisetScProd \multiSize{\oc s_2}
                                \uplus \max(1, \bangPotMult_x{s_1^2}) \multisetScProd \multiSize{\oc s_2}
                        \\
                            \quad= \mset{\bangPotMult_y{s_1^1}}
                                \uplus \left(\mset{\bangPotMult_x{s_1^1}}
                                \uplus \multiSize{s_1^1}
                                \uplus \max(1, \bangPotMult_x{s_1^1}) \multisetScProd \multiSize{\oc s_2}\right)
                                \uplus \left(\mset{\bangPotMult_x{s_1^2}}
                                \uplus \multiSize{s_1^2}
                                \uplus (\max(1, \bangPotMult_x{s_1^2}) \multisetScProd \multiSize{ \oc s_2})\right)
                        \\
                            \quad= \mset{\bangPotMult_y{s_1^1}}
                                \uplus \multiSize{s_1^1\esub{x}{\oc s_2}}
                                \uplus \multiSize{s_1^2\esub{x}{\oc s_2}}
                        \\
                            \quad\overset{\ih}{\succ} \mset{\bangPotMult_y{s_1^1}}
                                \uplus \multiSize{s_1^1\isub{x}{s_2}}
                                \uplus \multiSize{s_1^1\isub{x}{s_2}}
                        \\
                            \quad= \mset{\bangPotMult_y{s_1^1}}
                                \uplus \multiSize{s_1^1\isub{x}{s_2}}
                                \uplus \max(1, \bangPotMult_y{s_1^1\isub{x}{s_2}}) \multisetScProd \multiSize{s_1^1\isub{x}{s_2}}
                        \\
                                \quad= \multiSize{s_1^1\isub{x}{s_2}\esub{y}{s_1^2\isub{x}{s_2}}}
                                \quad=\quad \multiSize{s_1\isub{x}{s_2}}
                    \end{array}
                \end{equation*}

            \item[\bltIV] $y \in \freeVar{s_1^1}$: Then
                $\bangPotMult_y{s_1^1} \geq 1$. By
                \Cref{lem:Bang_Substitution_compat_PotMult},
                $\bangPotMult_y{s_1^1\isub{x}{s_2}} =
                \bangPotMult_y{s_1^1}$ thus one has that:
                \begin{equation*}
                    \begin{array}{l}
                        \multiSize{s_1\esub{x}{\oc s_2}}
                        \\
                            \quad= \multiSize{(s_1^1\esub{y}{s_1^2})\esub{x}{\oc s_2}}
                        \\
                            \quad= \mset{\bangPotMult_x{s_1^1\esub{y}{s_1^2}}}
                                \uplus \multiSize{s_1^1\esub{y}{s_1^2}}
                                \uplus \max(1, \bangPotMult_x{s_1^1\esub{y}{s_1^2}}) \multisetScProd \multiSize{\oc s_2}
                        \\
                            \quad= \mset{\bangPotMult_x{s_1^1} + \max(1, \bangPotMult_y{s_1^1})} \cdot \bangPotMult_x{s_1^2}
                                \uplus \mset{\bangPotMult_y{s_1^1}}
                                \uplus \multiSize{s_1^1} \uplus \max(1, \bangPotMult_y{s_1^1}) \multisetScProd \multiSize{s_1^2}
                            \\ 
                             \quad\;\;\quad
                                \uplus \max(1, \bangPotMult_x{s_1^1} + \max(1, \bangPotMult_y{s_1^1}) \cdot \bangPotMult_x{s_1^2}) \multisetScProd \multiSize{\oc s_2}
                        \\
                           \quad= \mset{\bangPotMult_x{s_1^1}}
                                \uplus \bangPotMult_y{s_1^1} \multisetScProd \mset{\bangPotMult_x{s_1^2}}
                                \uplus \mset{\bangPotMult_y{s_1^1}}
                                \uplus \multiSize{s_1^1}
                                \uplus \bangPotMult_y{s_1^1} \multisetScProd \multiSize{s_1^2}
                            \\
                            \quad\;\;\quad
                                \uplus \max(1, \bangPotMult_x{s_1^1}) \multisetScProd \multiSize{\oc s_2}
                                \uplus \bangPotMult_y{s_1^1} \multisetScProd (\max(1, \bangPotMult_x{s_1^2}) \multisetScProd \multiSize{\oc  s_2})
                        \\
                            \quad= \mset{\bangPotMult_y{s_1^1}}
                                \uplus \left(\mset{\bangPotMult_x{s_1^1}}
                                \uplus \multiSize{s_1^1}
                                \uplus \max(1, \bangPotMult_x{s_1^1}) \multisetScProd \multiSize{\oc s_2}\right)
                        \\
                            \quad\;\;\quad
                                \uplus\; \bangPotMult_y{s_1^1} \multisetScProd \left(\mset{\bangPotMult_x{s_1^2}}
                                \uplus \multiSize{s_1^2}
                                \uplus (\max(1, \bangPotMult_x{s_1^2}) \multisetScProd \multiSize{\oc  s_2})\right)
                        \\
                            \quad= \mset{\bangPotMult_y{s_1^1}}
                                \uplus \multiSize{s_1^1\esub{x}{\oc s_2}}
                                \uplus \multiSize{s_1^2\esub{x}{\oc s_2}}
                        \\
                            \quad\overset{\ih}{\succ} \mset{\bangPotMult_y{s_1^1}}
                                \uplus \multiSize{s_1^1\isub{x}{s_2}}
                                \uplus \bangPotMult_y{s_1^1} \multisetScProd \multiSize{s_1^1\isub{x}{s_2}}
                        \\
                            \quad= \mset{\bangPotMult_y{s_1^1\isub{x}{s_2}}}
                                \uplus \multiSize{s_1^1\isub{x}{s_2}}
                                \uplus \max(1, \bangPotMult_y{s_1^1\isub{x}{s_2}}) \multisetScProd \multiSize{s_1^1\isub{x}{s_2}}
                        \\
                                \quad= \multiSize{s_1^1\isub{x}{s_2}\esub{y}{s_1^2\isub{x}{s_2}}}
                                \quad=\quad \multiSize{s_1\isub{x}{s_2}}
                    \end{array}
                \end{equation*}
            \end{itemize}
        \end{itemize}

    \item[\bltII] $s_1 = \der{s'_1}$: Then $s_1\isub{x}{s_2} =
        \der{s'_1\isub{x}{s_2}}$ and by \ih on $s'_1$, one has that
        $\multiSize{s'_1\esub{x}{\oc s_2}} \succ
        \multiSize{s'_1\isub{x}{s_2}}$ thus:
        \begin{equation*}
            \begin{array}{rcl}
                \multiSize{s_1\esub{x}{\oc  s_2}}
                    &=& \multiSize{\der{s'_1}\esub{x}{\oc  s_2}}
                \\
                    &=& \mset{\bangPotMult_x{\der{s'_1}}} \uplus \multiSize{\der{s'_1}} \uplus \max(1, \bangPotMult_x{\der{s'_1}}) \multisetScProd \multiSize{\oc s_2}
                \\
                    &=& \mset{\bangPotMult_x{s'_1}} \uplus \multiSize{s'_1} \uplus \max(1, \bangPotMult_x{s'_1}) \multisetScProd \multiSize{\oc s_2}
                \\
                    &=& \multiSize{s'_1\esub{x}{\oc s_2}}
                \\
                    &\overset{\ih}{\succ}& \multiSize{s'_1\isub{x}{s_2}}
                \\
                    &=& \multiSize{\der{s'_1\isub{x}{s_2}}}
                    \quad=\quad \multiSize{s_1\isub{x}{s_2}}
            \end{array}
        \end{equation*}

    \item[\bltII] $s_1 = \oc s'_1$: Then $s_1\isub{x}{s_2} =
        \oc(s'_1\isub{x}{s_2})$ and by \ih on $s'_1$, one has that
        $\multiSize{s'_1\esub{x}{\oc s_2}} \succ
        \multiSize{s'_1\isub{x}{s_2}}$ thus:
        \begin{equation*}
            \begin{array}{rcl}
                \multiSize{s_1\esub{x}{\oc s_2}}
                    &=& \multiSize{(\oc s'_1)\esub{x}{\oc  s_2}}
                \\
                    &=& \mset{\bangPotMult_x{\oc s'_1}} \uplus \multiSize{\oc s'_1} \uplus \max(1, \bangPotMult_x{\oc s'_1}) \multisetScProd \multiSize{\oc s_2}
                \\
                    &=& \mset{\bangPotMult_x{s'_1}} \uplus \multiSize{s'_1} \uplus \max(1, \bangPotMult_x{s'_1}) \multisetScProd \multiSize{\oc  s_2}
                \\
                    &=& \multiSize{s'_1\esub{x}{\oc s_2}}
                \\
                    &\overset{\ih}{\succ}& \multiSize{s'_1\isub{x}{s_2}}
                \\
                    &=& \multiSize{\oc (s'_1\isub{x}{s_2})}
                    \quad=\quad \multiSize{s_1\isub{x}{s_2}}
            \end{array}
        \end{equation*}
    \end{itemize}

\item[\bltI] $\bangLCtxt = \bangLCtxt'\esub{y}{s}$: By \ih on
    $\bangLCtxt'$, one has that
    $\multiSize{s_1\esub{x}{\bangLCtxt'<\oc s_2>}} \succ
    \multiSize{\bangLCtxt'<s_1\isub{x}{s_2}>}$. Since $y \notin
    \freeVar{s_1}$ then $\bangPotMult_y{s_1} = 0$, and using
    \Cref{lem:Bang_Substitution_decreases_PotMult}, one therefore has
    that:
    \begin{equation*}
        \begin{array}{l}
            \multiSize{s_1\esub{x}{\bangLCtxt<\oc s_2>}}
        \\
            \quad = \multiSize{s_1\esub{x}{\bangLCtxt'<\oc s_2>\esub{y}{s}}}
        \\
            \quad = \mset{\bangPotMult_x{s_1}}
                \uplus \multiSize{s_1}
                \uplus \max(1, \bangPotMult_x{s_1}) \multisetScProd \multiSize{\bangLCtxt'<\oc s_2>\esub{y}{s}}
        \\
            \quad = \mset{\bangPotMult_x{s_1}}
                \uplus \multiSize{s_1}
                \uplus \max(1, \bangPotMult_x{s_1}) \multisetScProd
                    \left( \mset{\bangPotMult_y{\bangLCtxt'<\oc s_2>}}
                    \uplus \multiSize{\bangLCtxt'<\oc s_2>}
                    \uplus \bangPotMult_y{\bangLCtxt'<\oc s_2>} \multisetScProd \multiSize{s} \right)
        \\
            \quad = \left( \mset{\bangPotMult_x{s_1}}
                \uplus \multiSize{s_1}
                \uplus \max(1, \bangPotMult_x{s_1}) \multisetScProd \multiSize{\bangLCtxt'<\oc s_2>} \right)
            \\
                \qquad\; \uplus \max(1, \bangPotMult_x{s_1}) \multisetScProd \mset{\bangPotMult_y{\bangLCtxt'<\oc s_2>}}
                    \uplus \max(1, \bangPotMult_x{s_1}) \multisetScProd \bangPotMult_y{\bangLCtxt'<\oc s_2>} \multisetScProd \multiSize{s}
        \\
            \quad = \multiSize{s_1\esub{x}{\bangLCtxt'<\oc s_2>}}
                \uplus \max(1, \bangPotMult_x{s_1}) \multisetScProd \mset{\bangPotMult_y{\bangLCtxt'<\oc s_2>}}
                \uplus \max(1, \bangPotMult_x{s_1}) \multisetScProd \bangPotMult_y{\bangLCtxt'<\oc s_2>} \multisetScProd \multiSize{s}
        \\
            \quad\overset{\ih}{\succ} \multiSize{\bangLCtxt'<s_1\isub{x}{s_2}>}
                \uplus \max(1, \bangPotMult_x{s_1}) \multisetScProd \mset{\bangPotMult_y{\bangLCtxt'<\oc s_2>}}
                \uplus \max(1, \bangPotMult_x{s_1}) \multisetScProd \bangPotMult_y{\bangLCtxt'<\oc s_2>} \multisetScProd \multiSize{s}
        \\
            \quad = \multiSize{\bangLCtxt'<s_1\isub{x}{s_2}>}
                \uplus \mset{\max(1, \bangPotMult_x{s_1}) \cdot \bangPotMult_y{\bangLCtxt'<\oc s_2>}}
                \uplus (\max(1, \bangPotMult_x{s_1}) \cdot \bangPotMult_y{\bangLCtxt'<\oc s_2>}) \multisetScProd \multiSize{s}
        \\
            \quad = \multiSize{\bangLCtxt'<s_1\isub{x}{s_2}>}
                \uplus \mset{\bangPotMult_y{s_1} + \max(1, \bangPotMult_x{s_1}) \cdot \bangPotMult_y{\bangLCtxt'<\oc s_2>}}
                \uplus (\bangPotMult_y{s_1} + \max(1, \bangPotMult_x{s_1}) \cdot \bangPotMult_y{\bangLCtxt'<\oc s_2>}) \multisetScProd \multiSize{s}
        \\
            \quad = \multiSize{\bangLCtxt'<s_1\isub{x}{s_2}>}
                \uplus \mset{\bangPotMult_y{s_1\esub{x}{\bangLCtxt'<\oc s_2>}}}
                \uplus (\bangPotMult_y{s_1\esub{x}{\bangLCtxt'<\oc s_2>}}) \multisetScProd \multiSize{s}
        \\
            \quad \succ \multiSize{\bangLCtxt'<s_1\isub{x}{s_2}>}
                \uplus \mset{\bangPotMult_y{\bangLCtxt'<s_1\isub{x}{s_2}>}}
                \uplus \bangPotMult_y{\bangLCtxt'<s_1\isub{x}{s_2}>} \multisetScProd \multiSize{s}
        \\
            \quad = \multiSize{\bangLCtxt'<s_1\isub{x}{s_2}>\esub{y}{s}}
        \end{array}
    \end{equation*}
\end{itemize}
\end{proof}
 } \deliaLu \giulioLu%

\begin{lemma}
    \label{lem:Bang_Full_PotMulti_decreases_Substitution}%
    Let $t \bangArr_{S_\indexOmega}<s!> u$, then $\bangPotMult_x{t}
    \geq \bangPotMult_x{u}$.
\end{lemma}
\stableProof{%
    \begin{proof}
Let $t, u \in \bangSetTerms$ such that $t \bangArr_{S_\indexOmega}<s!>
u$. By definition, there exist a context $\bangStratCtxt_\indexOmega$
and terms $s_t, s_u \in \bangSetTerms$ such that $t =
\bangStratCtxt_\indexOmega<s_t>$, $\bangStratCtxt_\indexOmega<s_u>$
and $t \mapstoR[\bangSymbBang] u$. By induction on
$\bangStratCtxt_\indexOmega$:
\begin{itemize}
\item[\bltI] $\bangStratCtxt_\indexOmega = \Hole$: Then $t = s_t$, $u
    = s_u$ and there exist a context $\bangLCtxt$ and terms $s_1, s_2
    \in \bangSetTerms$ such that $s_t = s_1\esub{y}{\bangLCtxt<\oc
    s_2>}$ and $s_u = \bangLCtxt<s_1\isub{y}{s_2}>$ and one concludes
    using \Cref{lem:Bang_Substitution_decreases_PotMult}.

\item[\bltI] $\bangStratCtxt_\indexOmega =
    \abs{y}{\bangStratCtxt'_\indexOmega}$: By \ih on
    $\bangStratCtxt'_\indexOmega$, one has that
    $\bangPotMult_x{\bangStratCtxt'_\indexOmega<s_t>} \geq
    \bangPotMult_x{\bangStratCtxt'_\indexOmega<s_u>}$ thus:
    \begin{equation*}
        \bangPotMult_x{t}
        = \bangPotMult_x{\bangStratCtxt_\indexOmega<s_t>}
        = \bangPotMult_x{\abs{y}{\bangStratCtxt'_\indexOmega<s_t>}}
        = \bangPotMult_x{\bangStratCtxt'_\indexOmega<s_t>}
        \geq \bangPotMult_x{\bangStratCtxt'_\indexOmega<s_u>}
        = \bangPotMult_x{\abs{y}{\bangStratCtxt'_\indexOmega<s_u>}}
        = \bangPotMult_x{\bangStratCtxt_\indexOmega<s_u>}
        = \bangPotMult_x{u}
    \end{equation*}

\item[\bltI] $\bangStratCtxt_\indexOmega =
    \app[\,]{\bangStratCtxt'_\indexOmega}{s}$: By \ih on
    $\bangStratCtxt'_\indexOmega$, one has that
    $\bangPotMult_x{\bangStratCtxt'_\indexOmega<s_t>} \geq
    \bangPotMult_x{\bangStratCtxt'_\indexOmega<s_u>}$ thus:
    \begin{equation*}
        \begin{array}{l}
            \bangPotMult_x{t}
            \;=\; \bangPotMult_x{\bangStratCtxt_\indexOmega<s_t>}
            \;=\; \bangPotMult_x{\app[\,]{\bangStratCtxt'_\indexOmega<s_t>}{s}}
            \;=\; \bangPotMult_x{\bangStratCtxt'_\indexOmega<s_t>} + \bangPotMult_x{s}
        \\
            \qquad \; \geq \; \bangPotMult_x{\bangStratCtxt'_\indexOmega<s_u>} + \bangPotMult_x{s}
            \;=\; \bangPotMult_x{\app[\,]{\bangStratCtxt'_\indexOmega<s_u>}{s}}
            \;=\; \bangPotMult_x{\bangStratCtxt_\indexOmega<s_u>}
            \;=\; \bangPotMult_x{u}
        \end{array}
    \end{equation*}

\item[\bltI] $\bangStratCtxt_\indexOmega =
    \app[\,]{s}{\bangStratCtxt'_\indexOmega}$: By \ih on
    $\bangStratCtxt'_\indexOmega$, one has that
    $\bangPotMult_x{\bangStratCtxt'_\indexOmega<s_t>} \geq
    \bangPotMult_x{\bangStratCtxt'_\indexOmega<s_u>}$ thus:
    \begin{equation*}
        \begin{array}{l}
            \bangPotMult_x{t}
            \;=\; \bangPotMult_x{\bangStratCtxt_\indexOmega<s_t>}
            \;=\; \bangPotMult_x{\app[\,]{s}{\bangStratCtxt'_\indexOmega<s_t>}}
            \;=\; \bangPotMult_x{s} + \bangPotMult_x{\bangStratCtxt'_\indexOmega<s_t>}
        \\
            \qquad \; \geq \; \bangPotMult_x{s} + \bangPotMult_x{\bangStratCtxt'_\indexOmega<s_u>}
            \;=\; \bangPotMult_x{\app[\,]{s}{\bangStratCtxt'_\indexOmega<s_u>}}
            \;=\; \bangPotMult_x{\bangStratCtxt_\indexOmega<s_u>}
            \;=\; \bangPotMult_x{u}
        \end{array}
    \end{equation*}

\item[\bltI] $\bangStratCtxt_\indexOmega =
    \bangStratCtxt'_\indexOmega\esub{y}{s}$: By \ih on
    $\bangStratCtxt'_\indexOmega$, one has that
    $\bangPotMult_x{\bangStratCtxt'_\indexOmega<s_t>} \geq
    \bangPotMult_x{\bangStratCtxt'_\indexOmega<s_u>}$ and
    $\bangPotMult_y{\bangStratCtxt'_\indexOmega<s_t>} \geq
    \bangPotMult_y{\bangStratCtxt'_\indexOmega<s_u>}$ thus:
    \begin{equation*}
        \begin{array}{l}
            \bangPotMult_x{t}
            \;=\; \bangPotMult_x{\bangStratCtxt<s_t>}
            \;=\; \bangPotMult_x{\bangStratCtxt'<s_t>\esub{y}{s}}
            \;=\; \bangPotMult_x{\bangStratCtxt'<s_t>} + \max(1, \bangPotMult_y{\bangStratCtxt'<s_t>}) \cdot \bangPotMult_x{s}
        \\
            \qquad \;\geq\; \bangPotMult_x{\bangStratCtxt'<s_u>} + \max(1, \bangPotMult_y{\bangStratCtxt'<s_u>}) \cdot \bangPotMult_x{s}
            \;=\; \bangPotMult_x{\bangStratCtxt'<s_u>\esub{y}{s}}
            \;=\; \bangPotMult_x{\bangStratCtxt<s_u>}
            \;=\; \bangPotMult_x{u}
        \end{array}
    \end{equation*}

\item[\bltI] $\bangStratCtxt_\indexOmega =
    s\esub{y}{\bangStratCtxt'_\indexOmega}$: By \ih on
    $\bangStratCtxt'_\indexOmega$, one has that
    $\bangPotMult_x{\bangStratCtxt'_\indexOmega<s_t>} \geq
    \bangPotMult_x{\bangStratCtxt'_\indexOmega<s_u>}$ thus:
    \begin{equation*}
        \begin{array}{l}
            \bangPotMult_x{t}
            \;=\; \bangPotMult_x{\bangStratCtxt<s_t>}
            \;=\; \bangPotMult_x{s\esub{y}{\bangStratCtxt'<s_t>}}
            \;=\; \bangPotMult_x{s} + \max(1, \bangPotMult_y{s}) \cdot \bangPotMult_x{\bangStratCtxt'<s_t>}
        \\
            \qquad \;\geq\; \bangPotMult_x{s} + \max(1, \bangPotMult_y{s}) \cdot \bangPotMult_x{\bangStratCtxt'<s_u>}
            \;=\; \bangPotMult_x{s\esub{y}{\bangStratCtxt'<s_u>}}
            \;=\; \bangPotMult_x{\bangStratCtxt<s_u>}
            \;=\; \bangPotMult_x{u}
        \end{array}
    \end{equation*}

\item[\bltI] $\bangStratCtxt_\indexOmega =
    \der{\bangStratCtxt'_\indexOmega}$: By \ih on
    $\bangStratCtxt'_\indexOmega$, one has that
    $\bangPotMult_x{\bangStratCtxt'_\indexOmega<s_t>} \geq
    \bangPotMult_x{\bangStratCtxt'_\indexOmega<s_u>}$ thus:
    \begin{equation*}
        \bangPotMult_x{t}
        = \bangPotMult_x{\bangStratCtxt_\indexOmega<s_t>}
        = \bangPotMult_x{\der{\bangStratCtxt'_\indexOmega<s_t>}}
        = \bangPotMult_x{\bangStratCtxt'_\indexOmega<s_t>}
        \geq \bangPotMult_x{\bangStratCtxt'_\indexOmega<s_u>}
        = \bangPotMult_x{\der{\bangStratCtxt'_\indexOmega<s_u>}}
        = \bangPotMult_x{\bangStratCtxt_\indexOmega<s_u>}
        = \bangPotMult_x{u}
    \end{equation*}

\item[\bltI] $\bangStratCtxt_\indexOmega =
    \oc\bangStratCtxt'_\indexOmega$: By \ih on
    $\bangStratCtxt'_\indexOmega$, one has that
    $\bangPotMult_x{\bangStratCtxt'_\indexOmega<s_t>} \geq
    \bangPotMult_x{\bangStratCtxt'_\indexOmega<s_u>}$ thus:
    \begin{equation*}
        \bangPotMult_x{t}
        = \bangPotMult_x{\bangStratCtxt_\indexOmega<s_t>}
        = \bangPotMult_x{\oc \bangStratCtxt'_\indexOmega<s_t>}
        = \bangPotMult_x{\bangStratCtxt'_\indexOmega<s_t>}
        \geq \bangPotMult_x{\bangStratCtxt'_\indexOmega<s_u>}
        = \bangPotMult_x{\oc \bangStratCtxt'_\indexOmega<s_u>}
        = \bangPotMult_x{\bangStratCtxt_\indexOmega<s_u>}
        = \bangPotMult_x{u}
    \end{equation*}
\end{itemize}
\end{proof}
} \deliaLu \giulioLu

\begin{lemma}
    \label{lem:Bang_Full_MultiSize_decreases_Substitution}%
    Let $t \bangArr_{S_\indexOmega}<s!> u$, then $\multiSize{t} \succ
    \multiSize{u}$.
\end{lemma}
\stableProof{%
    \begin{proof}
Let $t, u \in \bangSetTerms$ such that $t \bangArr_{S_\indexOmega}<s!>
u$. By definition, there exist a context $\bangStratCtxt_\indexOmega$
and terms $s_t, s_u \in \bangSetTerms$ such that $t =
\bangStratCtxt_\indexOmega<s_t>$, $\bangStratCtxt_\indexOmega<s_u>$
and $t \mapstoR[\bangSymbBang] u$. By induction on
$\bangStratCtxt_\indexOmega$:
\begin{itemize}
\item[\bltI] $\bangStratCtxt_\indexOmega = \Hole$: Then $t = s_t$, $u
    = s_u$ and there exist a context $\bangLCtxt$ and terms $s_1, s_2
    \in \bangSetTerms$ such that $s_t = s_1\esub{x}{\bangLCtxt<\oc
    s_2>}$ and $s_u = \bangLCtxt<s_1\isub{x}{s_2}>$ and one concludes
    using \Cref{lem:Bang_MultiSize_Decreases_Substitution}.

\item[\bltI] $\bangStratCtxt_\indexOmega =
    \abs{y}{\bangStratCtxt'_\indexOmega}$: By \ih on
    $\bangStratCtxt'_\indexOmega$, one has that
    $\multiSize{\bangStratCtxt'_\indexOmega<s_t>} \succ
    \multiSize{\bangStratCtxt'_\indexOmega<s_u>}$ thus:
    \begin{equation*}
        \multiSize{t}
        = \multiSize{\bangStratCtxt_\indexOmega<s_t>}
        = \multiSize{\abs{y}{\bangStratCtxt'_\indexOmega<s_t>}}
        = \multiSize{\bangStratCtxt'_\indexOmega<s_t>}
        \succ \multiSize{\bangStratCtxt'_\indexOmega<s_u>}
        = \multiSize{\abs{y}{\bangStratCtxt'_\indexOmega<s_u>}}
        = \multiSize{\bangStratCtxt_\indexOmega<s_u>}
        = \multiSize{u}
    \end{equation*}

\item[\bltI] $\bangStratCtxt_\indexOmega =
    \app[\,]{\bangStratCtxt'_\indexOmega}{s}$: By \ih on
    $\bangStratCtxt'_\indexOmega$, one has that
    $\multiSize{\bangStratCtxt'_\indexOmega<s_t>} \succ
    \multiSize{\bangStratCtxt'_\indexOmega<s_u>}$ thus:
    \begin{equation*}
        \begin{array}{l}
            \multiSize{t}
            \;=\; \multiSize{\bangStratCtxt_\indexOmega<s_t>}
            \;=\; \multiSize{\app[\,]{\bangStratCtxt'_\indexOmega<s_t>}{s}}
            \;=\; \multiSize{\bangStratCtxt'_\indexOmega<s_t>} \uplus \multiSize{s}
        \\
            \qquad \; \succ \; \multiSize{\bangStratCtxt'_\indexOmega<s_u>} \uplus \multiSize{s}
            \;=\; \multiSize{\app[\,]{\bangStratCtxt'_\indexOmega<s_u>}{s}}
            \;=\; \multiSize{\bangStratCtxt_\indexOmega<s_u>}
            \;=\; \multiSize{u}
        \end{array}
    \end{equation*}

\item[\bltI] $\bangStratCtxt_\indexOmega =
    \app[\,]{s}{\bangStratCtxt'_\indexOmega}$: By \ih on
    $\bangStratCtxt'_\indexOmega$, one has that
    $\multiSize{\bangStratCtxt'_\indexOmega<s_t>} \succ
    \multiSize{\bangStratCtxt'_\indexOmega<s_u>}$ thus:
    \begin{equation*}
        \begin{array}{l}
            \multiSize{t}
            \;=\; \multiSize{\bangStratCtxt_\indexOmega<s_t>}
            \;=\; \multiSize{\app[\,]{s}{\bangStratCtxt'_\indexOmega<s_t>}}
            \;=\; \multiSize{s} \uplus \multiSize{\bangStratCtxt'_\indexOmega<s_t>}
        \\
            \qquad \; \succ \; \multiSize{s} \uplus \multiSize{\bangStratCtxt'_\indexOmega<s_u>}
            \;=\; \multiSize{\app[\,]{s}{\bangStratCtxt'_\indexOmega<s_u>}}
            \;=\; \multiSize{\bangStratCtxt_\indexOmega<s_u>}
            \;=\; \multiSize{u}
        \end{array}
    \end{equation*}

\item[\bltI] $\bangStratCtxt_\indexOmega =
    \bangStratCtxt'_\indexOmega\esub{y}{s}$: By \ih on
    $\bangStratCtxt'_\indexOmega$, one has that
    $\multiSize{\bangStratCtxt'_\indexOmega<s_t>} \succ
    \multiSize{\bangStratCtxt'_\indexOmega<s_u>}$. Moreover, by
    \Cref{lem:Bang_Full_PotMulti_decreases_Substitution}, one deduces
    that $\bangPotMult_y{\bangStratCtxt'_\indexOmega<s_t>} \geq
    \bangPotMult_y{\bangStratCtxt'_\indexOmega<s_u>}$, thus:
    \begin{equation*}
        \begin{array}{l}
            \multiSize{t}
            \;=\; \multiSize{\bangStratCtxt_\indexOmega<s_t>}
            \;=\; \multiSize{\bangStratCtxt'_\indexOmega<s_t>\esub{y}{s}}
            \;=\; \multiSize{\bangStratCtxt'_\indexOmega<s_t>} \uplus \max(1, \bangPotMult_y{\bangStratCtxt'_\indexOmega<s_t>}) \multisetScProd \multiSize{s}
        \\
            \qquad \; \succ \; \multiSize{\bangStratCtxt'_\indexOmega<s_u>} \uplus \max(1, \bangPotMult_y{\bangStratCtxt'_\indexOmega<s_u>}) \multisetScProd \multiSize{s}
            \;=\; \multiSize{\bangStratCtxt'_\indexOmega<s_u>\esub{y}{s}}
            \;=\; \multiSize{\bangStratCtxt_\indexOmega<s_u>}
            \;=\; \multiSize{u}
        \end{array}
    \end{equation*}

\item[\bltI] $\bangStratCtxt_\indexOmega =
    s\esub{y}{\bangStratCtxt'_\indexOmega}$: By \ih on
    $\bangStratCtxt'_\indexOmega$, one has that
    $\multiSize{\bangStratCtxt'_\indexOmega<s_t>} \succ
    \multiSize{\bangStratCtxt'_\indexOmega<s_u>}$ thus:
    \begin{equation*}
        \begin{array}{l}
            \multiSize{t}
            \;=\; \multiSize{\bangStratCtxt_\indexOmega<s_t>}
            \;=\; \multiSize{s\esub{y}{\bangStratCtxt'_\indexOmega<s_t>}}
            \;=\; \multiSize{s} \uplus \max(1, \bangPotMult_y{s}) \multisetScProd\multiSize{\bangStratCtxt'_\indexOmega<s_t>}
        \\
            \qquad \; \succ \; \multiSize{s} \uplus \max(1, \bangPotMult_y{s}) \multisetScProd\multiSize{\bangStratCtxt'_\indexOmega<s_u>}
            \;=\; \multiSize{s\esub{y}{\bangStratCtxt'_\indexOmega<s_u>}}
            \;=\; \multiSize{\bangStratCtxt_\indexOmega<s_u>}
            \;=\; \multiSize{u}
        \end{array}
    \end{equation*}

\item[\bltI] $\bangStratCtxt_\indexOmega =
    \der{\bangStratCtxt'_\indexOmega}$: By \ih on
    $\bangStratCtxt'_\indexOmega$, one has that
    $\multiSize{\bangStratCtxt'_\indexOmega<s_t>} \succ
    \multiSize{\bangStratCtxt'_\indexOmega<s_u>}$ thus:
    \begin{equation*}
        \multiSize{t}
        = \multiSize{\bangStratCtxt_\indexOmega<s_t>}
        = \multiSize{\der{\bangStratCtxt'_\indexOmega<s_t>}}
        = \multiSize{\bangStratCtxt'_\indexOmega<s_t>}
        \succ \multiSize{\bangStratCtxt'_\indexOmega<s_u>}
        = \multiSize{\der{\bangStratCtxt'_\indexOmega<s_u>}}
        = \multiSize{\bangStratCtxt_\indexOmega<s_u>}
        = \multiSize{u}
    \end{equation*}

\item[\bltI] $\bangStratCtxt_\indexOmega =
    \oc\bangStratCtxt'_\indexOmega$: By \ih on
    $\bangStratCtxt'_\indexOmega$, one has that
    $\multiSize{\bangStratCtxt'_\indexOmega<s_t>} \succ
    \multiSize{\bangStratCtxt'_\indexOmega<s_u>}$ thus:
    \begin{equation*}
        \multiSize{t}
        = \multiSize{\bangStratCtxt_\indexOmega<s_t>}
        = \multiSize{\oc \bangStratCtxt'_\indexOmega<s_t>}
        = \multiSize{\bangStratCtxt'_\indexOmega<s_t>}
        \succ \multiSize{\bangStratCtxt'_\indexOmega<s_u>}
        = \multiSize{\oc \bangStratCtxt'_\indexOmega<s_u>}
        = \multiSize{\bangStratCtxt_\indexOmega<s_u>}
        = \multiSize{u}
    \end{equation*}
\end{itemize}
\end{proof}%
} \deliaLu \giulioLu

\begin{corollary}
    \label{lem:Bang_Full_s!_Terminating}%
    The reduction $\bangArr_{S_\indexOmega}<s!>$ is terminating.
\end{corollary}
    \begin{proof}
Let $t \in \bangSetTerms$. As the finite multisets over natural
numbers is well-ordered by $\prec$, we can proceed by induction on
$\multiSize{t}$. If $t$ is a $\bangSymbSubs$-normal form, then it is
trivially terminating. Otherwise, $t \bangArr_{S_\indexOmega}<s!> u$,
then $\multiSize{t} \succ \multiSize{u}$
(\Cref{lem:Bang_Full_MultiSize_decreases_Substitution}), and so $u$ is
strongly $\bangSymbSubs$-normalizing, by the \ih\ Therefore, $t$ is
strongly $\bangSymbSubs$-normalizing.
\end{proof}
%

\subsubsection{Confluence of $\bangArr_{S_\indexOmega}<s!>$}

\begin{corollary}
    \label{lem:Bang_Full_s!_Confluent}%
    The reduction $\bangArr_{S_\indexOmega}<s!>$ is confluent.
\end{corollary}
\begin{proof} \giulioLu \deliaLu
    By \Cref{lem:Bang_Full_s!_Locally_Confluent,lem:Bang_Full_s!_Locally_Confluent,lem:Bang_Full_s!_Terminating}, using Newman's Lemma.
\end{proof}

\subsection{Confluence of $\bangArr_{S_\indexOmega}$}

\begin{lemma}
    \label{lem:Bang_Full_dB_U_d!_and_s!_strongly_commute}%
    The reductions $\bangArr_{S_\indexOmega}<dB> \cup
    \bangArr_{S_\indexOmega}<d!>$ and $\bangArr_{S_\indexOmega}<s!>$
    strongly commute.
\end{lemma}
\stableProof{
    \begin{proof}
  Let $t, u_1, u_2 \in \bangSetTerms$ such that $t
\bangArr_{S_\indexOmega}<R> u_1$ and $t \bangArr_{S_\indexOmega}<s!>
u_2$ for some $\rel \in \{\bangSymbBeta, \bangSymbBang\}$. Let us show
that there exists $s \in \bangSetTerms$ such that $u_1
\bangArr_{S_\indexOmega}<s!> s$ and $u_2 \bangArr*_{S_\indexOmega}<R>
s$.

By induction on $t$:
\begin{itemize}
\item[\bltI] $t = x$: Impossible since it contradicts the hypothesis
    $t \bangArr_{S_\indexOmega}<s!> u_2$.

\item[\bltI] $t = \abs{x}{t'}$: Since by hypothesis $t
    \bangArr_{S_\indexOmega}<R> u_1$ and $t
    \bangArr_{S_\indexOmega}<s!> u_2$, then necessarily $u_1 =
    \abs{x}{u'_1}$ and $u_2 = \abs{x}{u'_2}$ for some $u'_1, u'_2 \in
    \bangSetTerms$ such that $t' \bangArr_{S_\indexOmega}<R> u'_1$ and
    $t' \bangArr_{S_\indexOmega}<s!> u'_2$. By \ih on $t'$, there
    exists $s' \in \bangSetTerms$ such that $u'_1
    \bangArr_{S_\indexOmega}<s!> s'$ and $u'_2
    \bangArr*_{S_\indexOmega}<R> s'$. We set $s := \abs{x}{s'}$ since
    by contextual closure $u_1 = \abs{x}{u'_1}
    \bangArr_{S_\indexOmega}<s!> \abs{x}{s'} = s'$ and $u_2 =
    \abs{x}{u'_2} \bangArr*_{S_\indexOmega}<R> \abs{x}{s'} = s'$.
    Graphically,
    \begin{equation*}
        \begin{array}{lll}
            \abs{x}{t'}                         &\bangArr_{S_\indexOmega}<R>    &  \abs{x}{u'_1}
        \\[0.2cm]
           \;\;\bangDownArr_{S_\indexOmega}<s!> &                               &\;\;\bangDownArr_{S_\indexOmega}<s!>
        \\[0.2cm]
           \abs{x}{u'_2}                        &\bangArr*_{S_\indexOmega}<R>   & \abs{x}{s'}
       \end{array}
    \end{equation*}
\item[\bltI] $t = \app{t_1}{t_2}$: We distinguish two cases:
    \begin{itemize}
    \item[\bltII] $t$ is the $\bangSymbBeta$-redex reduced in the step
        $t \bangArr_{S_\indexOmega}<R> u_1$: Then $t_1 =
        \bangLCtxt<\abs{x}{t'_1}>$ and $u_1 =
        \bangLCtxt<t'_1\esub{x}{t_2}>$. By hypothesis
        $\app{\bangLCtxt<\abs{x}{t'_1}>}{t_2}
        \bangArr_{S_\indexOmega}<s!> u_2$ so that four different cases
        can be distinguished:
        \begin{itemize}
        \item[\bltIII] $u_2 = \app{\bangLCtxt<\abs{x}{u'_1}>}{t_2}$
            for some $u'_1 \in \bangSetTerms$ such that $t'_1
            \bangArr_{S_\indexOmega}<s!> u'_1$: We set $s :=
            \bangLCtxt<u'_1\esub{x}{t_2}>$ which concludes this case
            since by contextual closure $u_1 =
            \bangLCtxt<t'_1\esub{x}{t_2}> \bangArr_{S_\indexOmega}<s!>
            \bangLCtxt<u'_1\esub{x}{t_2}> = s$ and $u_2 =
            \app{\bangLCtxt<\abs{x}{u'_1}>}{t_2}
            \bangArr_{S_\indexOmega}<dB> \bangLCtxt<u'_1\esub{x}{t_2}>
            = s$. Graphically,
            \begin{equation*}
                \begin{array}{lll}
                    \app{\bangLCtxt<\abs{x}{t'_1}>}{t_2}        &\bangArr_{S_\indexOmega}<dB>   &\bangLCtxt<t'_1\esub{x}{t_2}>
                \\[0.2cm]
                    \;\;\bangDownArr_{S_\indexOmega}<s!>        &                               &\;\;\bangDownArr_{S_\indexOmega}<s!>
                \\[0.2cm]
                    \app{\bangLCtxt<\abs{x}{u'_1}>}{t_2}        &\bangArr_{S_\indexOmega}<dB>   &\bangLCtxt<u'_1\esub{x}{t_2}>
                \end{array}
            \end{equation*}

        \item[\bltIII] $u_2 = \app{\bangLCtxt<\abs{x}{t'_1}>}{u'_2}$
            for some $u'_2 \in \bangSetTerms$ such that $t_2
            \bangArr_{S_\indexOmega}<s!> u'_2$: We set $s :=
            \bangLCtxt<t'_1\esub{x}{u'_2}>$ which concludes this case
            since by contextual closure $u_1 =
            \bangLCtxt<t'_1\esub{x}{t_2}> \bangArr_{S_\indexOmega}<s!>
            \bangLCtxt<t'_1\esub{x}{u'_2}> = s$ and $u_2 =
            \app{\bangLCtxt<\abs{x}{t'_1}>}{u'_2}
            \bangArr_{S_\indexOmega}<dB>
            \bangLCtxt<t'_1\esub{x}{u'_2}> = s$. Graphically,
            \begin{equation*}
                \begin{array}{lll}
                    \app{\bangLCtxt<\abs{x}{t'_1}>}{t_2}        &\bangArr_{S_\indexOmega}<dB>   &  \bangLCtxt<t'_1\esub{x}{t_2}>
                \\[0.2cm]
                    \;\;\bangDownArr_{S_\indexOmega}<s!>        &                               &\;\;\bangDownArr_{S_\indexOmega}<s!>
                \\[0.2cm]
                        \app{\bangLCtxt<\abs{x}{t'_1}>}{u'_2}   &\bangArr_{S_\indexOmega}<dB>  & \bangLCtxt<t'_1\esub{x}{u'_2}>
                \end{array}
            \end{equation*}

        \item[\bltIII] $u_2 = \app{\bangLCtxt'<\abs{x}{t'_1}>}{t_2}$
            for some context $\bangLCtxt'$ such that $\bangLCtxt
            \bangArr_{S_\indexOmega}<s!> \bangLCtxt'$: We set $s :=
            \bangLCtxt'<t'_1\esub{x}{t_2}>$ which concludes this case
            since by
            \Cref{lem:Bang_Full_dB_d!_Context_Reduction_Lifted_to_Terms},
            $u_1 = \bangLCtxt<t'_1\esub{x}{t_2}>
            \bangArr_{S_\indexOmega}<s!>
            \bangLCtxt'<t'_1\esub{x}{t_2}> = s$ and $u_2 =
            \app{\bangLCtxt'<\abs{x}{t'_1}>}{t_2}
            \bangArr_{S_\indexOmega}<dB>
            \bangLCtxt'<t'_1\esub{x}{t_2}> = s$. Graphically,
            \begin{equation*}
                \begin{array}{lll}
                    \app{\bangLCtxt<\abs{x}{t'_1}>}{t_2}    &\bangArr_{S_\indexOmega}<dB>   &  \bangLCtxt<t'_1\esub{x}{t_2}>
                \\[0.2cm]
                    \;\;\bangDownArr_{S_\indexOmega}<s!>    &                               &\;\;\bangDownArr_{S_\indexOmega}<s!>
                \\[0.2cm]
                    \app{\bangLCtxt'<\abs{x}{t'_1}>}{t_2}   &\bangArr_{S_\indexOmega}<dB>   & \bangLCtxt'<t'_1\esub{x}{t_2}>
                \end{array}
            \end{equation*}

        \item[\bltIII] $\bangLCtxt =
            \bangLCtxt_1<\bangLCtxt_3\esub{y}{\bangLCtxt_2<\oc
            t''_1>}>$, $u_2 =
            \app{\bangLCtxt_1<\bangLCtxt_2<\bangLCtxt_3\isub{y}{t''_1}\bangCtxtPlug{\abs{x}{t'_1\isub{y}{t''_1}}}>>}{t_2}$
            for some contexts $\bangLCtxt_1, \bangLCtxt_2,
            \bangLCtxt_3$ and some term $t''_1 \in \bangSetTerms$: We
            set $s :=
            \bangLCtxt_1<\bangLCtxt_2<\bangLCtxt_3\isub{y}{t''_1}\bangCtxtPlug{t'_1\isub{y}{t''_1}\esub{x}{t_2}}>>$
            which concludes this case since by $\alpha$-conversion $y
            \notin \freeVar{t_2}$ thus $t_2\isub{x}{t''_1} = t_2$ and
            therefore: Graphically,
            \begin{equation*}
                \begin{array}{lll}
                    \app{\bangLCtxt_1<\bangLCtxt_3<\abs{x}{t'_1}>\esub{y}{\bangLCtxt_2<\oc t''_1>}>}{t_2}                           &\bangArr_{S_\indexOmega}<dB>   &\bangLCtxt_1<\bangLCtxt_3<t'_1\esub{x}{t_2}>\esub{y}{\bangLCtxt_2<\oc t''_1>}>
                \\[0.2cm]
                    \;\;\bangDownArr_{S_\indexOmega}<s!>                                                                            &                               &\;\;\bangDownArr_{S_\indexOmega}<s!>
                \\[0.2cm]
                    \app{\bangLCtxt_1<\bangLCtxt_2<\bangLCtxt_3\isub{y}{t''_1}\bangCtxtPlug{\abs{x}{t'_1\isub{y}{t''_1}}}>>}{t_2}   &\bangArr*_{S_\indexOmega}<dB>  &\bangLCtxt_1<\bangLCtxt_2<\bangLCtxt_3\isub{y}{t''_1}\bangCtxtPlug{t'_1\isub{y}{t''_1}\esub{x}{t_2}}>>
                \end{array}
            \end{equation*}
        \end{itemize}

    \item[\bltII] Otherwise:  By hypothesis $t
        \bangArr_{S_\indexOmega}<R> u_1$ and $t
        \bangArr_{S_\indexOmega}<s!> u_2$, and four cases can thus be
        distinguished:
        \begin{itemize}
        \item[\bltIII] $u_1 = \app{u'_1}{t_2}$ and $u_2 =
            \app{u'_2}{t_2}$ with $t_1 \bangArr_{S_\indexOmega}<R>
            u'_1$ and $t_1 \bangArr_{S_\indexOmega}<s!> u'_2$ for some
            $u'_1, u'_2 \in \bangSetTerms$: By \ih on $t_1$, there
            exists $s' \in \bangSetTerms$ such that $u'_1
            \bangArr_{S_\indexOmega}<s!> s'$ and $u'_2
            \bangArr*_{S_\indexOmega}<R> s'$. We set $s :=
            \app{s'}{t_2}$ since by contextual closure $u_1 =
            \app{u'_1}{t_2} \bangArr_{S_\indexOmega}<s!> \app{s'}{t_2}
            = s'$ and $u_2 = \app{u'_2}{t_2}
            \bangArr*_{S_\indexOmega}<R> \app{s'}{t_2} = s'$.
            Graphically,
            \begin{equation*}
                \begin{array}{lll}
                    \app{t_1}{t_2}                          &\bangArr_{S_\indexOmega}<R>    &\app{u'_1}{t_2}
                \\[0.2cm]
                    \;\;\bangDownArr_{S_\indexOmega}<s!>    &                               &\;\;\bangDownArr_{S_\indexOmega}<s!>
                \\[0.2cm]
                    \app{u'_2}{t_2}                         &\bangArr*_{S_\indexOmega}<R>   &\app{s'}{t_2}
                \end{array}
            \end{equation*}

        \item[\bltIII] $u_1 = \app{u'_1}{t_2}$ and $u_2 =
            \app{t_1}{u'_2}$ with $t_1 \bangArr_{S_\indexOmega}<R>
            u'_1$ and $t_2 \bangArr_{S_\indexOmega}<s!> u'_2$ for some
            $u'_1, u'_2 \in \bangSetTerms$: We set $s' :=
            \app{u'_1}{u'_2}$ since by contextual closure $u_1 =
            \app{u'_1}{t_2} \bangArr_{S_\indexOmega}<s!>
            \app{u'_1}{u'_2} = s'$ and $u_2 = \app{t_1}{u'_2}
            \bangArr_{S_\indexOmega}<R> \app{u'_1}{u'_2} = s'$.
            Graphically,
            \begin{equation*}
                \begin{array}{lll}
                    \app{t_1}{t_2}                          &\bangArr_{S_\indexOmega}<R>    &\app{u'_1}{t_2}
                \\[0.2cm]
                    \;\;\bangDownArr_{S_\indexOmega}<s!>    &                               &\;\;\bangDownArr_{S_\indexOmega}<s!>
                \\[0.2cm]
                    \app{t_1}{u'_2}                         &\bangArr_{S_\indexOmega}<R>    &\app{u'_1}{u'_2}
                \end{array}
            \end{equation*}

        \item[\bltIII] $u_1 = \app{t_1}{u'_1}$ and $u_2 =
            \app{u'_2}{t_2}$ with $t_2 \bangArr_{S_\indexOmega}<R>
            u'_1$ and $t_1 \bangArr_{S_\indexOmega}<s!> u'_2$ for some
            $u'_1, u'_2 \in \bangSetTerms$: We set $s' :=
            \app{u'_2}{u'_1}$ since by contextual closure $u_1 =
            \app{t_1}{u'_1} \bangArr_{S_\indexOmega}<s!>
            \app{u'_2}{u'_1} = s'$ and $u_2 = \app{u'_2}{t_2}
            \bangArr_{S_\indexOmega}<R> \app{u'_2}{u'_1} = s'$.
            Graphically,
            \begin{equation*}
                \begin{array}{lll}
                    \app{t_1}{t_2}                          &\bangArr_{S_\indexOmega}<R>    &  \app{t_1}{u'_1}
                \\[0.2cm]
                    \;\;\bangDownArr_{S_\indexOmega}<s!>    &                               &\;\;\bangDownArr_{S_\indexOmega}<s!>
                \\[0.2cm]
                     \app{u'_2}{t_2}                        &\bangArr_{S_\indexOmega}<R>    &\app{u'_2}{u'_1}   
                 \end{array}
            \end{equation*}

        \item[\bltIII] $u_1 = \app{t_1}{u'_1}$ and $u_2 =
            \app{t_1}{u'_2}$ with $t_2 \bangArr_{S_\indexOmega}<R>
            u'_1$ and $t_2 \bangArr_{S_\indexOmega}<s!> u'_2$ for some
            $u'_1, u'_2 \in \bangSetTerms$: By \ih on $t_2$, there
            exists $s' \in \bangSetTerms$ such that $u'_1
            \bangArr_{S_\indexOmega}<s!> s'$ and $u'_2
            \bangArr*_{S_\indexOmega}<R> s'$. We set $s :=
            \app{t_1}{s'}$ since by contextual closure $u_1 =
            \app{t_1}{u'_1} \bangArr_{S_\indexOmega}<s!> \app{t_1}{s'}
            = s'$ and $u_2 = \app{t_1}{u'_2}
            \bangArr*_{S_\indexOmega}<R> \app{t_1}{s'} = s'$.
            Graphically,
            \begin{equation*}
                \begin{array}{lll}
                    \app{t_1}{t_2}                              &\bangArr_{S_\indexOmega}<R>    &  \app{t_1}{u'_1}
                \\[0.2cm]
                    \;\;\bangDownArr_{S_\indexOmega}<s!>   &                               &\;\;\bangDownArr_{S_\indexOmega}<s!>
                \\[0.2cm]
                    \app{t_1}{u'_2}                                   &\bangArr*_{S_\indexOmega}<R>  & \app{t_1}{s'} 
                 \end{array}
            \end{equation*}

        \end{itemize}
    \end{itemize}

\item[\bltI] $t = t_1\esub{x}{t_2}$: We distinguish two cases:
    \begin{itemize}
    \item[\bltII] $t$ is the $\bangSymbSubs$-redex reduced in the step
        $t \bangArr_{S_\indexOmega}<s!> u_2$: Then $t_2 =
        \bangLCtxt<\oc t'_2>$ and $u_2 =
        \bangLCtxt<t_1\isub{x}{t'_2}>$ for some context $\bangLCtxt$
        and some term $t'_2 \in \bangSetTerms$. Since
        $t_1\esub{x}{\bangLCtxt<\oc t'_2>} \bangArr_{S_\indexOmega}<R>
        u_1$ then three cases can be distinguished:
        \begin{itemize}
        \item[\bltIII] $t_1 \bangArr_{S_\indexOmega}<R> t'_1$ and $u_1
            = t'_1\esub{x}{t_2}$ for some term $t'_1 \in
            \bangSetTerms$: We set $s :=
            \bangLCtxt<t'_1\isub{x}{t'_2}>$ which concludes this case
            since by induction on $t_1$, one has that
            $t_1\isub{x}{t'_2} \bangArr_{S_\indexOmega}<R>
            t'_1\isub{x}{t'_2}$ thus by contextual closure $u_1 =
            t'_1\esub{x}{\bangLCtxt<\oc t'_2>}
            \bangArr_{S_\indexOmega}<s!>
            \bangLCtxt<t'_1\isub{x}{t'_2}> = s$ and $u_2 =
            \bangLCtxt<t_1\isub{x}{t'_2}> \bangArr_{S_\indexOmega}<R>
            \bangLCtxt<t'_1\isub{x}{t'_2}> = s$ Graphically,
            \begin{equation*}
                \begin{array}{lll}
                    t_1\esub{x}{\bangLCtxt<\oc t'_2>}       &\bangArr_{S_\indexOmega}<R>    &t'_1\esub{x}{\bangLCtxt<\oc t'_2>}
                \\[0.2cm]
                    \;\;\bangDownArr_{S_\indexOmega}<s!>    &                               &\;\;\bangDownArr_{S_\indexOmega}<s!>
                \\[0.2cm]
                    \bangLCtxt<t_1\isub{x}{t'_2}>           &\bangArr_{S_\indexOmega}<R>    &\bangLCtxt<t'_1\isub{x}{t'_2}>
            \end{array}
            \end{equation*}

        \item[\bltIII] $t'_2 \bangArr_{S_\indexOmega}<R> t''_2$ and
            $u_1 = t_1\esub{x}{\bangLCtxt<\oc t''_2>}$ for some term
            $t'_2 \in \bangSetTerms$: We set $s :=
            \bangLCtxt<t_1\isub{x}{t''_2}>$ which concludes this case
            since by induction on $t_1$, one has that
            $t_1\isub{x}{t'_2} \bangArr*_{S_\indexOmega}<R>
            t_1\isub{x}{t''_2}$ thus by contextual closure $u_1 =
            t_1\esub{x}{\bangLCtxt<\oc t''_2>}
            \bangArr_{S_\indexOmega}<s!>
            \bangLCtxt<t_1\isub{x}{t''_2}> = s$ and $u_2 =
            \bangLCtxt<t_1\isub{x}{t'_2}> \bangArr*_{S_\indexOmega}<R>
            \bangLCtxt<t_1\isub{x}{t''_2}> = s$. Graphically,
            \begin{equation*}
                \begin{array}{lll}
                    t_1\esub{x}{\bangLCtxt<\oc t'_2>}       &\bangArr_{S_\indexOmega}<R>    &t_1\esub{x}{\bangLCtxt<\oc t''_2>}
                \\[0.2cm]
                    \;\;\bangDownArr_{S_\indexOmega}<s!>    &                               &\;\;\bangDownArr_{S_\indexOmega}<s!>
                \\[0.2cm]
                   \bangLCtxt<t_1\isub{x}{t'_2}>            &\bangArr^*_{S_\indexOmega}<R>  &\bangLCtxt<t_1\isub{x}{t''_2}>
            \end{array}
            \end{equation*}

        \item[\bltIII] $\bangLCtxt \bangArr_{S_\indexOmega}<R>
            \bangLCtxt'$ and $u_1 = t_1\esub{x}{\bangLCtxt'<\oc
            t'_2>}$ for some context $\bangLCtxt'$: We set $s :=
            \bangLCtxt'<t_1\isub{x}{t'_2}>$ which concludes this case
            since by
            \Cref{lem:Bang_Full_dB_d!_Context_Reduction_Lifted_to_Terms},
            $u_1 = t_1\esub{x}{\bangLCtxt'<\oc t'_2>}
            \bangArr_{S_\indexOmega}<s!>
            \bangLCtxt'<t_1\isub{x}{t'_2}> = s$ and $u_2 =
            \bangLCtxt<t_1\isub{x}{t'_2}> \bangArr_{S_\indexOmega}<R>
            \bangLCtxt'<t_1\isub{x}{t'_2}> = s$. Graphically,
            \begin{equation*}
                \begin{array}{lll}
                    t_1\esub{x}{\bangLCtxt<\oc t'_2>}       &\bangArr_{S_\indexOmega}<R>    &   t_1\esub{x}{\bangLCtxt'<\oc t'_2>}
                \\[0.2cm]
                    \;\;\bangDownArr_{S_\indexOmega}<s!>    &                               &\;\;\bangDownArr_{S_\indexOmega}<s!>
                \\[0.2cm]
                   \bangLCtxt<t_1\isub{x}{t'_2}>            &\bangArr_{S_\indexOmega}<R>    & \bangLCtxt'<t_1\isub{x}{t'_2}>
            \end{array}
            \end{equation*}
        \end{itemize}

    \item[\bltII] Otherwise: By hypothesis $t
        \bangArr_{S_\indexOmega}<R> u_1$ and $t
        \bangArr_{S_\indexOmega}<s!> u_2$, and four cases can thus be
        distinguished:
        \begin{itemize}
        \item[\bltIII] $u_1 = u'_1\esub{x}{t_2}$ and $u_2 =
            u'_2\esub{x}{t_2}$ with $t_1 \bangArr_{S_\indexOmega}<R>
            u'_1$ and $t_1 \bangArr_{S_\indexOmega}<s!> u'_2$ for some
            $u'_1, u'_2 \in \bangSetTerms$: By \ih on $t_1$, there
            exists $s' \in \bangSetTerms$ such that $u'_1
            \bangArr_{S_\indexOmega}<s!> s'$ and $u'_2
            \bangArr*_{S_\indexOmega}<R> s'$. We set $s :=
            s'\esub{x}{t_2}$ since by contextual closure $u_1 =
            u'_1\esub{x}{t_2} \bangArr_{S_\indexOmega}<s!>
            s'\esub{x}{t_2} = s'$ and $u_2 = u'_2\esub{x}{t_2}
            \bangArr*_{S_\indexOmega}<R> u'_2\esub{x}{t_2} = s'$.
            Graphically,
            \begin{equation*}
                \begin{array}{lll}
                    t_1\esub{x}{t_2}                        &\bangArr_{S_\indexOmega}<R>    &  u'_1\esub{x}{t_2}
                \\[0.2cm]
                    \;\;\bangDownArr_{S_\indexOmega}<s!>    &                               &\;\;\bangDownArr_{S_\indexOmega}<s!>
                \\[0.2cm]
                    u'_2\esub{x}{t_2}                       &\bangArr*_{S_\indexOmega}<R>   & s'\esub{x}{t_2}   
              \end{array}
            \end{equation*}

        \item[\bltIII] $u_1 = u'_1\esub{x}{t_2}$ and $u_2 =
            t_1\esub{x}{u'_2}$ with $t_1 \bangArr_{S_\indexOmega}<R>
            u'_1$ and $t_2 \bangArr_{S_\indexOmega}<s!> u'_2$ for some
            $u'_1, u'_2 \in \bangSetTerms$: We set $s' :=
            u'_1\esub{x}{u'_2}$ since by contextual closure $u_1 =
            u'_1\esub{x}{t_2} \bangArr_{S_\indexOmega}<s!>
            u'_1\esub{x}{u'_2} = s'$ and $u_2 = u'_1\esub{x}{u'_2}
            \bangArr_{S_\indexOmega}<R> u'_1\esub{x}{u'_2} = s'$.
            Graphically,
            \begin{equation*}
                \begin{array}{lll}
                    t_1\esub{x}{t_2}                        &\bangArr_{S_\indexOmega}<R>    &u'_1\esub{x}{t_2}
                \\[0.2cm]
                    \;\;\bangDownArr_{S_\indexOmega}<s!>    &                               &\;\;\bangDownArr_{S_\indexOmega}<s!>
              \\[0.2cm]
                    t_1\esub{x}{u'_2}                       &\bangArr_{S_\indexOmega}<R>    &u'_1\esub{x}{u'_2}   
              \end{array}
            \end{equation*}

        \item[\bltIII] $u_1 = t_1\esub{x}{u'_1}$ and $u_2 =
            u'_2\esub{x}{t_2}$ with $t_2 \bangArr_{S_\indexOmega}<R>
            u'_1$ and $t_1 \bangArr_{S_\indexOmega}<s!> u'_2$ for some
            $u'_1, u'_2 \in \bangSetTerms$: We set $s' :=
            u'_2\esub{x}{u'_1}$ since by contextual closure $u_1 =
            t_1\esub{x}{u'_1} \bangArr_{S_\indexOmega}<s!>
            u'_2\esub{x}{u'_1} = s'$ and $u_2 = u'_2\esub{x}{t_2}
            \bangArr_{S_\indexOmega}<R> u'_2\esub{x}{u'_1} = s'$.
            Graphically,
            \begin{equation*}
                \begin{array}{lll}
                    t_1\esub{x}{t_2}                        &\bangArr_{S_\indexOmega}<R>    &  t_1\esub{x}{u'_1}
                \\[0.2cm]
                    \;\;\bangDownArr_{S_\indexOmega}<s!>    &                               &\;\;\bangDownArr_{S_\indexOmega}<s!>
                \\[0.2cm]
                    u'_2\esub{x}{t_2}                       &\bangArr_{S_\indexOmega}<R>    & u'_2\esub{x}{u'_1}
              \end{array}
            \end{equation*}

        \item[\bltIII] $u_1 = t_1\esub{x}{u'_1}$ and $u_2 =
            t_1\esub{x}{u'_2}$ with $t_2 \bangArr_{S_\indexOmega}<R>
            u'_1$ and $t_2 \bangArr_{S_\indexOmega}<s!> u'_2$ for some
            $u'_1, u'_2 \in \bangSetTerms$: By \ih on $t_2$, there
            exists $s' \in \bangSetTerms$ such that $u'_1
            \bangArr_{S_\indexOmega}<s!> s'$ and $u'_2
            \bangArr*_{S_\indexOmega}<R> s'$. We set $s :=
            t_1\esub{x}{s'}$ since by contextual closure $u_1 =
            t_1\esub{x}{u'_1} \bangArr_{S_\indexOmega}<s!>
            t_1\esub{x}{s'} = s'$ and $u_2 = t_1\esub{x}{u'_2}
            \bangArr*_{S_\indexOmega}<R> t_1\esub{x}{s'} = s'$.
            Graphically,
            \begin{equation*}
                \begin{array}{lll}
                    t_1\esub{x}{t_2}                        &\bangArr_{S_\indexOmega}<R>    &t_1\esub{x}{u'_1}
                \\[0.2cm]
                    \;\;\bangDownArr_{S_\indexOmega}<s!>    &                               &\;\;\bangDownArr_{S_\indexOmega}<s!>
                \\[0.2cm]
                    t_1\esub{x}{u'_2}                       &\bangArr^*_{S_\indexOmega}<R>  &t_1\esub{x}{s'}
              \end{array}
            \end{equation*}
        \end{itemize}
    \end{itemize}

\item[\bltI] $t = \der{t'_0}$:  We distinguish two cases:
    \begin{itemize}
    \item[\bltII] $t$ is the $\bangSymbBang$-redex reduced in the step
        $t \bangArr_{S_\indexOmega}<d!> u_1$: Then $t =
        \der{\bangLCtxt<\oc t'>}$ and $u_1 = \bangLCtxt<t'>$ for some
        context $\bangLCtxt$ and some term $t' \in \bangSetTerms$. By
        hypothesis $\der{\bangLCtxt<\oc t'>}
        \bangArr_{S_\indexOmega}<s!> u_2$ so that three different
        cases can be distinguished:
        \begin{itemize}
        \item[\bltIII] $t' \bangArr_{S_\indexOmega}<s!> t''$ and $u_2
            = \der{\bangLCtxt<\oc t''>}$ for some $t'' \in
            \bangSetTerms$: We set $s := \bangLCtxt<t''>$ which
            concludes this case since by contextual closure $u_1 =
            \bangLCtxt<t'> \bangArr_{S_\indexOmega}<s!>
            \bangLCtxt<t''> = s'$ and $u_2 = \der{\bangLCtxt<\oc t''>}
            \bangArr_{S_\indexOmega}<d!> \bangLCtxt<t''> = s'$.
            Graphically,
            \begin{equation*}
                \begin{array}{lll}
                    \der{\bangLCtxt<\oc t'>}                &\bangArr_{S_\indexOmega}<d!>   &\bangLCtxt<t'>
                \\[0.2cm]
                    \;\;\bangDownArr_{S_\indexOmega}<s!>    &                               &\;\;\bangDownArr_{S_\indexOmega}<s!>
                \\[0.2cm]
                    \der{\bangLCtxt<\oc t''>}               &\bangArr_{S_\indexOmega}<d!>   &\bangLCtxt<t''>
                \end{array}
            \end{equation*}

        \item[\bltIII] $\bangLCtxt \bangArr_{S_\indexOmega}<s!>
            \bangLCtxt'$ and $u_2 = \der{\bangLCtxt'<\oc t'>}$ for
            some context $\bangLCtxt'$: We set $s := \bangLCtxt'<t'>$
            which concludes this case since by
            \Cref{lem:Bang_Full_dB_d!_Context_Reduction_Lifted_to_Terms}
            $u_1 = \bangLCtxt<t'> \bangArr_{S_\indexOmega}<s!>
            \bangLCtxt'<t'> = s$ and $u_2 = \der{\bangLCtxt'<\oc t'>}
            \bangArr_{S_\indexOmega}<d!> \bangLCtxt'<t'> = s$.
            Graphically,
            \begin{equation*}
                \begin{array}{lll}
                    \der{\bangLCtxt<\oc t'>}                &\bangArr_{S_\indexOmega}<d!>   &\bangLCtxt<t'>
                \\[0.2cm]
                    \;\;\bangDownArr_{S_\indexOmega}<s!>    &                               &\;\;\bangDownArr_{S_\indexOmega}<s!>
                \\[0.2cm]
                    \der{\bangLCtxt'<\oc t'>}               &\bangArr_{S_\indexOmega}<d!>   &\bangLCtxt'<t'>
                \end{array}
            \end{equation*}

        \item[\bltIII] $\bangLCtxt =
            \bangLCtxt_1<\bangLCtxt_3\esub{x}{\bangLCtxt_2<\oc t''>}>$
            and $u_2 =
            \der{\bangLCtxt_1<\bangLCtxt_2<\bangLCtxt_3\isub{x}{t''}\bangCtxtPlug{\oc
            t'\isub{x}{t''}}>>}$ for some context $\bangLCtxt_1,
            \bangLCtxt_2, \bangLCtxt_3$ and some term $t'' \in
            \bangSetTerms$: We set $s =
            \bangLCtxt_1<\bangLCtxt_2<\bangLCtxt_3\isub{x}{t''}\bangCtxtPlug{t'}>>$
            which concludes this case since by contextual closure $u_1
            = \bangLCtxt_1<\bangLCtxt_3<t'>\esub{x}{\bangLCtxt_2<\oc
            t''>}> \bangArr_{S_\indexOmega}<s!>
            \bangLCtxt_1<\bangLCtxt_2<\bangLCtxt_3\isub{x}{t''}\bangCtxtPlug{t'\isub{x}{t''}}>>
            = s$ and $u_2 =
            \der{\bangLCtxt_1<\bangLCtxt_2<\bangLCtxt_3\isub{x}{t''}\bangCtxtPlug{\oc
            t'\isub{x}{t''}}>>} \bangArr_{S_\indexOmega}<d!>
            \bangLCtxt_1<\bangLCtxt_2<\bangLCtxt_3\isub{x}{t''}\bangCtxtPlug{t'\isub{x}{t''}}>>
            = s$. Graphically,
            \begin{equation*}
                \begin{array}{lll}
                    \der{\bangLCtxt_1<\bangLCtxt_3<\oc t'>\esub{x}{\bangLCtxt_2<\oc t''>}>}                         &\bangArr_{S_\indexOmega}<R>    &\bangLCtxt_1<\bangLCtxt_3<t'>\esub{x}{\bangLCtxt_2<\oc t''>}>
                \\[0.2cm]
                    \;\;\bangDownArr_{S_\indexOmega}<s!>                                                            &                               &\;\;\bangDownArr_{S_\indexOmega}<s!>
                \\[0.2cm]
                    \der{\bangLCtxt_1<\bangLCtxt_2<\bangLCtxt_3\isub{x}{t''}\bangCtxtPlug{\oc t'\isub{x}{t''}}>>}   &\bangArr_{S_\indexOmega}<R>    &\bangLCtxt_1<\bangLCtxt_2<\bangLCtxt_3\isub{x}{t''}\bangCtxtPlug{t'\isub{x}{t''}}>>   
                \end{array}
            \end{equation*}
        \end{itemize}

    \item[\bltII] Otherwise: Since by hypothesis $t
        \bangArr_{S_\indexOmega}<R> u_1$ and $t
        \bangArr_{S_\indexOmega}<s!> u_2$, then necessarily $u_1 =
        \der{u'_1}$ and $u_2 = \der{u'_2}$ for some $u'_1, u'_2 \in
        \bangSetTerms$ such that $t'_0 \bangArr_{S_\indexOmega}<R>
        u'_1$ and $t'_0 \bangArr_{S_\indexOmega}<s!> u'_2$. By \ih on
        $t'$, there exists $s' \in \bangSetTerms$ such that $u'_1
        \bangArr_{S_\indexOmega}<s!> s'$ and $u'_2
        \bangArr*_{S_\indexOmega}<R> s'$. We set $s := \der{s'}$ since
        by contextual closure $u_1 = \der{u'_1}
        \bangArr_{S_\indexOmega}<s!> \der{s'} = s'$ and $u_2 =
        \der{u'_2} \bangArr*_{S_\indexOmega}<R> \der{s'} = s'$.
        Graphically,
        \begin{equation*}
            \begin{array}{lll}
                \der{t'_0}                              &\bangArr_{S_\indexOmega}<R>    &\der{u'_1}
            \\[0.2cm]
                \;\;\bangDownArr_{S_\indexOmega}<s!>    &                               &\;\;\bangDownArr_{S_\indexOmega}<s!>
            \\[0.2cm]
                \der{u'_1}                              &\bangArr_{S_\indexOmega}<R>    &\der{s'}
            \end{array}
        \end{equation*}
    \end{itemize}

\item[\bltI] $t = \oc t'$: Since by hypothesis $t
    \bangArr_{S_\indexOmega}<R> u_1$ and $t
    \bangArr_{S_\indexOmega}<s!> u_2$, then necessarily $u_1 = \oc
    u'_1$ and $u_2 = \oc u'_2$ for some $u'_1, u'_2 \in \bangSetTerms$
    such that $t \bangArr_{S_\indexOmega}<R> u'_1$ and $t
    \bangArr_{S_\indexOmega}<s!> u'_2$. By \ih on $t'$, there exists
    $s' \in \bangSetTerms$ such that $u'_1
    \bangArr_{S_\indexOmega}<s!> s'$ and $u'_2
    \bangArr*_{S_\indexOmega}<R> s'$. We set $s := \oc s'$ since by
    contextual closure $u_1 = \oc u'_1 \bangArr_{S_\indexOmega}<s!>
    \oc s' = s'$ and $u_2 = \oc u'_2 \bangArr*_{S_\indexOmega}<R> \oc
    s' = s'$. Graphically,
    \begin{equation*}
        \begin{array}{lll}
            \oc t'                                  &\bangArr_{S_\indexOmega}<R>    &\oc u'_1
        \\[0.2cm]
            \;\;\bangDownArr_{S_\indexOmega}<s!>    &                               &\;\;\bangDownArr_{S_\indexOmega}<s!>
        \\[0.2cm]
          \oc u'_2                                  &\bangArr*_{S_\indexOmega}<R>   &\oc s'
        \end{array}
    \end{equation*}
\end{itemize}
\end{proof}

} \deliaLu \giulioLu

\begin{corollary}
    \label{lem:Bang_Full_dB_U_d!_and_s!_Strongly_Commute}%
    The reductions $\bangArr_{S_\indexOmega}<dB> \cup
    \bangArr_{S_\indexOmega}<d!>$ and $\bangArr_{S_\indexOmega}<s!>$
    commute.
\end{corollary}
\stableProof{
\begin{proof}
    Immediate consequence of
    \Cref{lem:Bang_Full_dB_U_d!_and_s!_strongly_commute}.
\end{proof}
} \deliaLu \giulioLu%

\RecBangSurfaceFullConfluence*
\label{prf:Bang_Surface_Full_Confluence}%
\stableProof{
\begin{proof}
\begin{itemize}
\item[\bltI] \textbf{(\surfaceTxt^)} Proof of confluence of the
    \surfaceTxt reduction can be found in \cite{BucciarelliKesnerRiosViso20}.

\item[\bltI] \textbf{(\fullTxt^)} By
  \Cref{lem:Bang_Full_dB_U_d!_Confluent,lem:Bang_Full_s!_Confluent,lem:Bang_Full_dB_U_d!_and_s!_Strongly_Commute}, using the Hindley--Rosen lemma
  \cite[Prop. 3.3.5]{barendregt84nh}.
    \qedhere
\end{itemize}
\end{proof}
} \giulioLu \deliaLu%

}

\section{Proofs of \Cref{sec:Meaningfulness}}

\begin{lemma}
    \label{lem:Bang_building_testing_ctxt_from_type_ctxt}%
    Let $\Pi \bangBKRVTr \Gamma \vdash t : \sigma$. If
    $\bangBKRVInhPred{\Gamma}$ then there are $\bangTCtxt$ and
    $\Pi'$ such that $\Pi' \bangBKRVTr \emptyset \vdash \bangTCtxt<t>
    : \sigma$, in particular $\bangBKRVInhPred{\sigma}$.
\end{lemma}
\stableProof{%
    \begin{proof}
Let $\Pi \bangBKRVTr \Gamma \vdash t : \sigma$. We reason by induction
on the number of variables in $\typeCtxtDom{\Gamma}$:
\begin{itemize}
\item[\bltI] $\Gamma = \emptyset$: We set $\bangTCtxt \coloneqq \Hole$, so $\bangTCtxt<t> = t $ and we are done
     since $\Pi \bangBKRVTr \emptyset \vdash
    t : \sigma$ holds by hypothesis.

\item[\bltI] $\Gamma = \Gamma', x : \M$ with $\M \neq \emptymset$:
    Since $\bangBKRVInhPred{\Gamma}$, then $\bangBKRVInhPred{\Gamma'}$
    and $\bangBKRVInhPred{\M}$, thus there exists $\Pi_u \bangBKRVTr
    \emptyset \vdash u : \M$ for some $u \in \bangSetTerms$. Consider
    the following derivation:
    \begin{equation*}
        \begin{prooftree}
            \hypo{\Pi \bangBKRVTr \Gamma', x : \M \vdash t : \sigma}
            \inferBangBKRVAbs{\Gamma' \vdash \abs{x}{t} : \M \typeArrow \sigma}
            \hypo{\Pi_u \bangBKRVTr \emptyset \vdash u : \M}
            \inferBangBKRVEs{\Gamma' \vdash \app[\,]{(\abs{x}{t})}{u} : \sigma}
        \end{prooftree}
    \end{equation*}
    By the \ih on $\Gamma'$, one deduces that there exists $\Pi'
    \bangBKRVTr \emptyset \vdash
    \bangTCtxt'<\app[\,]{(\abs{x}{t})}{u}> : \sigma$. This concludes
    this case by taking $\bangTCtxt =
    \bangTCtxt'<\app[\,]{(\abs{x}{\Hole})}{u}>$.
    \qedhere
\end{itemize}
\end{proof}

} \deliaLu \giulioLu

\begin{lemma}
    \label{lem:Bang_building_testing_ctxt_from_arg_types}%
    Let $\bangBKRVTr \Gamma \vdash t : \sigma$. If
    $\bangBKRVInhPred{\bangTypeArgs{\sigma}}$, then there are a
    testing context $\bangTCtxt$ and a derivation $\bangBKRVTr \Gamma
    \vdash \bangTCtxt<t> : \tau$ for some type $\tau$ not functional.
\end{lemma}
\stableProof{%
    \begin{proof}
Let $\Pi \bangBKRVTr \Gamma \vdash t : \sigma$. We reason by induction
on $\sigma$:
\begin{itemize}
\item[\bltI] $\sigma = \alpha$ or $\sigma = \M$: We set $\bangTCtxt
    \coloneqq \Hole$ and $\tau \coloneqq \sigma$ which concludes this
    case since $\tau$ is not functional, $\bangTCtxt<t> = t$ and
    $\bangBKRVTr \Gamma \vdash t : \sigma$ holds by hypothesis.

\item[\bltI] $\sigma = \M \typeArrow \mu$: Since
	$\bangBKRVInhPred{\genericArgs{\bangBKRVTypeSys}{\sigma}}$, then
	$\bangBKRVInhPred{\genericArgs{\bangBKRVTypeSys}{\mu}}$ and $\bangBKRVInhPred{\M}$ thus
	there exists $\Pi_u \bangBKRVTr \emptyset \vdash u : \M$ for some
	$u \in \bangSetTerms$. Consider the following derivation:
    \begin{equation*}
        \begin{prooftree}
            \hypo{\Pi \bangBKRVTr \Gamma \vdash t : \M \typeArrow \mu}
            \hypo{\Pi_u \bangBKRVTr \emptyset \vdash u : \M}
            \inferBangBKRVApp{\Gamma \vdash \app[\,]{t}{u} : \mu}
        \end{prooftree}
    \end{equation*}
    By \ih on $\mu$, there is a testing context $\bangTCtxt'$ and a
    derivation $\bangBKRVTr \Gamma \vdash \bangTCtxt'<\app[\,]{t}{u}>
    : \tau$ for some type $\tau$ not functional. Setting $\bangTCtxt =
    \bangTCtxt'<\app[\,]{\Hole}{u}>$ concludes this case.
    \qedhere
\end{itemize}
\end{proof}

} \deliaLu \giulioLu %

\RecBangInhabitationFromTesting*
\label{prf:Bang_O_|-_T<t>_:_[]_==>_Gam_|-_t_:_sig_and_Gam_and_args(sig)_inh}
\stableProof{%
    \begin{proof}
We prove the stronger statement below, to have the right induction
hypothesis:
\begin{center}
    Let $\bangBKRVTr \Gamma \vdash \bangTCtxt<t> : \sigma$ with
    $\bangBKRVInhPred{\Gamma}$ and
    $\bangBKRVInhPred{\genericArgs{\bangBKRVTypeSys}{\sigma}}$, then $\bangBKRVTr \Gamma'
    \vdash t : \sigma'$ with $\bangBKRVInhPred{\Gamma'}$ and
    $\bangBKRVInhPred{\genericArgs{\bangBKRVTypeSys}{\sigma'}}$.
\end{center}

From that,
\Cref{lem:Bang_O_|-_T<t>_:_[]_==>_Gam_|-_t_:_sig_and_Gam_and_args(sig)_inh}
follows immediately, because $\bangBKRVInhPred{\emptyset}$ is
vacuously true.

Let $\Pi \bangBKRVTr \Gamma \vdash \bangTCtxt<t> : \sigma$ with
$\bangBKRVInhPred{\Gamma}$ and $\bangBKRVInhPred{\genericArgs{\bangBKRVTypeSys}{\sigma}}$.
We reason by induction on $\bangTCtxt$:
\begin{itemize}
\item[\bltI] $\bangTCtxt = \Hole$: Trivial by hypothesis.

\item[\bltI] $\bangTCtxt = \app[\,]{(\abs{x}{\bangTCtxt'})}{u}$: Then
    $\Pi$ has the following form:
    \begin{equation*}
        \begin{prooftree}
            \hypo{\Pi_1 \bangBKRVTr \Gamma_1; x : \M \vdash \bangTCtxt'<t> : \sigma}
            \inferBangBKRVAbs{\Gamma_1 \vdash \abs{x}{\bangTCtxt'<t>} : \M \typeArrow \sigma}
            \hypo{\Pi_2 \bangBKRVTr \Gamma_2 \vdash u : \M}
            \inferBangBKRVApp{\Gamma_1 + \Gamma_2 \vdash \app[\,]{(\abs{x}{\bangTCtxt'<t>})}{u} : \sigma}
        \end{prooftree}
    \end{equation*}
    with $\Gamma = \Gamma_1 + \Gamma_2$. By hypothesis
    $\bangBKRVInhPred{\Gamma}$, thus in particular
    $\bangBKRVInhPred{\Gamma_1}$ and $\bangBKRVInhPred{\Gamma_2}$. Using
    \Cref{lem:Bang_building_testing_ctxt_from_type_ctxt} on $\Pi_2$,
    one deduces $\bangBKRVInhPred{\M}$ and thus $\bangBKRVInhPred{\Gamma_1; x
    : \M}$. By \ih on $\Pi_1$, one obtains $\Pi' \bangBKRVTr \Gamma'
    \vdash t : \sigma'$ with $\bangBKRVInhPred{\Gamma'}$ and
    $\bangBKRVInhPred{\genericArgs{\bangBKRVTypeSys}{\sigma'}}$.

\item[\bltI] $\bangTCtxt = \app{\bangTCtxt'}{u}$: Then $\Pi$ has the
    following form:
    \begin{equation*}
        \begin{prooftree}
            \hypo{\Pi_1 \bangBKRVTr \Gamma_1 \vdash \bangTCtxt'<t> : \M \typeArrow \sigma}
            \hypo{\Pi_2 \bangBKRVTr \Gamma_2 \vdash u : \M}
            \inferBangBKRVApp{\Gamma_1 + \Gamma_2 \vdash \app{\bangTCtxt'<t>}{u} : \sigma}
        \end{prooftree}
    \end{equation*}
    with $\Gamma = \Gamma_1 + \Gamma_2$. By hypothesis
    $\bangBKRVInhPred{\Gamma}$, thus in particular
    $\bangBKRVInhPred{\Gamma_1}$ and $\bangBKRVInhPred{\Gamma_2}$. Using
    \Cref{lem:Bang_building_testing_ctxt_from_type_ctxt} on $\Pi_2$,
    one deduces $\bangBKRVInhPred{\M}$, and thus
    $\bangBKRVInhPred{\genericArgs{\bangBKRVTypeSys}{\M \typeArrow \sigma}}$. By the \ih on
    $\Pi_1$, one obtains $\Pi' \bangBKRVTr \Gamma' \vdash t : \sigma'$
    with $\bangBKRVInhPred{\Gamma'}$ and
    $\bangBKRVInhPred{\genericArgs{\bangBKRVTypeSys}{\sigma'}}$.
    \qedhere
\end{itemize}
\end{proof}

} \deliaLu \giulioLu

\begin{lemma}
    \label{lem:Bang_O_|-_t_:_M_and_SNF(t)_==>_t_=_!u}%
    Let $\Pi \bangBKRVTr \emptyset \vdash t : \sigma$ where $\sigma$
    is not functional. If $t$ is a $\bangStratCtxt$-normal form,
    then $t = \oc u$ for some $u \in \bangSetTerms$.
\end{lemma}
\stableProof{%
    \begin{proof} \deliaLu \giulioLu We prove the stronger statement
below:
\begin{center}
    Let $\bangBKRVTr \emptyset \vdash t : \sigma$ with $\sigma$ not
    functional (\resp functional). If $t$ is a
    $\bangStratCtxt$-normal form, then $t = \oc u$ (\resp $t =
    \abs{x}{u}$) for some $u \in \bangSetTerms$.
\end{center}
Let $\Pi \bangBKRVTr \emptyset \vdash t : \sigma$. By induction on
$t$:
\begin{itemize}
\item[\bltI] $t = x$: Then $\Pi$ can only have the following form:
    \begin{equation*}
        \begin{prooftree}
            \inferBangBKRVVar{x : \mset{\sigma} \vdash x : \sigma}
        \end{prooftree}
    \end{equation*}
    which contradicts the hypothesis that the type environment of
    $\Pi$ is empty.

\item[\bltI] $t = \abs{x}{t'}$: Then $\Pi$ has the following form:
    \begin{equation*}
        \begin{prooftree}
            \hypo{\Pi' \bangBKRVTr x : \N \vdash t' : \tau}
            \inferBangBKRVAbs{\emptyset \vdash \abs{x}{t'} : \N \typeArrow \tau}
        \end{prooftree}
    \end{equation*}
    where $\sigma = \N \typeArrow \tau$ is functional, which concludes
    this case.

\item[\bltI] $t = \app{t_1}{t_2}$: Then in particular $t_1$ is a
    $\bangStratCtxt$-normal form and $\Pi$ has the following form:
    \begin{equation*}
        \begin{prooftree}
            \hypo{\Pi_1 \bangBKRVTr \emptyset \vdash t_1 : \N \typeArrow \sigma}
            \hypo{\Pi_2 \bangBKRVTr \emptyset \vdash t_2 : \N}
            \inferBangBKRVApp{\emptyset \vdash \app{t_1}{t_2} : \sigma}
        \end{prooftree}
    \end{equation*}
    By the \ih on $t_1$, $t_1 = \abs{x}{t'_1}$ which contradicts the
    hypothesis that $t$ is a $\bangStratCtxt$-normal form.

\item[\bltI] $t = t_1\esub{x}{t_2}$: Then in particular $t_2$ is a
    $\bangStratCtxt$-normal form and $\Pi$ has the following form:
    \begin{equation*}
        \begin{prooftree}
            \hypo{\Pi_1 \bangBKRVTr x : \M \vdash t_1 : \sigma}
            \hypo{\Pi_2 \bangBKRVTr \emptyset \vdash t_2 : \M}
            \inferBangBKRVEs{\emptyset \vdash t_1\esub{x}{t_2} : \sigma}
        \end{prooftree}
    \end{equation*}
    By the \ih on $t_2$, $t_2 = \oc t'_2$ which contradicts the
    hypothesis that $t$ is a $\bangStratCtxt$-normal form.

\item[\bltI] $t = \der{t'}$: Then $t'$ is a $\bangStratCtxt$-normal
    form and $\Pi$ has the following form:
    \begin{equation*}
        \begin{prooftree}
            \hypo{\Pi' \bangBKRVTr \emptyset \vdash t' : \mset{\sigma}}
            \inferBangBKRVDer{\emptyset \vdash \der{t'} : \sigma}
        \end{prooftree}
    \end{equation*}
    By the \ih on $t'$, $t' = \oc t'$ which contradicts the hypothesis
    that $t$ is a $\bangStratCtxt$-normal form.

\item[\bltI] $t = \oc t'$: Then $\Pi$ has the following form:
    \begin{equation*}
        \begin{prooftree}
            \hypo{\Pi_i \bangBKRVTr \emptyset \vdash u : \tau_i}
            \delims{\left(}{\right)_{i \in I}}
            \inferBangBKRVBg{\quad\emptyset \vdash \oc t : \mset{\tau_i}_{i \in I}\quad}
        \end{prooftree}
    \end{equation*}
    where $\sigma = \mset{\tau_i}_{i \in I}$ is not functional, which
    concludes this case.
    \qedhere
\end{itemize}
\end{proof}
} \deliaLu \giulioLu

\RecBangMeaningfulTypabilityInhabitation*%
\label{prf:Bang_Meaningfulness_Typing_and_Inhabitation}%
\stableProof{%
    \begin{proof} \deliaLu \giulioLu%
\begin{description}
\item[$(\Rightarrow)$] Let $t$ be meaningful. Then there exists a
    testing context $\bangTCtxt$ and a term $u \in \bangSetTerms$ such
    that $\bangTCtxt<t> \bangArr*_S \oc u$. Using rule
    $\bangBKRVBgRuleName$ without any premises, the following typing
    derivation holds $\bangBKRVTr \emptyset \vdash \oc u :
    \emptymset$. By subject expansion
    (\Cref{lem:Bang_BKRV}.\ref{lem:Bang_BKRV_Surface_Typing_SRE}), one
    deduces that $\bangBKRVTr \emptyset \vdash \bangTCtxt<t> :
    \emptymset$. Finally, by
    \Cref{lem:Bang_O_|-_T<t>_:_[]_==>_Gam_|-_t_:_sig_and_Gam_and_args(sig)_inh},
    one concludes that $\bangBKRVTr \Gamma \vdash t : \sigma$ with
    $\bangBKRVInhPred{\Gamma}$ and
    $\bangBKRVInhPred{\bangTypeArgs{\sigma}}$, so that
    $\typing{\Gamma}{\sigma}$ is $\bangBKRVTypeSys$-testable.

\item[$(\Leftarrow)$] Let $\bangBKRVTr \Gamma \vdash t : \sigma$ with
    that $\bangBKRVInhPred{\Gamma}$ and
    $\bangBKRVInhPred{\bangTypeArgs{\sigma}}$. By
    \Cref{lem:Bang_building_testing_ctxt_from_type_ctxt}, one deduces
    that $\bangBKRVTr \emptyset \vdash \bangTCtxt<t> : \sigma$, and by
    \Cref{lem:Bang_building_testing_ctxt_from_arg_types}, $\bangBKRVTr
    \emptyset \vdash \bangTCtxt'<\bangTCtxt<t>> : \tau$ for some type
    $\tau$ not functional. By
    \Cref{lem:Bang_BKRV}.\ref{lem:Bang_BKRV_Characterization}, there
    exists a $\bangStratCtxt$-normal form $s \in \bangSetTerms$ such
    that $\bangTCtxt'<\bangTCtxt<t>> \bangArr*_S s$. Moreover,
    by~\Cref{lem:Bang_BKRV}.\ref{lem:Bang_BKRV_Surface_Typing_SRE},
    one has that $\bangBKRVTr \emptyset \vdash s : \tau$ thus, using
    \Cref{lem:Bang_O_|-_t_:_M_and_SNF(t)_==>_t_=_!u}, one concludes
    that $s = \oc u$ for some $u \in \bangSetTerms$. Since testing
    contexts are stable by composition, we conclude that $t$ is
    \bangCalculusSymb-meaningful.
    \qedhere
\end{description}
\end{proof}
} \deliaLu \giulioLu

\begin{definition}
	For every $t \in \bangSetTerms$, let $\llbracket t \rrbracket
	\coloneqq \{ (\Gamma;\sigma) \ \bangBKRVTypeSys\text{-testable}
	\mid \bangBKRVTr\; \Gamma \vdash t : \sigma\}$. The
	\emph{equational theory} induced by system $\bangBKRVTypeSys$ is
	the relation $\equiv_\bangBKRVTypeSys$ on $\bangSetTerms$ defined
	by: $t \equiv_\bangBKRVTypeSys u$ iff $\llbracket t \rrbracket =
	\llbracket u \rrbracket$.
\end{definition}

\begin{lemma}
	\label{lem:consistency-B}
	The equational theory $\equiv_\bangBKRVTypeSys$ induced by
	$\bangBKRVTypeSys$ is a consistent \bangTheory-theory.
\end{lemma}
\stableProof{
\begin{proof}
	By definition, $\equiv_\bangBKRVTypeSys$ is an equivalence on
	$\bangSetTerms$. It contains $\bangArr_F$ (that is, if $t
	\bangArr_F u$ then $t \equiv_\bangBKRVTypeSys u$) by
	\Cref{lem:Bang_BKRV}.(\ref{lem:Bang_BKRV_Surface_Typing_SRE}).
	From the definition of system $\bangBKRVTypeSys$, it follows
	immediately that $t \equiv_\bangBKRVTypeSys u$ implies
	$\bangFCtxt<t> \equiv_\bangBKRVTypeSys \bangFCtxt<u>$. We have
	seen in \Cref{sec:Meaningfulness} that $\llbracket xx \rrbracket =
	\emptyset$ and $\llbracket \abs{x}{x \oc x} \rrbracket \neq
	\emptyset$, hence $xx \not \equiv_\bangBKRVTypeSys \abs{x}{x \oc
	x}$ and so $\equiv_\bangBKRVTypeSys$ is consistent. 
\end{proof}
} \deliaLu \giulioLu%

\RecConsistencyThH*%
\label{prf:Bang_Theory_H_Consistent}%
\stableProof{
\begin{proof}
	According to \Cref{t:logical-characterization}, for every $t \in
	\bangSetTerms$, $\llbracket t \rrbracket = \emptyset$ iff $t$ is
	\bangCalculusSymb-meaningless. Therefore,
	$\equiv_\bangBKRVTypeSys$ equates all
	\bangCalculusSymb-meaningless terms. Since $\bangTheoryH$ is the
	smallest \bangTheory-theory equating all
	\bangCalculusSymb-meaningless terms, and since
	$\equiv_\bangBKRVTypeSys$ is a \bangTheory-theory
	(\Cref{lem:consistency-B}), then $\equiv_{\bangTheoryH} \,
	\subseteq \, \equiv_\bangBKRVTypeSys$. Hence, since
	$\equiv_\bangBKRVTypeSys$ is consistent
	(\Cref{lem:consistency-B}), then so is $\equiv_{\bangTheoryH}$.
\end{proof}
} \deliaLu \giulioLu%

\section{Proofs of \Cref{sec:typed-genericity}}

\RecBangBKRVEquivAGKMean*
\label{prf:Bang_BKRV_and_Inh_Equiv_AGKMean}%
\stableProof{%
    \begin{proof}
  \begin{description}
	\item[$(\Leftarrow)$] Immediate consequence of the definition of 
	$\Pi \bangAGKMeanTr \Gamma \vdash t : \sigma$.
	
	\item[$(\Rightarrow)$] Let $\Pi \bangBKRVTr \Gamma \vdash t :
    \sigma$ with $\bangBKRVInhPred{\Gamma}$ and
    $\bangBKRVInhPred{\bangTypeArgs{\sigma}}$. We proceed by induction
    on $\Pi$, and we consider its last rule. Cases:
    \begin{description}
    \item[\bangBKRVVarRuleName] Then $t = x$ and $\Pi$ has the
        following form:
        \begin{equation*}
            \begin{prooftree}
                \inferBangBKRVVar{x : \mset{\sigma} \vdash x : \sigma}
            \end{prooftree}
        \end{equation*}
        with $\Gamma = x : \mset{\sigma}$. By hypothesis
        $\bangBKRVInhPred{x : \mset{\sigma}}$ and
        $\bangBKRVInhPred{\bangTypeArgs{\sigma}}$, hence  $x :
        \mset{\sigma} \vdash x : \sigma$, which is the only judgment
        of $\Pi$, is $\bangBKRVTypeSys$-testable. Therefore, $\Pi
        \bangAGKMeanTr \Gamma \vdash t : \sigma$.

    \item[\bangBKRVAbsRuleName] Then $t = \abs{x}{t'}$ and $\Pi$ has
        the following form, with $\sigma = \M \typeArrow \tau$:
        \begin{equation*}
            \begin{prooftree}
                \hypo{\Phi \bangBKRVTr \Gamma, x : \M \vdash t' : \tau}
                \inferBangBKRVAbs{\Gamma \vdash \abs{x}{t'} : \M \typeArrow \tau}
            \end{prooftree}
        \end{equation*}
        By hypothesis $\bangBKRVInhPred{\Gamma}$ and
        $\bangBKRVInhPred{\bangTypeArgs{\M \typeArrow \tau}}$, thus
        $\bangBKRVInhPred{\M}$ and
        $\bangBKRVInhPred{\bangTypeArgs{\tau}}$, hence
        $\bangBKRVInhPred{\Gamma, x : \M}$. By \ih $\Phi
        \bangAGKMeanTr \Gamma, x : \M \vdash t' : \tau$. So, all
        judgments in $\Pi$ are $\bangBKRVTypeSys$-testable, hence $\Pi
        \bangAGKMeanTr \Gamma \vdash t : \sigma$.

    \item[\bangBKRVAppRuleName] Then $t = \app{t_1}{t_2}$ and $\Pi$
        has the following form, with $\Gamma = \Gamma_1 + \Gamma_2$:
        \begin{equation*}
            \begin{prooftree}
                \hypo{\Phi_1 \bangBKRVTr \Gamma_1 \vdash t_1 : \M \typeArrow \sigma}
                \hypo{\Phi_2 \bangBKRVTr \Gamma_2 \vdash t_2 : \M}
                \inferBangBKRVApp{\Gamma_1 + \Gamma_2 \vdash \app{t_1}{t_2} : \sigma}
            \end{prooftree}
        \end{equation*}
        By hypothesis $\bangBKRVInhPred{\Gamma_1 + \Gamma_2}$ and
        $\bangBKRVInhPred{\bangTypeArgs{\sigma}}$ thus
        $\bangBKRVInhPred{\Gamma_1}$ and $\bangBKRVInhPred{\Gamma_2}$.
        By \Cref{lem:Bang_building_testing_ctxt_from_type_ctxt}
        applied to $\Phi_2$, $\bangBKRVInhPred{\M}$ and hence
        $\bangBKRVInhPred{\bangTypeArgs{\M \typeArrow \sigma}}$. By
        \ih on both $\Phi_1$ and $\Phi_2$, $\Phi_1 \bangAGKMeanTr
        \Gamma_1 \vdash t_1 : \M \typeArrow \sigma$ and $\Phi_2
        \bangAGKMeanTr \Gamma_2 \vdash t_2 : \M$. So, all judgments in
        $\Pi$ are $\bangBKRVTypeSys$-testable, hence $\Pi
        \bangAGKMeanTr \Gamma \vdash t : \sigma$.
        
    \item[\bangBKRVEsRuleName] Then $t = t_1\esub{x}{t_2}$ and $\Pi$
        has the following form, with $\Gamma = \Gamma_1 + \Gamma_2$:
        \begin{equation*}
            \begin{prooftree}
                \hypo{\Phi_1 \bangBKRVTr \Gamma_1; x : \M \vdash t_1 : \sigma}
                \hypo{\Phi_2 \bangBKRVTr \Gamma_2 \vdash t_2 : \M}
                \inferBangBKRVEs{\Gamma_1 + \Gamma_2 \vdash t_1\esub{x}{t_2} : \sigma}
            \end{prooftree}
        \end{equation*}
        By hypothesis $\bangBKRVInhPred{\Gamma_1 + \Gamma_2}$ and
        $\bangBKRVInhPred{\bangTypeArgs{\sigma}}$ thus
        $\bangBKRVInhPred{\Gamma_1}$ and $\bangBKRVInhPred{\Gamma_2}$.
        By \Cref{lem:Bang_building_testing_ctxt_from_type_ctxt}
        applied to $\Phi_2$, $\bangBKRVInhPred{\M}$ and hence
        $\bangBKRVInhPred{\Gamma_1, x : \M}$. By \ih on both $\Phi_1$
        and $\Phi_2$, $\Phi_1 \bangAGKMeanTr \Gamma_1, x : \M \vdash
        t_1 : \sigma$ and $\Phi_2 \bangAGKMeanTr \Gamma_2 \vdash t_2 :
        \M$. So, all judgments in $\Pi$ are
        $\bangBKRVTypeSys$-testable, hence $\Pi \bangAGKMeanTr \Gamma
        \vdash t : \sigma$.

    \item[\bangBKRVDerRuleName] Then $t = \der{t'}$ and $\Pi$ has the
        following form:
        \begin{equation*}
            \begin{prooftree}
                \hypo{\Phi \bangBKRVTr \Gamma \vdash t' : \mset{\sigma}}
                \inferBangBKRVDer{\Gamma \vdash \der{t'} : \sigma}
            \end{prooftree}
        \end{equation*}
        By hypothesis $\bangBKRVInhPred{\Gamma}$ and
        $\bangBKRVInhPred{\bangTypeArgs{\sigma}}$. By
        \Cref{lem:Bang_building_testing_ctxt_from_type_ctxt} applied
        to $\Phi$,     $\bangBKRVInhPred{\mset{\sigma}}$. By \ih on
        $\Phi$, $\Phi \bangAGKMeanTr \Gamma \vdash t' :
        \mset{\sigma}$. So, all judgments in $\Pi$ are
        $\bangBKRVTypeSys$-testable, hence $\Pi \bangAGKMeanTr \Gamma
        \vdash t : \sigma$.

    \item[\bangBKRVBgRuleName] Then $t = \oc t'$ and $\Pi$ has the
        following form,  with $\Gamma = +_{i \in I} \Gamma_i$ and $\sigma =
        \mset{\tau_i}_{i \in I}$:
        \begin{equation*}
            \begin{prooftree}
                \hypo{(\Phi_i \bangBKRVTr \Gamma_i \vdash t' : \tau_i)}
                \inferBangBKRVBg{+_{i \in I} \Gamma_i \vdash \oc t' : \mset{\tau_i}_{i \in I}}
            \end{prooftree}
        \end{equation*}
       	By hypothesis $\bangBKRVInhPred{+_{i \in I} \Gamma_i}$ and
        $\bangBKRVInhPred{\bangTypeArgs{\mset{\tau_i}_{i \in I}}}$.
        Thus for any $i \in I$, $\bangBKRVInhPred{\Gamma_i}$ According
        to \Cref{lem:Bang_building_testing_ctxt_from_type_ctxt}
        applied to $\Phi_i$, $\bangBKRVInhPred{\tau_i}$. By \ih on
        each $\Phi_i$, $\Phi_i \bangAGKMeanTr \Gamma_i \vdash t' :
        \tau_i$. So, all judgments in $\Pi$ are
        $\bangBKRVTypeSys$-testable, hence $\Pi \bangAGKMeanTr \Gamma
        \vdash t : \sigma$.
        \qedhere
    \end{description}
\end{description}
\end{proof}
} \deliaLu \giulioLu

\RecBangMeaningfulTypeMonotonicity*
\label{prf:Bang_AGKMean_Typed_Genericity}%
\stableProof{
    \begin{proof} 
Let $u \in \bangSetTerms$. Since every judgment in $\Pi \bangAGKMeanTr
\bangFCtxt<t>$ is inhabited and we shall build $\Pi' \bangBKRVTr
\bangFCtxt<u>$ from $\Pi$ by possibly changing the subject in some
judgments without altering their types, $\Pi' \bangAGKMeanTr
\bangFCtxt<u>$ by construction. We proceed to build $\Pi'$ by
induction on $\bangFCtxt$. Cases:
\begin{itemize}
\item[\bltI] $\bangFCtxt = \Hole$: Then $\bangFCtxt<t> = t$ is
    meaningless, which is impossible since it contradicts
    \Cref{lem:Bang_Meaningfulness_Typing_and_Inhabitation,lem:Bang_BKRV_and_Inh_Equiv_AGKMean}.

\item[\bltI] $\bangFCtxt = \abs{x}{\bangFCtxt'}$: Then $\Pi$ is
    necessarily of the following form, with $\sigma = \M \typeArrow
    \tau$:
    \begin{equation*}
        \begin{prooftree}
            \hypo{\Phi \bangAGKMeanTr \Gamma, x : \M \vdash \bangFCtxt'<t> : \tau}
            \inferBangAGKMeanAbs{\Gamma \vdash \abs{x}{\bangFCtxt'<t>} : \M \typeArrow \tau}
        \end{prooftree}
    \end{equation*}
    By \ih on $\bangFCtxt'$, there is $\Phi' \bangAGKMeanTr \Gamma, x
    : \M \vdash \bangFCtxt'<u> : \tau$. Let $\Pi'$ be as follows:
    \begin{equation*}
        \begin{prooftree}
            \hypo{\Phi' \bangAGKMeanTr \Gamma, x : \M \vdash \bangFCtxt'<u> : \tau}
            \inferBangAGKMeanAbs{\Gamma \vdash \abs{x}{\bangFCtxt'<u>} : \M \typeArrow \tau}
        \end{prooftree}
    \end{equation*}

\item[\bltI] $\bangFCtxt = \app{\bangFCtxt'}{s}$: Then $\Pi$ is
    necessarily as follows, with $\Gamma = \Gamma_1 + \Gamma_2$ and :
    \begin{equation*}
        \begin{prooftree}
            \hypo{\Pi_1 \bangAGKMeanTr \Gamma_1 \vdash \bangFCtxt'<t> : \M \typeArrow \sigma}
            \hypo{\Pi_2 \bangAGKMeanTr \Gamma_2 \vdash s : \M}
            \inferBangBKRVApp{\Gamma_1 + \Gamma_2 \vdash \app{\bangFCtxt'<t>}{s} : \sigma}
        \end{prooftree}
    \end{equation*}
   	By \ih on $\bangFCtxt'$, there is $\Pi'_1 \bangAGKMeanTr \Gamma_1
    \vdash \bangFCtxt'<u> : \M \typeArrow \sigma$. We set $\Pi'$ as
    follows:
    \begin{equation*}
        \begin{prooftree}
            \hypo{\Pi_1 \bangAGKMeanTr \Gamma_1 \vdash \bangFCtxt'<u> : \M \typeArrow \sigma}
            \hypo{\Pi_2 \bangAGKMeanTr \Gamma_2 \vdash s : \M}
            \inferBangBKRVApp{\Gamma_1 + \Gamma_2 \vdash \app{\bangFCtxt'<u>}{s} : \sigma}
        \end{prooftree}
    \end{equation*}

\item[\bltI] $\bangFCtxt = \app[\,]{s}{\bangFCtxt'}$: Then $\Pi$ is
    necessarily of the following form:
    \begin{equation*}
        \begin{prooftree}
            \hypo{\Pi_1 \bangAGKMeanTr \Gamma_1 \vdash s : \M \typeArrow \sigma}
            \hypo{\Pi_2 \bangAGKMeanTr \Gamma_2 \vdash \bangFCtxt'<t> : \M}
            \inferBangBKRVApp{\Gamma_1 + \Gamma_2 \vdash \app[\,]{s}{\bangFCtxt'<t>} : \sigma}
        \end{prooftree}
    \end{equation*}
    with $\Gamma = \Gamma_1 + \Gamma_2$. By \ih on $\bangFCtxt'$,
    there is $\Pi'_2 \bangAGKMeanTr \Gamma_2 \vdash \bangFCtxt'<u> :
    \M$. We set $\Pi'$ as follows:
    \begin{equation*}
        \begin{prooftree}
            \hypo{\Pi_1 \bangAGKMeanTr \Gamma_1 \vdash s : \M \typeArrow \sigma}
            \hypo{\Pi'_2 \bangAGKMeanTr \Gamma_2 \vdash \bangFCtxt'<u> : \M}
            \inferBangBKRVApp{\Gamma_1 + \Gamma_2 \vdash \app[\,]{s}{\bangFCtxt'<u>} : \sigma}
        \end{prooftree}
    \end{equation*}

\item[\bltI] $\bangFCtxt = \bangFCtxt'\esub{x}{s}$: Then $\Pi$ is
    necessarily of the following form:
    \begin{equation*}
        \begin{prooftree}
            \hypo{\Pi_1 \bangAGKMeanTr \Gamma_1, x : \M \vdash \bangFCtxt'<t> : \sigma}
            \hypo{\Pi_2 \bangAGKMeanTr \Gamma_2 \vdash s : \M}
            \inferBangBKRVEs{\Gamma_1 + \Gamma_2 \vdash \bangFCtxt'<t>\esub{x}{s} : \sigma}
        \end{prooftree}
    \end{equation*}
    with $\Gamma = \Gamma_1 + \Gamma_2$. By \ih on $\bangFCtxt'$,
    there is $\Pi'_1 \bangAGKMeanTr \Gamma_1, x : \M \vdash
    \bangFCtxt'<u> : \sigma$. We set $\Pi'$ as follows:
    \begin{equation*}
        \begin{prooftree}
            \hypo{\Pi'_1 \bangAGKMeanTr \Gamma_1, x : \M \vdash \bangFCtxt'<u> : \sigma}
            \hypo{\Pi_2 \bangAGKMeanTr \Gamma_2 \vdash s : \M}
            \inferBangBKRVEs{\Gamma_1 + \Gamma_2 \vdash \bangFCtxt'<u>\esub{x}{s} : \sigma}
        \end{prooftree}
    \end{equation*}

\item[\bltI] $\bangFCtxt = s\esub{x}{\bangFCtxt'}$: Then $\Pi$ is
    necessarily of the following form:
    \begin{equation*}
        \begin{prooftree}
            \hypo{\Pi_1 \bangAGKMeanTr \Gamma_1, x : \M \vdash s : \sigma}
            \hypo{\Pi_2 \bangAGKMeanTr \Gamma_2 \vdash \bangFCtxt'<t> : \M}
            \inferBangBKRVEs{\Gamma_1 + \Gamma_2 \vdash s\esub{x}{\bangFCtxt'<t>} : \sigma}
        \end{prooftree}
    \end{equation*}
    with $\Gamma = \Gamma_1 + \Gamma_2$. By \ih on $\bangFCtxt'$,
    there is $\Pi_2' \bangAGKMeanTr \Gamma_2 \vdash \bangFCtxt'<u> :
    \M$. We set $\Pi'$ as follows:
    \begin{equation*}
        \begin{prooftree}
            \hypo{\Pi_1 \bangAGKMeanTr \Gamma_1, x : \M \vdash s : \sigma}
            \hypo{\Pi_2' \bangAGKMeanTr \Gamma_2 \vdash \bangFCtxt'<u> : \M}
            \inferBangBKRVEs{\Gamma_1 + \Gamma_2 \vdash s\esub{x}{\bangFCtxt'<u>} : \sigma}
        \end{prooftree}
    \end{equation*}

\item[\bltI] $\bangFCtxt = \oc\bangFCtxt'$: Then $\Pi$ is necessarily
    of the following form, for some finite index set $I$:
    \begin{equation*}
        \begin{prooftree}
            \hypo{\Pi_i \bangAGKMeanTr \Gamma_i \vdash \bangFCtxt'<t> : \sigma_i}
            \delims{\left(}{\right)_{i \in I}}
            \inferBangBKRVBg{+_{i \in I} \Gamma_i &\vdash \oc \bangFCtxt'<t> : \mset{\sigma_i}_{i \in I}}
        \end{prooftree}
    \end{equation*}
	with $\sigma = \mset{\sigma_i}_{i \in I}$ and $\Gamma = +_{i \in
    I} \Gamma_i$. For each $i \in I$, by \ih on $\bangFCtxt'$, one
    obtains $\Pi'_i \bangAGKMeanTr \Gamma_i \vdash \bangFCtxt'<u> :
    \sigma_i$. We set $\Pi'$ as follows:
    \begin{equation*}
        \begin{prooftree}
            \hypo{\Pi'_i \bangAGKMeanTr \Gamma_i \vdash \bangFCtxt'<u> : \sigma_i}
            \delims{\left(}{\right)_{i \in I}}
            \inferBangBKRVBg{+_{i \in I} \Gamma_i &\vdash \oc \bangFCtxt'<u> : \mset{\sigma_i}_{i \in I}}
        \end{prooftree}
    \end{equation*}

\item[\bltI] $\bangFCtxt = \der{\bangFCtxt'}$: Then $\Pi$ is
    necessarily of the following form:
    \begin{equation*}
        \begin{prooftree}
            \hypo{\Phi \bangAGKMeanTr \Gamma \vdash \bangFCtxt'<t> : \mset{\sigma}}
            \inferBangBKRVDer{\Gamma \vdash \der{\bangFCtxt'<t>} : \sigma}
        \end{prooftree}
    \end{equation*}
    By \ih on $\bangFCtxt'$, there is $\Phi' \bangAGKMeanTr \Gamma
    \vdash \bangFCtxt'<u> : \mset{\sigma}$. We set $\Pi'$ as follows:
    \begin{equation*}
        \begin{prooftree}
            \hypo{\Phi \bangAGKMeanTr \Gamma \vdash \bangFCtxt'<u> : \mset{\sigma}}
            \inferBangBKRVDer{\Gamma \vdash \der{\bangFCtxt'<u>} : \sigma}
        \end{prooftree}
    \end{equation*}
	\qedhere
\end{itemize}
\end{proof}
} \deliaLu \giulioLu




\begin{lemma}
	\label{lem:Bang_Extension_of_Theory_H}%
	Let $t, u \in \bangSetTerms$. If the smallest \bangTheory-theory
	containing $\bangTheoryH$ and equating $t$ and $u$ is consistent,
	then $t \equiv_{\bangTheoryH*} u$.
\end{lemma}
\stableProof{
\begin{proof}
	Let $t, u \in \bangSetTerms$ and $\mathcal{E}$ be the smallest
	\bangTheory-theory containing $\bangTheoryH$ and equating $t$ and
	$u$. Suppose by contrapositive that $t \not\equiv_{\bangTheoryH*}
	u$, that is, there is a context $\bangFCtxt$ such that
	$\bangFCtxt<t>$ is \bangCalculusSymb-meaningful and
	$\bangFCtxt<u>$ is \bangCalculusSymb-meaningless. We show that
	$\mathcal{E}$ is inconsistent.

	Let $s \in \bangSetTerms$ be a \bangCalculusSymb-meaningful term
	(as for example $\oc z$) and $x \in \bangSetVariables$ fresh. On
	one hand, by definition of meaningfulness, there exists a testing
	context $\bangTCtxt$ such that $\bangTCtxt<\bangFCtxt<t>>
	\bangArr*_S \oc u$ for some term $u \in \bangSetTerms$. In
	particular, $\app[\,]{(\abs{x}{s})}{\bangTCtxt<\bangFCtxt<t>>}
	\bangArr*_S \app[\,]{(\abs{x}{s})}{\oc u} \bangArr*_S s$. Since
	$\mathcal{E}$ is a \bangTheory-theory, then
	$\app[\,]{(\abs{x}{s})}{\bangTCtxt<\bangFCtxt<t>>}
	\equiv_\mathcal{E} s$. Since $\bangFCtxt<u>$ is
	\bangCalculusSymb-meaningless, then so are
	$\bangTCtxt<\bangFCtxt<u>>$. By
	typing, we deduce that
	$\app[\,]{(\abs{x}{s})}{\bangTCtxt<\bangFCtxt<u>>}$ is also
	\bangCalculusSymb-meaningless. Since $t \equiv_{\mathcal{E}} u$
	then by contextuality and transitivity
	$\app[\,]{(\abs{x}{s})}{\bangTCtxt<\bangFCtxt<u>>}
	\equiv_{\mathcal{E}} s$. Since $\mathcal{E}$ equates all
	meaningless terms, one concludes by transitivity that $\Omega
	\equiv_{\mathcal{E}} s$. Since $s$ is an arbitrary meaningful
	term, $\mathcal{E}$ is inconsistent.
\end{proof}
} \deliaLu \giulioLu%

\begin{corollary}
		\label{cor:inclusion-in-H*}
	Let $\mathcal{E}$ be a consistent \bangTheory-theory. If
	$\bangTheoryH \subseteq \mathcal{E}$ then $\mathcal{E} \subseteq
	\bangTheoryH*$. 
\end{corollary}
\stableProof{
\begin{proof}
	Let $t, u \in \bangSetTerms$ such that $t \equiv_{\mathcal{E}} u$.
	Let us show that $t \equiv_{\bangTheoryH*} u$ (hence $\mathcal{E}
	\subseteq \bangTheoryH*$). Let $\bangTheoryH^{t,u}$ be the
	smallest \bangTheory-theory containing $\bangTheoryH$ and equating
	$t$ and $u$. Note that $\bangTheoryH^{t,u} \subseteq \mathcal{E}$.
	Hence, since by hypothesis $\mathcal{E}$ does not equate all
	terms, $\bangTheoryH^{t,u}$ is consistent too. Therefore, $t
	\equiv_{\bangTheoryH*} u$ by
	\Cref{lem:Bang_Extension_of_Theory_H}.
\end{proof}
} \deliaLu \giulioLu%

\RecMaximalityConsistency*
\label{prf:Bang_Theory_H*_Consistent}%
\stableProof{
\begin{proof}~
	\begin{description}
	\item[(\bangTheory-theory):] ~
		\begin{itemize}
		\item From its definition, it follows immediately that
			$\bangTheoryH*$ is an equivalence on $\bangSetTerms$.
		
		\item To prove that $\bangTheoryH*$ is closed under full
			contexts, suppose that $t \equiv_{\bangTheoryH*} u$: we
			have to prove that, for every full context $\bangFCtxt$,
			$\bangFCtxt<t> \equiv_{\bangTheoryH*} \bangFCtxt<u>$; that
			is, for every full contexts $\bangFCtxt, \bangFCtxt'$,
			$\bangFCtxt'<\bangFCtxt<t>>$ is
			\bangCalculusSymb-meaningful iff so is
			$\bangFCtxt'<\bangFCtxt<u>>$. Since
			$\bangFCtxt'<\bangFCtxt>$ is a full context and $t
			\equiv_{\bangTheoryH*} u$, we are done by definition.

		\item To prove that $\bangTheoryH*$ contains $\bangArr_F$,
			suppose that $t \bangArr_F u$, so that $\bangFCtxt<t>
			\bangArr_F \bangFCtxt<u>$ holds for every context
			$\bangFCtxt$. By
			\Cref{lem:Bang_BKRV}.(\ref{lem:Bang_BKRV_Surface_Typing_SRE}),
			for every typing (and in particular every testable typing)
			$(\Gamma; \sigma)$, $\bangBKRVTr \Gamma \vdash
			\bangFCtxt<t> : \sigma$ iff $\bangBKRVTr \Gamma \vdash
			\bangFCtxt<u> : \sigma$. Hence, by
			\Cref{t:logical-characterization}, $\bangFCtxt<t>$ is
			\bangCalculusSymb-meaningful iff so is $\bangFCtxt<u>$;
			hence, $t \equiv_{\bangTheoryH*} u$ by definition.
		\end{itemize}

	\item [(Containing $\bangTheoryH$):] As $\bangTheoryH*$ is also a
		\bangTheory-theory, it suffices to prove that $\bangTheoryH*$
		equates all \bangCalculusSymb-meaningless terms. Let $t, u \in
		\bangSetTerms$ be \bangCalculusSymb-meaningless. By surface
		genericity (\Cref{lem:Bang_Qualitative_Surface_Genericity}),
		for every full context $\bangFCtxt$, $\bangFCtxt<t>$ is
		\bangCalculusSymb-meaningful iff so is $\bangFCtxt<u>$;
		therefore, $t \equiv_{\bangTheoryH*} u$ by definition.

	\item [(Consistency):] The term $\oc x$ is trivially
		\bangCalculusSymb-meaningful, while we have seen in
		\Cref{sec:Meaningfulness} that $\Omega$ is
		\bangCalculusSymb-meaningless. Take $\bangFCtxt = \Hole$, then
		$\bangFCtxt<\oc x> = \oc x$ is \bangCalculusSymb-meaningful
		and $\bangFCtxt<\Omega> = \Omega$ is
		\bangCalculusSymb-meaningless. Therefore, $\oc x
		\not\equiv_{\bangTheoryH*} \Omega$.

	\item [(Maximality):] Let $\mathcal{E}$ be a consistent
		\bangTheory-theory containing $\bangTheoryH$. By
		\Cref{cor:inclusion-in-H*}, $\mathcal{E} \subseteq
		\bangTheoryH*$. Therefore, $\bangTheoryH*$ is maximal among
		the consistent \bangTheory-theories containing $\bangTheoryH$.

	\item [(Uniqueness):] Let $\mathcal{E}$ be a \emph{maximal}
		consistent \bangTheory-theory containing $\bangTheoryH$. By
		\Cref{cor:inclusion-in-H*}, $\mathcal{E} \subseteq
		\bangTheoryH*$. We conclude that $\mathcal{E} = \bangTheoryH*$
		by maximality of $\mathcal{E}$. Therefore, $\bangTheoryH*$ is
		unique among the maximal consistent \bangTheory-theories
		containing $\bangTheoryH$.
		\qedhere
	\end{description}
\end{proof}
} \deliaLu \giulioLu%







\section{Proofs of Section 5}

\subsection{Proofs of Subsection \Cref{subsec:Cbn_Meaningfulness}}

\begin{lemma}
    \label{lem:Cbn_building_testing_ctxt_from_type_ctxt}%
    Let $(\Pi_i \cbnBKRVTr \Gamma_i \vdash t : \sigma_i)_{i \in I}$ with $I$ finite.
    If $\cbnBKRVInhPred{+_{i \in I} \Gamma_i}$ then
    $\cbnBKRVInhPred{\mset{\sigma_i}_{i \in I}}$.
\end{lemma}
\stableProof{%
\begin{proof}
Let $(\Pi_i \cbnBKRVTr \Gamma_i \vdash t : \sigma_i)_{i \in I}$ such
that $\cbnBKRVInhPred{+_{i \in I} \Gamma_i}$. We reason by induction
on the number of variables in $\typeCtxtDom{+_{i \in I} \Gamma_i}$:
\begin{itemize}
\item[\bltI] If $\typeCtxtDom{+_{i \in I} \Gamma_i}$ is empty, then $\Gamma_i = \emptyset$ for all $i \in I$. Thus $\Pi_i \cbnBKRVTr \emptyset \vdash t : \sigma_i$ for all $i \in I$ by hypothesis, hence $\cbnBKRVTr \emptyset \vdash \oc t : [\sigma_i]_{i \in I}$ by applying the rule \bangBKRVBgRuleName. Therefore, $\cbnBKRVInhPred{\mset{\sigma_i}_{i \in I}}$.

\item[\bltI] If $\typeCtxtDom{+_{i \in I} \Gamma_i}$ is not empty,
    then every $\Gamma_i = \Gamma'_i, x : \M_i$ (with at least one
    $\M_j \neq  \emptymset$): Since $\cbnBKRVInhPred{+_{i \in I}
    \Gamma_{i \in I}}$, then $\cbnBKRVInhPred{+_{i \in I} \Gamma'_{i
    \in I}}$ and $\cbnBKRVInhPred{\biguplus_{i \in I} \M_i}$. Let $\M_i =
    \mset{\rho_i^j}_{j \in J_i}$, then by definition, there exists
    $(\Pi_i^j \cbnBKRVTr \emptyset \vdash u : \rho_i^j)_{i \in I, j
    \in J_i}$. Let $\Phi_i$ be the following derivation:
    \begin{equation*}
        \begin{prooftree}
            \hypo{\Pi_i \cbnBKRVTr \Gamma'_i, x : \mset{\rho_i^j}_{j \in J_i} \vdash t : \sigma_i}
            \inferBangBKRVAbs{\Gamma'_i \vdash \abs{x}{t} : \mset{\rho_i^j}_{j \in J_i} \typeArrow \sigma_i}
            \hypo{(\Pi_i^j \cbnBKRVTr \emptyset \vdash u : \rho_i^j)_{j \in J_i}}
            \inferBangBKRVEs{\Gamma'_i \vdash \app[\,]{(\abs{x}{t})}{u} : \sigma_i}
        \end{prooftree}
    \end{equation*}
    We thus obtain $(\Phi_i \cbnBKRVTr \Gamma'_i \vdash
    \app[\,]{(\abs{x}{t})}{u} : \sigma_i)_{i \in I}$ and by \ih on the
    number of variables in $\typeCtxtDom{+_{i \in I} \Gamma'_i}$, one concludes that
    $\cbnBKRVInhPred{\mset{\sigma_i}_{i \in I}}$.
    \qedhere
\end{itemize}
\end{proof}

} \deliaLu \giulioLu

\begin{lemma}
    \label{lem:Cbn_O_|-_T<t>_:_[]_==>_Gam_|-_t_:_sig_and_Gam_and_args(sig)_inh}%
    Let $t \in \cbnSetTerms$ and $\cbnTCtxt$ be a testing context. If
    $\cbnBKRVTr\; \emptyset \vdash \cbnTCtxt<t> : \sigma$ with
    $\cbnBKRVInhPred{\cbnTypeArgs{\sigma}}$, then there exist $\Gamma$
    and $\sigma'$ such that $\cbnBKRVTr\; \Gamma \vdash t : \sigma'$
    with $\cbnBKRVInhPred{\Gamma}$ and
    $\cbnBKRVInhPred{\cbnTypeArgs{\sigma'}}$.
\end{lemma}
\stableProof{%
\begin{proof}
We prove the stronger statement below, to have the right induction
hypothesis:
\begin{center}
    Let $\cbnBKRVTr \Gamma \vdash \cbnTCtxt<t> : \sigma$ with
    $\cbnBKRVInhPred{\Gamma}$ and
    $\cbnBKRVInhPred{\cbnTypeArgs{\sigma}}$,
    then there exists $\Gamma'$ and $\sigma'$ such that $\cbnBKRVTr \Gamma' \vdash t
    : \sigma'$ with $\cbnBKRVInhPred{\Gamma'}$ and
    $\cbnBKRVInhPred{\cbnTypeArgs{\sigma'}}$.
\end{center}

From that,
\Cref{lem:Cbn_O_|-_T<t>_:_[]_==>_Gam_|-_t_:_sig_and_Gam_and_args(sig)_inh}
follows immediately, because $\cbnBKRVInhPred{\emptyset}$ is vacuously
true.

Let $\Pi \cbnBKRVTr \Gamma \vdash \cbnTCtxt<t> : \sigma$ with
$\cbnBKRVInhPred{\Gamma}$ and $\cbnBKRVInhPred{\cbnTypeArgs{\sigma}}$.
We reason by induction on $\cbnTCtxt$:
\begin{itemize}
\item[\bltI] $\cbnTCtxt = \Hole$: Trivial by hypothesis by taking
    $\Gamma' := \Gamma$ and $\sigma ' := \sigma$.

\item[\bltI] $\cbnTCtxt = \app[\,]{(\abs{x}{\cbnTCtxt'})}{u}$: Then
    $\Pi$ has the following form, where $\Gamma = \Gamma_1 +_{i \in I} \Gamma_2^i$ and $I$ is finite:
    \begin{equation*}
        \begin{prooftree}
            \hypo{\Pi_1 \cbnBKRVTr \Gamma_1, x : \mset{\tau_i}_{i \in I} \vdash \cbnTCtxt'<t> : \sigma}
            \inferCbnBKRVAbs{\Gamma_1 \vdash \abs{x}{\cbnTCtxt'<t>} : \mset{\tau_i}_{i \in i} \typeArrow \sigma}
            \hypo{(\Pi_2^i \cbnBKRVTr \Gamma_2^i \vdash u : \tau_i)_{i \in I}}
            \inferCbnBKRVApp{\Gamma_1 +_{i \in I} \Gamma_2^i \vdash \app[\,]{(\abs{x}{\cbnTCtxt'<t>})}{u} : \sigma}
        \end{prooftree}
    \end{equation*}
   	By hypothesis
    $\cbnBKRVInhPred{\Gamma}$, thus in particular
    $\cbnBKRVInhPred{\Gamma_1}$ and $\cbnBKRVInhPred{+_{\iI}
    \Gamma_2^i}$. Using
    \Cref{lem:Cbn_building_testing_ctxt_from_type_ctxt} on
    $(\Pi_2^i)_{i \in I}$, one deduces that
    $\cbnBKRVInhPred{\mset{\tau_i}_{i \in I}}$ thus
    $\cbnBKRVInhPred{\Gamma_1, x : \mset{\tau_i}_{i \in I}}$. By \ih
    on $\Pi_1$, one obtains $\Pi' \cbnBKRVTr \Gamma' \vdash t :
    \sigma'$ with $\cbnBKRVInhPred{\Gamma'}$ and
    $\cbnBKRVInhPred{\typeArgs{\sigma'}}$.

\item[\bltI] $\cbnTCtxt = \app{\cbnTCtxt'}{u}$: Then $\Pi$ has the
    following form, with $\Gamma = \Gamma_1 +_{i \in I} \Gamma_2^i$:
    \begin{equation*}
        \begin{prooftree}
            \hypo{\Pi_1 \cbnBKRVTr \Gamma_1 \vdash \cbnTCtxt'<t> : \mset{\tau_i}_{i \in I} \typeArrow \sigma}
            \hypo{(\Pi_2^i \cbnBKRVTr \Gamma_2^i \vdash u : \tau_i)_{i \in I}}
            \inferBangBKRVApp{\Gamma_1 +_{i \in I} \Gamma_2^i \vdash \app{\cbnTCtxt'<t>}{u} : \sigma}
        \end{prooftree}
    \end{equation*}
    By hypothesis
    $\cbnBKRVInhPred{\Gamma}$, thus in particular
    $\cbnBKRVInhPred{\Gamma_1}$ and $\cbnBKRVInhPred{+_{i \in I}
    \Gamma_2^i}$. By
    \Cref{lem:Cbn_building_testing_ctxt_from_type_ctxt} on
    $(\Pi_2^i)_{i \in I}$, 
    $\cbnBKRVInhPred{\mset{\tau_i}_{i \in I}}$, and since
    $\cbnBKRVInhPred{\cbnTypeArgs{\sigma}}$ thus
    $\cbnBKRVInhPred{\cbnTypeArgs{\mset{\tau_i}_{i \in I} \typeArrow
    \sigma}}$. By the \ih on $\Pi_1$, one obtains $\Pi' \cbnBKRVTr
    \Gamma' \vdash t : \sigma'$ with $\cbnBKRVInhPred{\Gamma'}$ and
    $\cbnBKRVInhPred{\typeArgs{\sigma'}}$.
    \qedhere
\end{itemize}
\end{proof}

} \deliaLu \giulioLu

\RecCbnMeaningfulnessCharacterizations*
\label{prf:cbnBKRV_characterizes_meaningfulness}
\stableProof{
\begin{proof}
\begin{enumerate}
\item
    See~\cite{barendregt84nh,BucciarelliKR15}.

\item ~ \begin{itemize}
    \item [$(1) \Leftrightarrow (2):$]
        See~\cite{deCarvalho07,KesnerV14,BucciarelliKR21,BucciarelliKesnerRiosViso20}.

    \item[$(3) \Rightarrow (2):$] Trivial since
        $\cbnBKRVTypeSys$-testable requires to be
        $\cbnBKRVTypeSys$-typable.

    \item[$(1) \Rightarrow (3):$] Let $t$ be
        \cbnCalculusSymb-meaningful, then there exists a testing
        context $\cbnTCtxt$ such that $\cbnTCtxt<t> \cbnArr*_S \Id$.
        Notice that $\Id$ is $\cbnBKRVTypeSys$-testable thanks to the
        following derivation:
        \begin{equation*}
            \begin{prooftree}
                \inferCbnBKRVVar{x : \mset{\mset{\alpha} \typeArrow \alpha} \vdash x : \mset{\alpha} \typeArrow \alpha}
                \inferCbnBKRVAbs{\emptyset \vdash \abs{x}{x} : \mset{\mset{\alpha} \typeArrow \alpha} \typeArrow \mset{\alpha} \typeArrow \alpha}
            \end{prooftree}
        \end{equation*}
        Indeed, $\cbnBKRVInhPred{\emptyset}$ and
        $\cbnBKRVInhPred{\genericArgs{\cbnBKRVTypeSys}{\mset{\mset{\alpha}
        \typeArrow \alpha} \typeArrow \mset{\alpha} \typeArrow
        \alpha}}$ since
        $\genericArgs{\cbnBKRVTypeSys}{\mset{\mset{\alpha} \typeArrow
        \alpha} \typeArrow \mset{\alpha} \typeArrow \alpha} =
        \emptyset$. By subject expansion (see
        \cite{BucciarelliKesnerRiosViso20,BucciarelliKesnerRiosViso23}),
        we deduce that $\cbnTCtxt<t>$ is also
        $\cbnBKRVTypeSys$-testable, thus $t$ is
        $\cbnBKRVTypeSys$-testable using
        \Cref{lem:Cbn_O_|-_T<t>_:_[]_==>_Gam_|-_t_:_sig_and_Gam_and_args(sig)_inh}.
        \qedhere
    \end{itemize}
\end{enumerate}
\end{proof}
} \deliaLu \giulioLu%

\subsection{Proofs of Subsection \Cref{subsec:Cbv_Meaningfulness}}

\begin{lemma}
    \label{lem:Cbv_building_testing_ctxt_from_type_ctxt}%
    Let $\Pi \cbvBKRVTr \Gamma \vdash t : \sigma$. If
    $\cbvBKRVInhPred{\Gamma}$ then $\cbvBKRVInhPred{\sigma}$.
\end{lemma}
\stableProof{%
    \begin{proof}
Let $\Pi \cbvBKRVTr \Gamma \vdash t : \sigma$. We reason by induction
on the number of variables in $\typeCtxtDom{\Gamma}$:
\begin{itemize}
\item[\bltI] $\Gamma = \emptyset$: Then $\cbvBKRVInhPred{\sigma}$
    holds by hypothesis.

\item[\bltI] $\Gamma = \Gamma', x : \M$ with $\M \neq \emptymset$:
    Since $\cbvBKRVInhPred{\Gamma}$, then $\cbvBKRVInhPred{\Gamma'}$
    and $\cbvBKRVInhPred{\M}$, thus there exists $\Pi_u \cbvBKRVTr
    \emptyset \vdash u : \M$ for some $u \in \cbvSetTerms$. Consider
    the following derivation:
    \begin{equation*}
        \begin{prooftree}
            \hypo{\Pi \cbvBKRVTr \Gamma', x : \M \vdash t : \sigma}
            \inferCbvBKRVAbs{\Gamma' \vdash \abs{x}{t} : \mset{\M \typeArrow \sigma}}
            \hypo{\Pi_u \cbvBKRVTr \emptyset \vdash u : \M}
            \inferCbvBKRVEs{\Gamma' \vdash \app[\,]{(\abs{x}{t})}{u} : \sigma}
        \end{prooftree}
    \end{equation*}
    By the \ih on $\Gamma'$, one concludes that
    $\cbvBKRVInhPred{\sigma}$.
    \qedhere
\end{itemize}
\end{proof}

} \deliaLu \giulioLu%

\begin{lemma}
    \label{lem:Cbv_O_|-_T<t>_:_[]_==>_Gam_|-_t_:_sig_and_Gam_and_args(sig)_inh}%
    Let $t \in \cbvSetTerms$ and $\cbvTCtxt$ be a testing context.
    If $\cbvBKRVTr \emptyset \vdash \cbvTCtxt<t> : \sigma$ with
    $\cbvBKRVInhPred{\sigma}$, then there exist $\Gamma$ and $\sigma'$ such that $\cbvBKRVTr \Gamma \vdash t :
    \sigma'$ with $\cbvBKRVInhPred{\Gamma}$ and
    $\cbvBKRVInhPred{\cbvTypeArgs{\sigma'}}$.
\end{lemma}
\stableProof{%
    \begin{proof}
We prove the stronger statement below, to have the right induction
hypothesis:
\begin{center}
    Let $\cbvBKRVTr \Gamma \vdash \cbvTCtxt<t> : \sigma$ with
    $\cbvBKRVInhPred{\Gamma}$ and
    $\cbvBKRVInhPred{\cbvTypeArgs{\sigma}}$, then there exist $\Gamma'$ and $\sigma'$ such that $\cbvBKRVTr \Gamma'
    \vdash t : \sigma'$ with $\cbvBKRVInhPred{\Gamma'}$ and
    $\cbvBKRVInhPred{\cbvTypeArgs{\sigma'}}$.
\end{center}

From that,
\Cref{lem:Cbv_O_|-_T<t>_:_[]_==>_Gam_|-_t_:_sig_and_Gam_and_args(sig)_inh}
follows immediately, because $\cbvBKRVInhPred{\emptyset}$ is vacuously
true.

Let $\Pi \cbvBKRVTr \Gamma \vdash \cbvTCtxt<t> : \sigma$ with
$\cbvBKRVInhPred{\Gamma}$ and $\cbvBKRVInhPred{\cbvTypeArgs{\sigma}}$. We
reason by induction on $\cbvTCtxt$:
\begin{itemize}
\item[\bltI] $\cbvTCtxt = \Hole$: Trivial by hypothesis.

\item[\bltI] $\cbvTCtxt = \app[\,]{(\abs{x}{\cbvTCtxt'})}{u}$: Then
    $\Pi$ has the following form:
    \begin{equation*}
        \begin{prooftree}
            \hypo{\Pi_1 \cbvBKRVTr \Gamma_1, x : \M \vdash \cbvTCtxt'<t> : \sigma}
            \inferBangBKRVAbs{\Gamma_1 \vdash \abs{x}{\cbvTCtxt'<t>} : \mset{\M \typeArrow \sigma}}
            \hypo{\Pi_2 \cbvBKRVTr \Gamma_2 \vdash u : \M}
            \inferBangBKRVApp{\Gamma_1 + \Gamma_2 \vdash \app[\,]{(\abs{x}{\cbvTCtxt'<t>})}{u} : \sigma}
        \end{prooftree}
    \end{equation*}
    with $\Gamma = \Gamma_1 + \Gamma_2$. By hypothesis
    $\cbvBKRVInhPred{\Gamma}$, thus in particular
    $\cbvBKRVInhPred{\Gamma_1}$ and $\cbvBKRVInhPred{\Gamma_2}$. Using
    \Cref{lem:Cbv_building_testing_ctxt_from_type_ctxt} on $\Pi_2$,
    one deduces $\cbvBKRVInhPred{\M}$ and thus
    $\cbvBKRVInhPred{\Gamma_1, x : \M}$. By \ih on $\Pi_1$, one
    obtains $\Pi' \cbvBKRVTr \Gamma' \vdash t : \sigma'$ with
    $\cbvBKRVInhPred{\Gamma'}$ and
    $\cbvBKRVInhPred{\cbvTypeArgs{\sigma'}}$.

\item[\bltI] $\cbvTCtxt = \app{\cbvTCtxt'}{u}$: Then $\Pi$ has the
    following form:
    \begin{equation*}
        \begin{prooftree}
            \hypo{\Pi_1 \cbvBKRVTr \Gamma_1 \vdash \cbvTCtxt'<t> : \mset{\M \typeArrow \sigma}}
            \hypo{\Pi_2 \cbvBKRVTr \Gamma_2 \vdash u : \M}
            \inferBangBKRVApp{\Gamma_1 + \Gamma_2 \vdash \app{\cbvTCtxt'<t>}{u} : \sigma}
        \end{prooftree}
    \end{equation*}
    with $\Gamma = \Gamma_1 + \Gamma_2$. By hypothesis
    $\cbvBKRVInhPred{\Gamma}$, thus in particular
    $\cbvBKRVInhPred{\Gamma_1}$ and $\cbvBKRVInhPred{\Gamma_2}$. 
    By \Cref{lem:Cbv_building_testing_ctxt_from_type_ctxt}, $\cbvBKRVInhPred{\M}$.
    As $\cbvBKRVInhPred{\cbvTypeArgs{\sigma}}$ by hypothesis and $\cbvTypeArgs{\M \typeArrow \sigma} = \{\M\} \cup \cbvTypeArgs{\sigma}$, then  $\cbvBKRVInhPred{\cbvTypeArgs{\M \typeArrow \sigma}}$ 
and
    thus, by the \ih on $\Pi_1$, one obtains $\Pi' \cbvBKRVTr \Gamma'
    \vdash t : \sigma'$ with $\cbvBKRVInhPred{\Gamma'}$ and
    $\cbvBKRVInhPred{\cbvTypeArgs{\sigma'}}$.
    \qedhere
\end{itemize}
\end{proof}

} \deliaLu \giulioLu%

\RecCbvMeaningfulnessCharacterizations*
\label{prf:cbvBKRV_characterizes_meaningfulness}%
\stableProof{
\begin{proof} ~
\begin{enumerate}
\item See \cite{AccattoliPaolini12}
\item ~ \begin{itemize}
    \item[$(1) \Leftrightarrow (2):$] See
        \cite{CarraroGuerrieri14,AccattoliGuerrieri22bis}.

    \item[$(3) \Leftarrow (2):$] Trivial since
        $\cbvBKRVTypeSys$-testable require to be
        $\cbvBKRVTypeSys$-typable.

    \item[$(1) \Rightarrow (3):$] Let $t$ be
        \cbnCalculusSymb-meaningful, then there exists a testing
        context $\cbvTCtxt$ and a value $v \in \cbvSetValues$ such
        that $\cbvTCtxt<t> \cbvArr*_S v$. Without loss of generality,
        we can assume that $v = \abs{x}{u}$ for some $u \in
        \cbvSetTerms$. Notice that $v$ is $\cbvBKRVTypeSys$-testable
        thanks to the following derivation:
        \begin{equation*}
            \begin{prooftree}
                \inferCbvBKRVAbs[0]{\emptyset \vdash \abs{x}{u} : \emptymset}
            \end{prooftree}
        \end{equation*}
        By subject expansion (see
        \cite{BucciarelliKesnerRiosViso20,BucciarelliKesnerRiosViso23}),
        we deduce that $\cbvTCtxt<t>$ is also
        $\cbvBKRVTypeSys$-testable, thus $t$ is
        $\cbvBKRVTypeSys$-testable using
        \Cref{lem:Cbv_O_|-_T<t>_:_[]_==>_Gam_|-_t_:_sig_and_Gam_and_args(sig)_inh}.
    \end{itemize}
\end{enumerate}
\end{proof}
} \deliaLu \giulioLu%

\end{document}